\documentclass[a4paper,onecolumn]{IEEEtran}

\normalsize
\usepackage[font=small,labelfont=small]{caption}

\usepackage{hyperref}
\usepackage{graphicx,wrapfig,lipsum}
\usepackage{amsmath,latexsym,amsfonts,amssymb}
\usepackage{placeins}
\usepackage{array}
\usepackage{pifont}
\usepackage{amssymb,eucal,times,exscale,epsfig}
\usepackage{algpseudocode}
\usepackage{algorithm}
\usepackage{amsthm}
\usepackage{cite}
\usepackage{calc}
\usepackage{color}
\usepackage{epsfig}
\usepackage{setspace}
\usepackage{epstopdf}

\newcolumntype{M}[1]{>{\centering\arraybackslash}m{#1}}
	
\begin{document}

\newcommand{\bfl}{\begin{flushleft}}
\newcommand{\efl}{\end{flushleft}}
\newcommand{\bit}{\begin{itemize}}
\newcommand{\eit}{\end{itemize}}
\newcommand{\ben}{\begin{enumerate}}
\newcommand{\een}{\end{enumerate}}
\newcommand{\bc}{\begin{center}}
\newcommand{\ec}{\end{center}}
\newcommand{\bfr}{\begin{flushright}}

\newtheorem{definition}{Definition}
\newtheorem{example}{Example}
\newtheorem{defn}{Definition}
\newtheorem{thm}{Theorem}[section]
\newtheorem{cor}[thm]{Corollary}
\newtheorem{prop}{Proposition}
\newtheorem{lem}[thm]{Lemma}
\newtheorem{conj}[thm]{Conjecture}
\newtheorem{constr}[thm]{Construction}
\newtheorem{note}{Remark}
\newtheorem{rem}[thm]{Remark}
\newtheorem{eg}{Example}
\newtheorem{property}{Property}
\newtheorem{open}{Open Problems}
\newcommand{\bcor}{\begin{cor}}
	\newcommand{\ecor}{\end{cor}}
\newcommand{\beq}{\begin{equation}}
\newcommand{\eeq}{\end{equation}}
\newcommand{\beqn}{\begin{equation*}}
\newcommand{\eeqn}{\end{equation*}}
\newcommand{\bea}{\begin{eqnarray}}
\newcommand{\eea}{\end{eqnarray}}
\newcommand{\bean}{\begin{eqnarray*}}
	\newcommand{\eean}{\end{eqnarray*}}
\newcommand{\bdefn}{\begin{defn}}
	\newcommand{\edefn}{\end{defn}}
\newcommand{\bnote}{\begin{note}}
	\newcommand{\enote}{\end{note}}
\newcommand{\bprop}{\begin{prop}}
	\newcommand{\eprop}{\end{prop}}
\newcommand{\blem}{\begin{lem}}
	\newcommand{\elem}{\end{lem}}
\newcommand{\bthm}{\begin{thm}}
	\newcommand{\ethm}{\end{thm}}
\newcommand{\bconj}{\begin{conj}}
	\newcommand{\econj}{\end{conj}}
\newcommand{\bconstr}{\begin{constr}}
	\newcommand{\econstr}{\end{constr}}
\newcommand{\bpf}{\begin{proof}}
	\newcommand{\epf}{\end{proof}}
\newcommand\inlineeqno{\stepcounter{equation}\ (\theequation)}

\newcommand{\fq}{\mbox{$\mathbb{F}_q$}}
\newcommand{\fqalpha}{\mbox{$\mathbb{F}_q^{\alpha}$}}
\newcommand{\fqbeta}{\mbox{$\mathbb{F}_q^{\beta}$}}
\newcommand{\fqB}{\mbox{$\mathbb{F}_q^B$}}
\newcommand{\fqstar}{\mbox{$\mathbb{F}_q^{*}$}}
\newcommand{\calc}{\mbox{${\cal C} $}}
\newcommand{\uc}{\mbox{$\underline{c}$}}
\newcommand{\uu}{\mbox{$\underline{u}$}}
\newcommand{\ug}{\mbox{$\underline{g}$}}
\newcommand{\ua}{\mbox{$\underline{a}$}}
\newcommand{\uz}{\mbox{$\underline{z}$}}
\newcommand{\ue}{\mbox{$\underline{e}$}}
\newcommand{\Zs}{\mbox{$\mathbb{Z}_s$}}
\newcommand{\Zt}{\mbox{$\mathbb{Z}_t$}}
\newcommand{\Zst}{\mbox{$\mathbb{Z}_s^t$}}
\newcommand{\Zsn}{\mbox{$\mathbb{Z}_s^n$}}
\newcommand{\calu}{\mbox{${\cal U}$}}
\newcommand{\cale}{\mbox{${\cal E}$}}
\newcommand{\gmds}{\mbox{$G_{\text{\tiny MDS}}$}}
\newcommand{\gpyr}{\mbox{$G_{\text{\tiny PYR}}$}}
\newcommand{\cmds}{\mbox{$C_{\text{\tiny MDS}}$}}
\newcommand{\cpyr}{\mbox{$C_{\text{\tiny PYR}}$}}

\newcommand{\pvk}{\color{blue} }
\newcommand{\pvknew}{\color{cyan} }

\title{Erasure Coding for Distributed Storage: An Overview}
\author{
	\IEEEauthorblockN{ S. B. Balaji, M. Nikhil Krishnan, Myna Vajha, Vinayak Ramkumar, }\\
	\IEEEauthorblockN{ Birenjith Sasidharan and P. Vijay Kumar}\\ \ \\
	\IEEEauthorblockA{Department of Electrical Communication Engineering,\\ Indian Institute of Science, Bangalore.\\
	}
	\thanks{ This survey article will appear in Science China Information Sciences (SCIS) journal.
		
	M. N. Krishnan and M. Vajha would like to acknowledge the support of Visvesvaraya PhD Scheme for Electronics \& IT awarded by DEITY, Govt. of India. B. Sasidharan presently works as Assistant Professor in the department of Electronics and Communication Engineering in Government Engineering College Barton Hill at Thiruvananthapuram, Kerala, India. P. V. Kumar is also a Visiting Professor at the University of Southern California.  His research is supported in part by the National Science Foundation under Grant No. 1421848 and in part by the joint UGC-ISF research program.}
}
\maketitle

\begin{abstract}
	In a distributed storage system, code symbols are dispersed across space in nodes or storage units as opposed to time. In settings such as that of a large data center, an important consideration is the efficient repair of a failed node.  Efficient repair calls for erasure codes that in the face of node failure, are efficient in terms of minimizing the amount of repair data transferred over the network, the amount of data accessed at a helper node as well as the number of helper nodes contacted. Coding theory has evolved to handle these challenges by introducing two new classes of erasure codes, namely regenerating codes and locally recoverable codes as well as by coming up with novel ways to repair the ubiquitous Reed-Solomon code. 
	This survey provides an overview of the efforts in this direction that have taken place over the past decade. 
\end{abstract}
%

\section{Introduction}
This survey article deals with the use of erasure coding for the reliable and efficient storage of large amounts of data in settings such as that of a data center.  The amount of data stored in a single data center can run into tens or hundreds of petabytes.   Reliability of data storage is ensured in part by introducing redundancy in some form, ranging from simple replication to the use of more sophisticated erasure-coding schemes such as Reed-Solomon codes.  Minimizing the storage overhead that comes with ensuring reliability is a key consideration in the choice of erasure-coding scheme.  More recently a second problem has surfaced, namely, that of node repair.

In \cite{RashmiShahGuKuang}, \cite{SathiaAstPap_Xorbas} the authors study the Facebook warehouse cluster and analyze the frequency of node failures as well as the resultant network traffic relating to node repair. It was observed in \cite{RashmiShahGuKuang} that a median of $50$ nodes are unavailable per day and that a median of $180$TB of cross-rack traffic is generated as a result of node unavailability. It was also reported that $98.08\%$ of the cases have exactly one block missing in a stripe. The erasure code that was deployed in this instance was an $[n=14,k=10]$ Reed Solomon (RS) code.   Here $n$ denotes the block length of the code and $k$ the dimension. The conventional repair of an $[n,k]$ RS code is inefficient in that the repair of a single node, calls for contacting $k$ other (helper) nodes and downloading $k$ times the amount of data stored in the failed node, which is clearly inefficient.  Thus there is significant practical interest in the design of erasure-coding techniques that offer both low overhead and which can also be repaired efficiently.

Coding theorists have responded to this need by coming up with two new classes of codes, namely ReGenerating (RG) and Locally Recoverable (LR) codes.  The focus in a RG code is on minimizing the amount of data download needed to repair a failed node, termed the {\em repair bandwidth} while LR codes seek to minimize the number of helper nodes contacted for node repair, termed the {\em repair degree}. In a different direction, coding theorists have also re-examined the problem of node repair in RS codes and have come up with new and more efficient repair techniques.  This survey provides an overview of these recent developments.   An outline of the survey itself appears in Fig.~\ref{fig:chapter_outline}.

RG codes are discussed in Section \ref{sec:regenerating_codes}.  The two principal classes of RG codes, namely Minimum Bandwidth Regenerating (MBR) and Minimum Storage Regeneration (MSR) appear in the two sections that follow.   These two classes of codes are at the two extreme ends of a tradeoff known as the storage-repair bandwidth (S-RB) tradeoff. A discussion on codes that correspond to the interior points of this tradeoff appears in Section~\ref{sec:on-srb}.  The theory of regenerating codes has been extended in several directions and these are explored in Section~\ref{sec:variations}. 
Section~\ref{sec:lrc} examines LR codes.  There have been several approaches at extending the theory of LR codes to handle multiple erasures and these are dealt with in Section~\ref{sec:multiple_erasures}.  A class of codes known as Locally ReGenerating (LRG) codes that offer both low repair bandwidth and low repair degree within a single erasure code is discussed in Section~\ref{sec:lrgc}.   This is followed by Section~\ref{sec:rs_repair} that discusses recent advances in the repair of Reed-Solomon codes.   A brief description of a different approach based on capacity considerations and leading to the development of a {\em liquid cloud storage} system appears in Section~\ref{sec:liquid}. The final section, discusses practical evaluations and implementations.

{\em Disclaimer:} This survey is presented from the perspective of the authors and is biased in this respect.  Given the explosion of research activity in this area, the survey also does not claim to be comprehensive and we offer our apologies to the authors whose work has inadvertently or for lack of space, not been appropriately cited. We direct the interested reader to some of the excellent surveys of codes on distributed storage contained in the literature including \cite{DimRamWuSuhSurvey}, \cite{DattaOggSurvey}, \cite{LiLiSurvey} and \cite{liuOggSurvey}.

\begin{figure}[h!]
\centering
\includegraphics[width=6.2in]{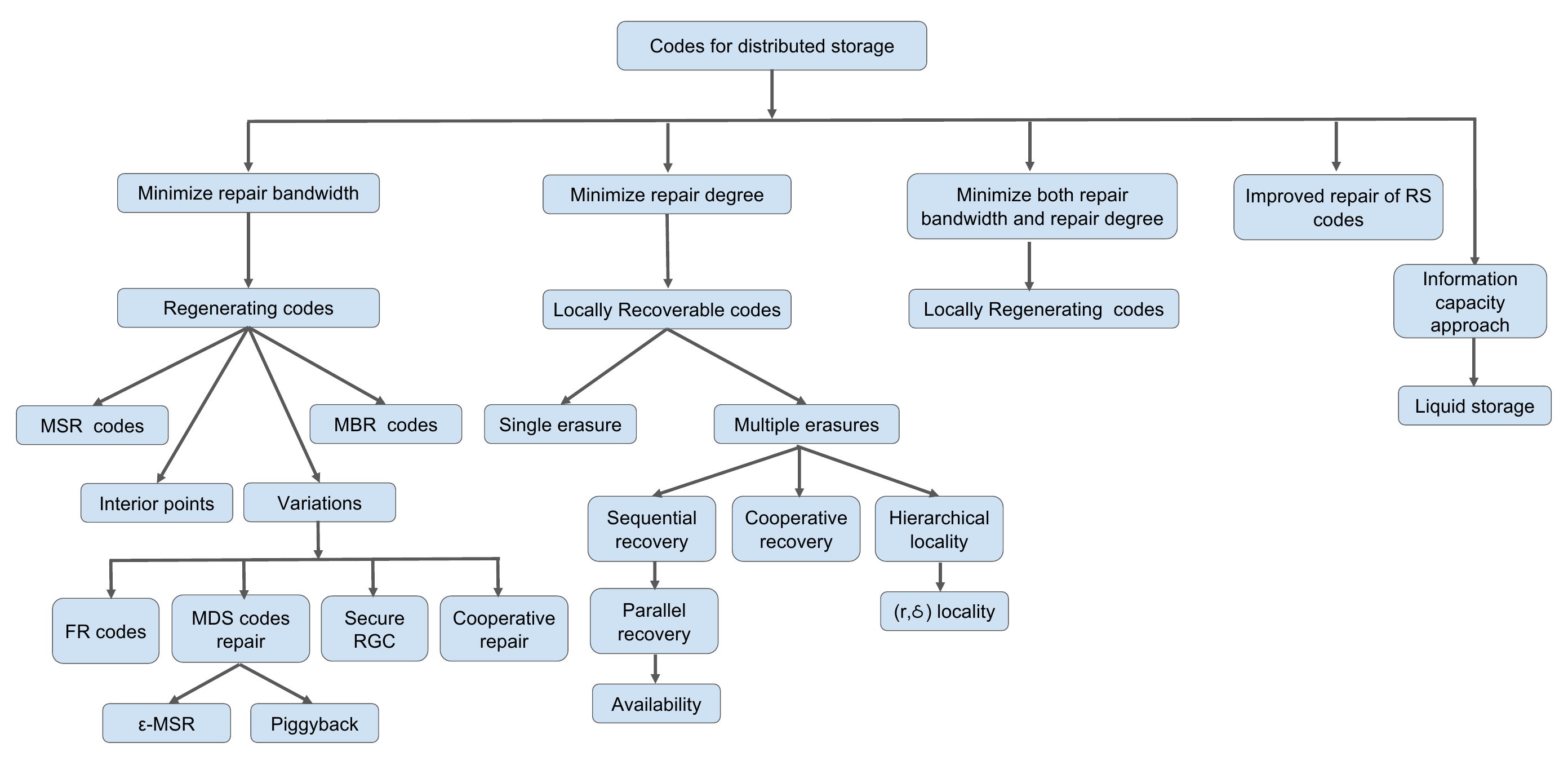} 
\caption{An overview of the different classes of codes for distributed storage discussed in this survey article.} \label{fig:chapter_outline} 
\end{figure}

\section{Regenerating Codes{\label{sec:regenerating_codes}}}
\begin{figure}[ht!]
		\begin{center}
			Parameters: $(\ (n,k,d), \ (\alpha, \beta), \ B, \ \mathbb{F}_q \ )$
		\end{center}
		\begin{center}
			\begin{minipage}{3.0in}
				\begin{center}
					\includegraphics[trim= 0in  0in 0in 0in, width=1.8in]{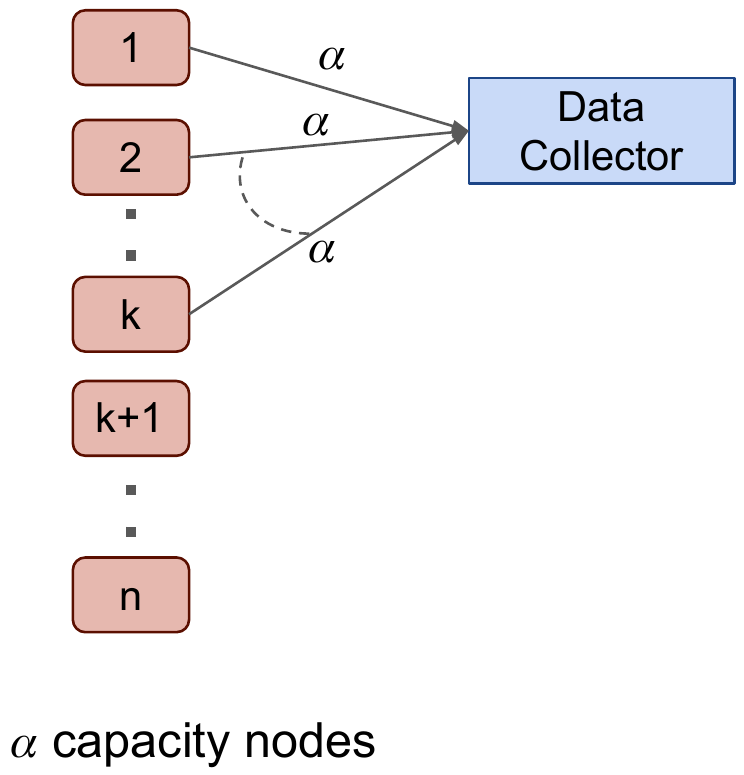} 
				\end{center}
				
			\end{minipage} 
			\begin{minipage}{3.0in}
				\begin{center}
					\includegraphics[trim= 0in  0in 0in 0in, width=1.6in]{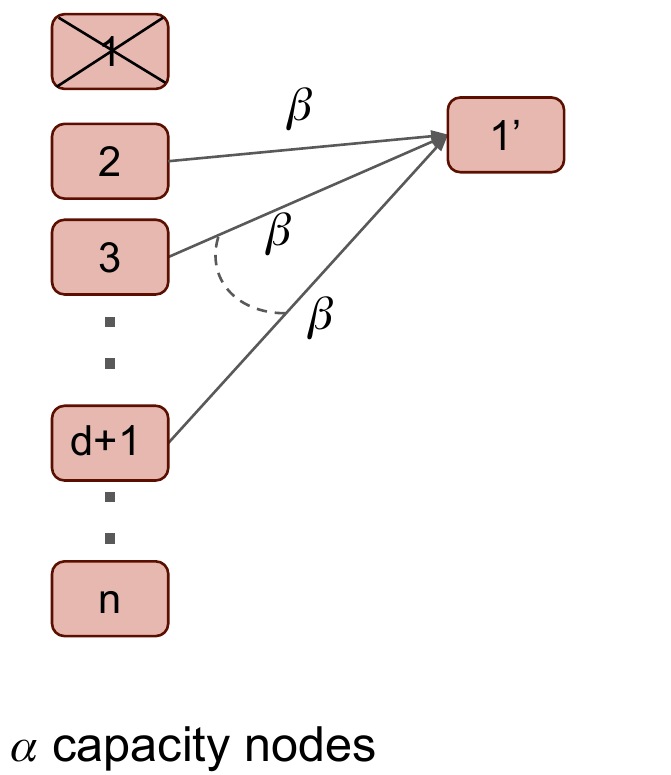} 
				\end{center}
			\end{minipage}
		\end{center}
		\caption{An illustration of the data collection and node repair properties of a regenerating code.} \label{fig:rgc}
\end{figure}

\begin{defn}[\cite{DimGodWuWaiRam}]
	Let \fq\ denote a finite field of size $q$. Then a regenerating (RG) code  \calc\ over \fq\ having integer parameter set $((n,k,d), (\alpha, \beta), B)$ where $1\le k \le n-1$, $k \le d \le n-1$, $\beta \le \alpha$, maps a file $\uu \in \fqB$ on to a collection $\{\uc_i\}_{i=1}^n$ of $n$ $\alpha$-tuples over $\mathbb{F}_q$ using an encoding map
	\bean
	E(\uu) = [\uc_1^T, \uc_2^T, \cdots, \uc_n^T]^T 
	\eean
	with the $\alpha$ components of $\uc_i$ stored on the $i$-th node in such a way that the following two properties (see Fig.~\ref{fig:rgc}) are satisfied:
	{\em Data Collection: }The message \uu\ can be uniquely recovered from the contents $\{c_{i_j}\}_{j=1}^k$ of any $k$ nodes.\\
	{\em Node Repair: } If the $f$-th node storing $\uc_f$ fails, then a replacement node can 
	\ben
	\item contact any subset $D \subseteq [n] \setminus \{f\}$ of the remaining $(n-1)$ nodes of size $|D| = d$,
	\item map the $\alpha$ contents $\uc_{h}$ of each helper node $h \in D$ on to a collection of $\beta$ repair symbols $\ua_{h,f}^D \in \fqbeta$,
	\item pool together the $d\beta$ repair symbols thus computed to use them to create a replacement vector $\hat{\uc}_f \in \fqalpha$ whose $\alpha$ components are stored in the replacement node, in such a way that the contents of the resultant nodes, with the replacement node replacing the failed node, once again forms a regenerating code. \\
	\een
	\edefn
	
	A regenerating code is said to be {\em exact-repair} (ER) regenerating code if the contents of the replacement node are exactly same as that of the failed node, ie., $\hat{\uc}_f = \uc_f$. Else the code is said to be {\em functional-repair} (FR) regenerating code. A regenerating code is said to be linear if
	\ben
	\item $E(\uu_1 + \theta \uu_2) = E(\uu_1) + \theta E(\uu_2)$, $\uu_1, \uu_2 \in \fqB, \theta \in \fq$ and
	\item the map mapping the contents $\uc_h$ of the $h$-th helper node on to the corresponding $\beta$ repair symbols $\ua_{h,f}^D$ is linear over \fq.
	\een
	
	Thus a regenerating code is a code over a vector alphabet \fqalpha\ and the quantity $\alpha$ is termed the {\em sub-packetization level} of the regenerating code. The total number $d\beta$ of \fq\ symbols to be transferred for repair of failure node is called the {\em repair bandwidth} of the regenerating code. The rate of the regenerating code is given by $R = \frac{B}{n\alpha}$. Its reciprocal $\frac{n\alpha}{B}$ is the {\em storage overhead}.

	\subsection{Cut-Set Bound}
	\begin{figure}[ht!]
		\begin{center}
		\includegraphics[trim= 0in  0in 0in 0in, width=4.5in]{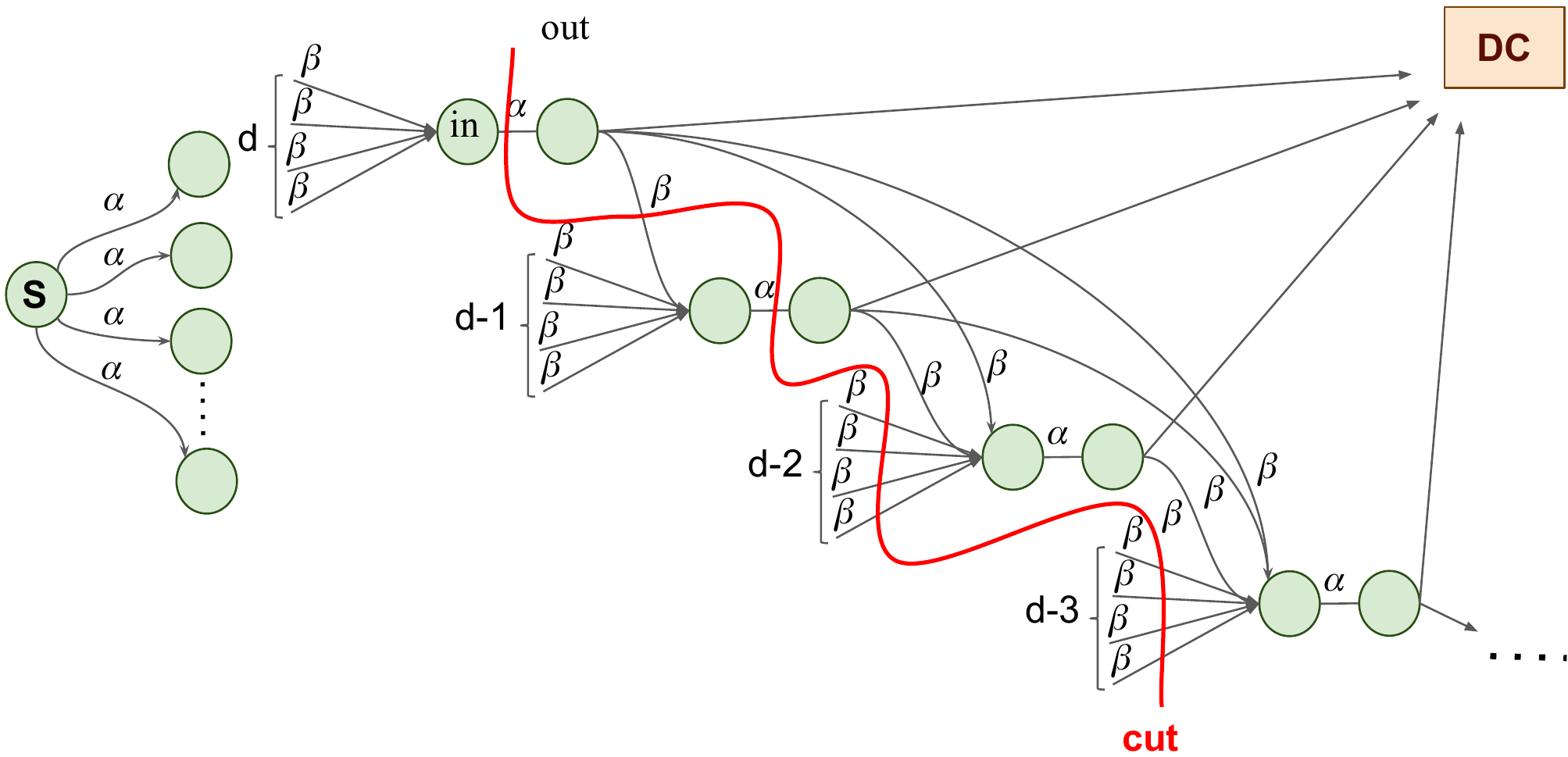} 
		\end{center}
		\caption{The graph behind the cut-set file size bound.\label{fig:cutsetbound}}	
	\end{figure}

	Let us assume that \calc\ is a functional-repair regenerating code having parameter set: $((n,k,d),$ $(\alpha, \beta), B)$. Since an exact-repair regenerating code is also a functional-repair code, this subsumes the case when \calc\ is an exact-repair regenerating code. Over time, nodes will undergo failures and every failed node will be replaced by a replacement node. Let us assume to begin with, that we are only interested in the behavior of the regenerating code over a finite-but-large number $N >> n$ of node repairs. For simplicity, we assume that repair is carried out instantaneously. Then at any given time instant $t$, there are $n$ functioning nodes whose contents taken together comprise a regenerating code. At this time instant, a data collector could connect to $k$ nodes, download all of their contents and decode to recover underlying message vector $\uu$. Thus in all, there are at most $N{n \choose k}$ distinct data collectors which are distinguished based on the particular set of $k$ nodes to which the data collector connects. 
	
	Next, we create a source node that possesses the $B$ message symbols $\{u_i\}_{i=1}^B$, and draw edges connecting the source to the initial set of $n$ nodes. We also draw edges between the $d$ helper nodes that assist a replacement node and the replacement node itself as well as edges connecting each data collector with the corresponding set of $k$ nodes from which the data collector downloads data. All edges are directed in the direction of information flow. We associate a capacity $\beta$ with edges emanating from a helper node to a replacement node and an $\infty$ capacity with all other edges. Each node can only store $\alpha$ symbols over $\fq$. We take this constraint into account using a standard graph-theory construct, in which a node is replaced by $2$ nodes separated by a directed edge (leading towards a data collector) of capacity $\alpha$. We have in this way, arrived at a graph (see Fig.\ref{fig:cutsetbound}) in which there is one source $S$ and at most $N{n \choose k}$ sinks $\{T_i\}$. 
	
	Each sink $T_i$ would like to be able to reconstruct all the $B$ source symbols $\{u_i\}$ from the symbols it receives. This is precisely the multicast setting of network coding. A principal result in network coding tells us that in a multicast setting,  one can transmit messages along the edges of the graph in such a way that each sink $T_i$ is able to reconstruct the source data, provided that the minimum capacity of a cut separating $S$ from $T_i$ is $\ge B$.   A cut separating $S$ from $T_i$ is simply a partition of the nodes of the network into $2$ sets: $A_i$ containing $S$ and $A_i^c$ containing $T_i$. The capacity of the cut is the sum of capacities of the edges leading from a node in $A_i$ to a node in $A_i^c$. A careful examination of the graph will reveal that the minimum capacity $Q$ of a cut separating a sink $T_i$ from source $S$ is given by $
	Q = \sum\limits_{i=0}^{k-1} \min \{\alpha, (d-i)\beta\}
	$
	(Fig.~\ref{fig:cutsetbound} shows an example cut separating source from sink). This leads to the following upper bound on file size \cite{DimGodWuWaiRam}:
	\bea
	\label{eq:cutsetbound}
	B &\le& \sum\limits_{i=0}^{k-1} \min \{\alpha, (d-i)\beta\}.
	\eea
	Network coding also tells us that when only a finite number of regenerations take place, this bound is achievable and furthermore achievable using linear network coding, i.e., using only linear operations at each node in the network when the size $q$ of the finite field \fq\ is sufficiently large. 
	In a subsequent result \cite{Wu}, Wu established using the specific structure of the graph, that even in the case when the number of sinks is infinite, the upper bound in \eqref{eq:cutsetbound} continues to be achievable using linear network coding.
	
	In summary, by drawing upon network coding, we have been able to characterize the maximum file size of a regenerating code given parameters $\{k, d, \alpha, \beta\}$ for the case of functional repair when there is constraint placed on the size $q$ of the finite field $\fq$. Note interestingly, that the upper bound on file size is independent of $n$. Quite possibly, the role played by $n$ is that of determining the smallest value of field size $q$ for which a linear network code can be found having file size $B$ satisfying \eqref{eq:cutsetbound}.
	A functional regenerating code having parameters: $((n, \ k, \ d), \ (\alpha, \ \beta), \ B)$ is said to be optimal provided (a) the file size $B$ achieves the bound in \eqref{eq:cutsetbound} with equality and (b) reducing either $\alpha$ or $\beta$ will cause the bound in \eqref{eq:cutsetbound} to be violated.
	
	\subsection{Storage-Repair Bandwidth Tradeoff}
	
	We have thus far, specified code parameters $(k,d) (\alpha, \beta)$ and asked what is the largest possible value of file size $B$. If however, we fix parameters $(n,k,d,B)$ and ask instead what are the smallest values of $(\alpha, \beta)$ for which one can hope to achieve \eqref{eq:cutsetbound}, it turns out, as might be evident from the form of the summands on the RHS of \eqref{eq:cutsetbound}, that there are several pairs $(\alpha, \beta)$ for which equality holds in \eqref{eq:cutsetbound}.  In other words, there are different flavors of optimality. 
		\begin{figure}[ht!]
			\bc
			\includegraphics[width=3.5in]{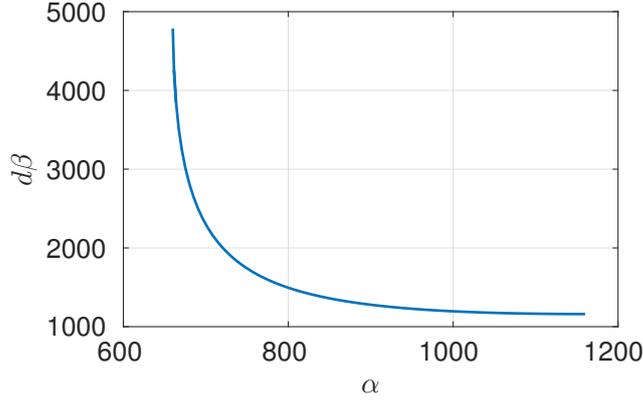}
			\caption{Storage-repair bandwidth tradeoff. Here, $(n = 60, k = 51, d=58, B = 33660)$.\label{fig:srb_tradeoff}}
			\ec
		\end{figure}
	For a given file size $B$, the storage overhead and normalized repair bandwidth are given respectively by $\frac{n\alpha}{B}$ and $\frac{d\beta}{B}$. Thus $\alpha$ reflects the amount of storage overhead while $\beta$ determines the normalized repair bandwidth. The several pairs $(\alpha, \beta)$ for which equality holds in \eqref{eq:cutsetbound}, represent a tradeoff between storage overhead on the one hand and normalized repair bandwidth on the other as can be seen from the example plot in Fig.~\ref{fig:srb_tradeoff}.  Clearly, the smallest value of $\alpha$ for which the equality can hold in \eqref{eq:cutsetbound} is given by $\alpha = \frac{B}{k}$. Given $\alpha = \frac{B}{k}$, the smallest permissible value of $\beta$ is given by $\beta = \frac{\alpha}{d-k+1}$. This represents the minimum storage regeneration point and codes achieving \eqref{eq:cutsetbound} with $\alpha = \frac{B}{k}$ and $\beta = \frac{\alpha}{d-k+1}$ are known as minimum storage regenerating (MSR) codes.  At the other end of the tradeoff, we have the minimum bandwidth regenerating (MBR) code whose associated $(\alpha, \beta)$ values are given by $\beta = \frac{B}{dk-{k \choose 2}}$, $\alpha = d\beta$.
	\begin{note}
		Since a regenerating code can tolerate $(n-k)$ erasures by the data collection property, it follows that the minimum Hamming weight $d_{\min}$ of a regenerating code must satisfy $d_{\min} \ge (n-k+1)$.  By the Singleton bound, the largest size $M$ of a code of block length $n$ and minimum distance $d_{\min}$ is given by $M \le Q^{n - d_{\min} + 1} \le Q^k$, where $Q$ is the size of alphabet of the code. Since $Q = q^{\alpha}$ in the case of regenerating code, it follows that the size $M$ of a regenerating code must satisfy $M \le q^{k\alpha}$, or equivalently $q^B \le q^{k\alpha}$, i.e., $B \le k\alpha$. But $B=k\alpha$ in the case of an MSR code and it follows that an MSR code is an MDS code over a vector alphabet. Such codes also go by the name MDS array code.
	\end{note}
	
	From a practical perspective, exact-repair regenerating codes are easier to implement as the contents of the $n$ nodes in operation do not change with time. Partly for this reason and partly for reasons of tractability, with few exceptions, most constructions of regenerating codes belong to the class of exact-repair regenerating codes. Examples of functional-repair regenerating code include the $d=(k+1)$ construction in \cite{ShaRasKumRam_ia} as well as the construction in \cite{HuChenLee_NCcloud}.
	
	Early constructions of regenerating codes focused on the two extreme points of the storage-repair bandwidth (S-RB) tradeoff, namely the MSR and MBR points. The various constructions of MBR and MSR codes are described in Sections \ref{sec:mbr}, \ref{sec:msr}. Not surprisingly, given the vast amount of data stored, the storage industry places a premium on low storage overhead. In this connection, we note that the maximum rate of an MBR code is given by:
	\bean
	R_{\text{MBR}} = \frac{B}{n\alpha} = \frac{(dk - {k \choose 2})\beta}{n d\beta} = \frac{dk-{k \choose 2}}{nd},
	\eean
	which can be shown to be upper bounded by $R_{\text{MBR}} \le \frac{1}{2}$ and is achieved when $k=d=(n-1)$.  In the case of MSR codes, there is no such limitation and MSR codes can have rates approaching $1$. 
	
	An RG code is said to be a a {\em Help-By-Transfer} (HBT) RG code if repair of a failed node can be accomplished without incurring any computation at a helper node.  If no computation is required at either helper node or at the replacement node, then the code is termed a {\em Repair-by-Transfer} (RBT) RG code. Clearly, an RBT RG code is also an HBT RG code.

\section{MBR Codes\label{sec:mbr}}
\begin{rem}
	If the $B$ message symbols are drawn randomly with uniform distribution from $\mathbb{F}_q^B$, it can be shown that in any regenerating code achieving the cut-set bound, the contents of each node correspond to a random variable that is uniform over $\mathbb{F}_q^{\alpha}$.  In an MBR code, repair is accomplished by downloading a total of just $\alpha$ symbols which clearly, is the minimum possible.  \end{rem}

\begin{rem}
	Let $\mathcal{C}$ be an MBR code. If $\mathcal{C}$ has the RBT property, it trivially follows that all scalar code-symbols of $\mathcal{C}$ are replicated at least twice. In \cite{KrishKumMBRRep16}, it is shown that for an MBR code it is not possible to have even a single scalar code-symbol replicated more than twice.  Thus the RBT property implies that the collection of $n\alpha$ scalar code-symbols associated with a codeword represent a set of $\frac{n\alpha}{2}$ distinct code symbols, each repeated twice.  The converse is not true in general. However when $d=(n-1)$, it can be shown that the two properties are equivalent. 
\end{rem}
\begin{rem}
	In \cite{ShahHBTNonexistence13}, it is shown that for $d<(n-1)$, it is not possible to construct an MBR code that has the HBT property.  
\end{rem}

\subsection{Polygonal MBR Codes}

\begin{figure}[ht!]
	\bc
	\includegraphics[width=2.2in]{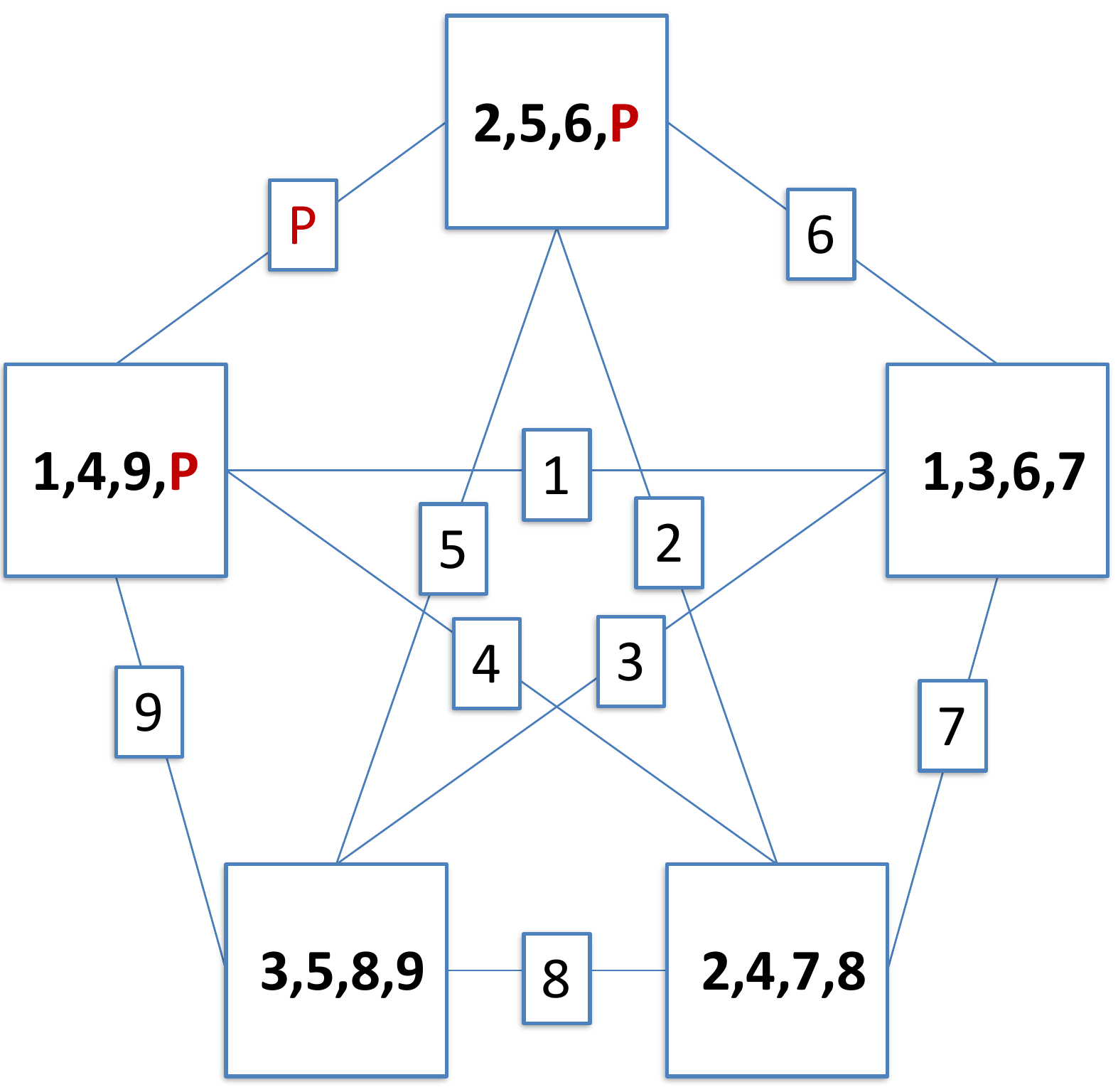}
	\caption{An example RBT MBR code for the parameters $n=5, k=3, d=4$. Here file size is $9$.  }  
	\label{fig:pentagon_code}
	\ec
\end{figure}
In the following, we describe with the help of an example, one of the first explicit families of MBR codes \cite{RasShaKumRam_allerton09}. We term these codes as {\it polygonal MBR codes}. The construction holds for parameters $k\leq d=n-1,\beta=1$ and the constructed MBR codes possess the RBT property.
\begin{eg}
	Consider the parameters $n=5,k=3,d=4$ and $\beta=1$. Thus $B=kd\beta-{k\choose 2}\beta=9$. First construct a complete graph with $n=5$ vertices and $N={5\choose 2}=10$ edges. The nine message symbols are then encoded using a $[10,9]$ MDS code to produce ten code-symbols. Each code-symbol is then uniquely assigned an edge. Each node of the MBR code stores the code-symbols corresponding to the edges incident on that node (see Fig. \ref{fig:pentagon_code}). The data collection property follows as any collection of $k=3$ nodes yields nine distinct (MDS) code-symbols. If a node fails, the replacement node can download from each of the remaining four nodes, the code-symbol corresponding to the edge it shares with the failed node. Hence repair is accomplished by merely transferring the data without any computation (RBT).
\end{eg}

\begin{rem}
	For the general construction, in order to construct an $[n,k,d=n-1],\beta=1$ MBR code, one first forms the complete graph on $n$ vertices. Each edge is then mapped to a code-symbol of an $[N,B]$ MDS code, where $N={n\choose 2}$ and $B$ is the file size parameter. An $O(n^2)$ field-size requirement is thus imposed by the underlying scalar MDS code.
\end{rem}
\subsection{Product-Matrix (PM) MBR codes}
A second, general construction for MBR codes is the PM construction~\cite{RasShaKum_pm11} which derives its name from the fact that the contents of $n$ nodes can be expressed in the form of a product of two matrices.  The two matrices are respectively an encoding matrix and a second, message matrix containing the message symbols.  This construction yields MBR codes for all feasible parameters $k\leq d\leq n-1$, $\beta=1$, with an $O(n)$ field-size requirement.   The $(n\times d)$ encoding matrix $\psi$ is of the form: $\psi\  = \left[ \begin{array}{cc}
\phi & \Delta  
\end{array} \right]$, where $\phi$, $\Delta$ are $(n \times k)$, $(n \times (d-k))$ matrices respectively. Let the $i$-th row of $\psi$ be denoted by $\psi_i^T$. The sub-matrices $\phi$ and $\Delta$ are here chosen such that any $d$ rows of $\psi$ and any $k$ rows of $\phi$ are linearly independent. The $(d\times d)$ symmetric message matrix $M$ is derived from the $B=kd-{k\choose 2}$ message symbols as follows:
\bean
M\  = \left[ \begin{array}{cc}
	S & V \\ 
	V^T & 0
\end{array} \right], \ \ \text{where $S$ is a symmetric $(k\times k)$ matrix and $V$ a $(k \times (d-k))$ matrix.}
\eean
The $i$-th node, under the PM-MBR construction, stores the matrix product $\psi_i^TM$.  The repair data passed on by helper node $j$ to replacement node $i$ is given by $\psi_j^TM\psi_i$.

\subsection{Other Work}	

In \cite{LinChungNovelRBTMBR14}, the authors introduce a family of RBT MBR codes for $d=n-1$, that are constructed based on a congruent transformation applied to a skew-symmetric matrix of message symbols. In comparison with the $O(n^2)$ field requirement of polygonal MBR codes, in this construction, a field-size of $O(n)$ suffices.  
In \cite{HanPaiZhengVarshUpdateEfficientPM15}, the authors stay within the PM framework, but provide a different set of encoding matrices for MSR and MBR codes that have least-possible update complexity within the PM framework. The authors of \cite{HanPaiZhengVarshUpdateEfficientPM15} also analyze the codes for their ability to correct errors and provide corresponding decoding algorithms.   The paper \cite{ShahHBTNonexistence13} proves the non-existence of HBT MBR codes with $d<(n-1)$.  The paper also provides PM-based constructions for two relaxations, namely (i) any failed node which is a part of a collection of systematic nodes can be recovered in HBT fashion from any $d$ other nodes  and (ii) for every failed node, there exists a corresponding set of $d$ helper nodes which permit HBT repair. The paper \cite{KrishKumMBRRep16} provides binary MBR constructions for the parameters $(k=d=n-2)$, $(k+1=d=n-2)$ and studies the existence of MBR codes with inherent double replication, for all parameters. In \cite{Raviv2AsymptOptAnyFieldMSRMBR17}, the authors provide regenerating-code constructions that asymptotically achieve the MSR or MBR point as $k$ increases and these codes can be constructed over any field, provided the file size is large enough. In \cite{MahdKhisMohajBWAdaptiveMBR17}, the authors introduce some extensions to the classical MBR framework by permitting the presence of a certain number of error-prone nodes during repair/reconstruction and by introducing flexibility in choosing the parameter $d$ during node repair.	

\begin{open}
	Determine the smallest possible field size $q$ of an MBR code for given $\{(n,k,d), (\alpha,\beta)\}$. 
\end{open}

\section{MSR Codes\label{sec:msr}}
Among the class of RG codes, MSR codes have received the greatest attention, and the reasons include: the fact that (a) the storage overhead of an MSR code can be made as small as desired, (b)  MSR codes are MDS codes and (c) MSR codes have been challenging to construct.
\subsection{Introduction\label{sec:intro_msr}}
As noted previously, an MSR code with parameters $(n,k,d, \alpha)$ has file size $B = k\alpha$ and $\beta = \frac{\alpha}{d-k+1}$. Although MSR codes are vector MDS codes that have optimum repair-bandwidth of $d\beta$ for the repair of any node among the $n$ nodes, there are papers in the literature that refer to a code as an MSR code even if optimal repair holds only for the systematic nodes. In the current paper, we refer to such codes as {\em systematic} MSR codes. While only $\beta$ symbols are sent by each of the $d$ helper nodes, the number of symbols accessed by the helper node in order to generate these $\beta$ symbols could be $> \beta$. The class of MSR codes that access at each helper node, only as many symbols as are transferred, are termed {\em optimal-access} MSR codes. MSR codes that alter a minimum number of parity symbols while updating a single, systematic symbol, are called {\em update-optimal} MSR codes.

There are several exact-repair (ER) MSR constructions available in the literature. In \cite{ShaRasKumRam_ia}, Shah et al. show that {\em interference alignment} (IA) is necessarily present in every exact-repair MSR code, and use IA techniques to construct systematic MSR codes, known as MISER codes, for $d=n-1 \ge 2k-1$. The IA condition in the context of MSR codes (observed earlier in \cite{WuDim}) demands that the interference components in the data passed by helper nodes must be aligned so that they can be cancelled at the replacement node by data received from the systematic helper nodes. In \cite{SuhRam}, Suh et al. build on \cite{ShaRasKumRam_ia} to construct MSR codes for $d \ge 2k-1$ with optimal repair bandwidth for all nodes, under the condition that the helper-node set necessarily includes systematic nodes. In \cite{RasShaKum_pm}, the well-known Product Matrix (PM) framework is introduced to provide MSR constructions for $d \geq 2k-1$, thereby settling the problem of MSR code construction in the low-rate regime, $k/n \leq 0.5$. While the method adopted in \cite{RasShaKum_pm} to provide a construction for $d > 2k-1$ is to suitably shorten a code for $d=2k-1$, an extension of the PM framework that yields constructions for any $d \geq 2k-1$ in a single step is provided in \cite{LinChungHanAl}. Apart from a few notable constructions such as the Hadamard-design-based code \cite{PapDimCad} for $(k+2,k)$ and its generalization for $(n-k) >2$ for systematic node-repair, the problem of high-rate constructions (i.e., $k/n \geq 0.5$) for all-node repair remained open. The first major result in this direction, is due to Cadambe et al. \cite{CadJafMalRamSuh} where the authors apply the notion of {\em symbol extension in interference alignment} where multiple symbols are grouped together to form a single vector symbol, to jointly achieve interference alignment.  The symbol-extension viewpoint is then used to show that ER MSR codes exist for all $(n,k,d)$, as $B$ goes to infinity. The second major development was the zigzag code construction \cite{TamWanBru,WangTamoBruck}, the first non-asymptotic high-rate MSR code construction with $d=(n-1)$ permitting rates as close as $1$ as desired, with additional desirable properties such as optimal access and optimal update.  Zigzag codes however, require a sub-packetization level ($\alpha$) that grows  exponentially with $k$ and a very large finite field size, while the earlier PM codes for the low-rate regime, have $\alpha = (k+1)$ and field-size that is linear in $n$. In a subsequent work~\cite{CadHuaLiMeh}, the authors present a systematic MSR construction having $\alpha = \frac{k^2}{4}$ and rate $R = 2/3$.  A second systematic MSR code with $\alpha = r^{\frac{k}{r+1}}$ is presented in \cite{WanTamBru_long}. A lower bound on sub-packetization level $\alpha$ of a general MSR code is derived in \cite{TamWanBru_access_tit}.  The same paper shows that $\alpha \geq r^{\frac{k-1}{r}}$ in the case of an optimal-access MSR code.  An improved lower bound for general MSR codes 
\bea 
2\log_2 \alpha (\log_{\left(\frac{r}{r-1}\right)} \alpha +1) + 1 & \geq & k,
\eea
appears in \cite{GopTamCal}. These developments made it clear that the ultimate goal in MSR code construction was to construct a high-rate MSR code that simultaneously had low sub-packetization level $\alpha$, low field-size $q$, arbitrary repair degree $d$ and the optimal-access property. 


In \cite{SasAgaKum}, a parity-check viewpoint is adopted to construct a high-rate MSR code for $d=n-1$ with a sub-packetization level $r^{\frac{n}{r}}$, requiring however, a large field-size. The construction was extended in \cite{RawKoyVis_msr}, to $d$ satisfying $k \leq d \leq n-1$.  In \cite{GopFazVar}, the authors provide a construction of MSR codes that holds for all $k \leq d \leq n-1$, but which once again required large field size. In  \cite{AgaSasKum}, the authors provide a construction for an optimal-access systematic MSR code that holds for any parameter set $(n,k,d=n-1)$ having sub-packetization $\alpha$ matching the lower bound given in \cite{TamWanBru_access_tit}. In \cite{TamWanBru, WangTamoBruck, WanTamBru_long, CadHuaLiMeh, SasAgaKum,AgaSasKum, RawKoyVis_msr, GopFazVar}, Combinatorial Nullstellansatz (see \cite{Alon_Comb}) is used to prove the MDS property due to which the codes are non-explicit and have large field sizes.


In \cite{RavSilEtz}, an explicit optimal-access, systematic MSR code is constructed with optimal $\alpha$, but for limited values of $n-k=2, 3$. In \cite{YeBar_1}, the authors present two different classes of explicit MSR constructions, one of which possessed the optimal-access property. Both constructions are for any $(n,k,d)$ with 
sub-packetization level growing exponential in $n$. 

In a major advance, in \cite{YeBar_2}, Ye and Barg present an explicit construction of a high-rate, optimal-access MSR code with $\alpha = r^{\lceil \frac{n}{r} \rceil}$, field size no larger than $r\lceil \frac{n}{r} \rceil$, and $d=(n-1)$. Essentially the same construction was independently rediscovered in \cite{SasVajKum_arxiv} from a different coupled-layer perspective, where layers of an arbitrary MDS codes are coupled by a simple pairwise coupling transform to yield an MSR code.  Just prior to the appearance of these two papers, in an earlier version of \cite{LiTangTian}, the authors show how a systematic MSR code can be converted into an MSR code by increasing the sub-packetization level by a factor of $r=(n-k)$ using  a pairwise symbol transformation. This result is then extended in \cite{LiTangTian}, to present a technique that takes an MDS code, increases sub-packetization level by a factor of $r$ and converts it into a code in which the optimal repair of $r$ nodes can be carried out. By applying this transform repeatedly $\lceil\frac{n}{r}\rceil$ times, it is shown that any scalar MDS code can be transformed into an MSR code.  It turns out that the three papers \cite{YeBar_2,SasVajKum_arxiv,LiTangTian}, either explicitly or implicitly, employed as a key part of the construction, essentially the same pairwise-coupling transform.   

Let $s=(d-k+1)$.  More recently, the lower bound $\alpha \geq s^{\frac{n}{s}}$ was derived in \cite{BalKum_subpkt} for optimal-access MSR codes.  The same paper also shows that the sub-packetization level of an MDS code that can optimally repair any $w$ of the $n$ nodes must satisfy $\alpha \ge s^{\lceil\frac{w}{s}\rceil}$.  These results established that the earlier constructions in \cite{SasAgaKum, RawKoyVis_msr, YeBar_2, SasVajKum_arxiv,LiTangTian, VajBalKum} were optimal in terms of sub-packetization level $\alpha$.    It is also shown in \cite{BalKum_subpkt}, that a vector MDS code that can repair failed nodes belonging to a fixed set of $Q$ nodes with minimum repair bandwidth and in optimal-access fashion, and having minimum sub-packetization level $\alpha = s^{\frac{n}{s}}$ must necessarily have a coupled-layer structure, similar to that found in \cite{YeBar_2,SasVajKum_arxiv,LiTangTian}. 
An explicit construction of MSR codes for $d < (n-1)$ with $\alpha$ achieving the lower bound $\alpha \geq s^{\frac{n}{s}}$ 
for $s =2,3,4$ was recently provided in \cite{VajBalKum}. 

\begin{open}Derive a tight lower bound on the sub-packetization level of MSR codes and provide matching constructions.
\end{open}
\begin{open}Constructions for explicit optimal-access MSR codes for any $(n,k,d)$ with optimal sub-packetization.
\end{open}




\begin{table}[ht!]
	\caption{A list of MSR constructions and the parameters. In the table $r=n-k$, $s=d-k+1$ and when All Node Repair is No, the constructions are systematic MSR. By `non-explicit' field-size, we mean that the order of the size of the field from which coefficients are picked is not given explicitly. }
	\bean
	\scalebox{0.85}{$
		\begin{array}{||l|c|c|c|c|c|c||}
		\hline \hline
		\text{MSR Code} &  \text{Parameters} & \alpha & \text{Field Size} & \text{All Node} & \text{Optimal} & \text{Notes}\\
		&&&& \text{Repair}& \text{ Access} &\\ \hline
		\text{\cite{ShaRasKumRam_ia} Shah et al.} & (n,k,d=n-1\ge 2k-1) & r & 2r & \text{No} & \text{Yes} & \text{IA framework}\\ \hline
		\text{\cite{SuhRam} Suh et al.} & (n, k, d \ge 2k-1) & s & 2r & \text{Yes} & \text{No} & \text{IA framework}\\
		& (n, k \le 3, d) & & & & &\\ \hline
		\text{\cite{RasShaKum_pm} Rashmi et al.}  & (n \ge 2k-1, k, d) & r & n & \text{Yes} & \text{No} & \text{Product Matrix framework}\\  \hline
		
		\text{\cite{PapDimCad} Papailiopoulos et al.} & (n,k,d=n-1)& r^k & \text{non-explicit} &\text{No} &  \text{No} & {\text{High rate systematic MSR}}\\ \hline 
		\text{\cite{TamWanBru} Tamo et al. }  & (n,k,d=n-1)  & r^{k+1} & \leq 4 \text{ when } r \leq 3, & \text{Yes} & \text{Yes} & \text{High rate MSR}\\
		\text{\cite{WangTamoBruck} Wang et al.} & & &\text{else non-explicit} & & & \text{known as Zigzag codes}\\
		\hline
		\text{\cite{CadHuaLiMeh} Cadambe et al.}& (n\ge \frac{3k}{2}, k, d=n-1) & O(k^2)& \text{non-explicit} & \text{No} & \text{Yes} & \\ \hline
		&  &  & & & &\text{Introduced Parity-check }\\
		\text{\cite{SasAgaKum} Sasidharan et al.} & (n, k, d=n-1) & r^{\left\lceil\frac{n}{r}\right\rceil} & O(n^r)& \text{Yes}& \text{Yes} & \text{viewpoint, Optimal $\alpha$}\\ \hline
		\text{\cite{GopFazVar} Goparaju et al.} & (n,k,d) & s^{k {r \choose s}} & \text{-} & \text{No} & \text{Yes}  & \text{Very large field-size needed.} \\ 
		& &  & & &  & \text{See Sec. IV in \cite{GopFazVar} for details.} \\ \hline
		\text{\cite{RawKoyVis_msr} Rawat et al.} & (n, k, d) & s^{\left\lceil\frac{n}{s}\right\rceil} & O(n^r) & \text{Yes}& \text{Yes} & \text{ Extended \cite{SasAgaKum} for $d<n-1$}\\ \hline
		\text{\cite{YeBar_1} Ye et al.} & (n,k,d) & s^n & sn & \text{Yes} & \text{No} & \\
		& (n,k,d) & s^{n-1} & n+1 & \text{Yes} & \text{Yes} &\\ \hline
		\text{\cite{YeBar_2} Ye et al.} & & & & & &\\ 
		\text{\cite{SasVajKum_arxiv} Sasidharan et al.} & (n, k, d=n-1)& r^{\left\lceil\frac{n}{r}\right\rceil} & r\lceil\frac{n}{r}\rceil & \text{Yes} & \text{Yes} & \text{Optimal } \alpha\\
		\text{\cite{LiTangTian} Li et al.} & &  & & & & \text{for optimal-access MSR}\\ \hline
		\text{\cite{VajBalKum} Vajha et al.} & (n,k,d) &  & & & & \\
		& d \in \{k+1, k+2, k+3\} & s^{\left\lceil \frac{n}{s} \right\rceil} & O(n) & \text{Yes} &\text{Yes}&\\ \hline
		\hline \hline
		\end{array}	
		$}
	\eean 
\end{table}

\subsection{Constructions of MSR Codes}
\paragraph{Product Matrix Construction \cite{RasShaKum_pm}:} We provide a brief description of the PM construction for parameter set $(n, k, d=2k-2), (\alpha=k-1, \beta = 1, B=k(k-1))$. The message symbols $\{u_i\}_{i=1}^B$ are arranged in the form of a $(d \times \alpha)$ matrix $M$: $M = \left[ \begin{array}{cc}
S_1 & S_2
\end{array}\right]^T$,
where the $S_1, S_2$ are symmetric $(k-1) \times (k-1)$ matrices containing the $B=k(k-1)$ message symbols. 
Encoding is carried out using a $(n \times d)$ matrix  $\Psi = [\begin{array}{cc}
\Phi & \Lambda \Phi \end{array}]$, where $\Phi$ is an $n \times (k-1)$ matrix and $\Lambda$ is a diagonal matrix. Let the $i$-th row of $\Psi$  be $\psi_i^T$, the $i$-th row of $\Phi$ be $\phi_i^T$ and the $i$-th diagonal element in $\Lambda$ be $\lambda_i$. The $\alpha$ symbols stored in node $i$ are given by: $\uc_i^T = \psi_i^T M = \phi_i^T S_1 + \lambda_i \phi_i^T S_2.$\\
The matrix $\Psi$ is required to satisfy the properties: 1) any $d$ rows of $\Psi$ are linearly independent, 2) any $\alpha$ rows of $\Phi$ are linearly independent and 3) the $n$ diagonal elements of $\Lambda$ are distinct.

\underline{Node Repair:} Let $f$ be the index of failed node, thus the aim is to reconstruct $\uc_f$.  The $i$-th helper node, $h_i$, $i \in [d]$, passes on the information: $\uc_{h_i}^T \phi_f = \psi_{h_i}^T M \phi_f$.  Upon aggregating the repair information we obtain the vector,
\bean
\left[\begin{array}{cccc}
	\psi_{h_1} & \psi_{h_2} & \cdots & \psi_{h_d}
\end{array}\right]^T
\left[M \phi_f \right].
\eean
As any $d$-rows of $\Psi$ are linearly independent, the vector $M\phi_f$ can be recovered.  From $M\phi_f$, we can obtain $S_1 \phi_f$ and $S_2\phi_f$.  Since $S_1$ and $S_2$ are symmetric, we can recover the contents $\uc_f^T = \phi_f^T S_1 + \lambda_f \phi_f^T S_2$ of the replacement node.  

\underline{Data Collection:} Let $\Psi_{\text{DC}} = [\begin{array}{cc} \Phi_{\text{DC}} & \Lambda_{\text{DC}} \Phi_{\text{DC}} \end{array}]$ be the $(k \times d)$ sub matrix of $\Psi$ corresponding to the $k$ nodes contacted for data collection. We wish to retrieve $M$ from $\Psi_{\text{DC}} M = \Phi_{\text{DC}} S_1 + \Lambda_{DC} \Phi_{\text{DC}} S_2$. This can be done in three steps:
\ben
\item First compute $\Psi_{\text{DC}} M \Phi_{\text{DC}}^T =\Phi_{\text{DC}} S_1 \Phi_{\text{DC}}^T + \Lambda_{DC} \Phi_{\text{DC}} S_2 \Phi_{\text{DC}}^T $ and set $P = \Phi_{\text{DC}} S_1 \Phi_{\text{DC}}^T$, $Q = \Phi_{\text{DC}} S_2 \Phi_{\text{DC}}^T$. 
\item It is clear that $P, Q$ are symmetric. Thus we know both $P_{ij} + \lambda_{i} Q_{ij}$ and $P_{ij} + \lambda_{j} Q_{ij}$. Since $\lambda_i \ne \lambda_j$ for $i \ne j$, we can recover $P_{ij}$ and $Q_{ij}$ for all $i \ne j$.
\noindent \item Since we know $P_{ij}$ for $j \neq i$, we can compute the vector $\phi_i^T S_1 [\phi_1, \cdots, \phi_{i-1}, \phi_{i+1}, \cdots, \phi_{k}]$.  Since any $\alpha$ rows of $\Phi$ are linearly independent, we can recover $\{\phi_i^T S_1 | 1 \le i \le k \}$.  For any set of $\alpha$ distinct elements $\phi_i^T$, we can compute 
$\left[ \begin{array}{ccc}
\phi_1 & \cdots & \phi_{\alpha}
\end{array}\right]^T S_1$, 
from which $S_1$ can be recovered. $S_2$ can be similarly recovered from $Q$. The present description assumes data collection from the first $k$ nodes, while a similar argument holds true for any arbitrary set of $k$ nodes.
\een

\paragraph{Coupled Layer Code:\label{par:clay} } We present here the constructions in \cite{YeBar_2,SasVajKum_arxiv,LiTangTian} from a coupled-layer perspective.  We explain the construction here 
only for parameter sets of the form:
\bean
(n=st,\ k=s(t-1),\ d=n-1), (\alpha=s^t,\ \beta=s^{t-1}), q \ge n,
\eean
where $s \ge 1, t \ge 2$.  (The construction can however, be extended to yield MSR codes for any $(n,k,d=n-1)$ using a technique called shortening). The coupled-layer code can be constructed in two steps: (a) in the first step, we layer $\alpha$, $(n,k)$ MDS codewords to form an uncoupled data-cube, (b) in the second step, the symbols within the uncoupled-data cube are transformed using a pairwise-forward-transform (PFT) to obtain the coupled layer code.  While we discuss only the case when the MDS code employed in the layers is a scalar MDS code, there is a straightforward extension that permits the use of vector MDS codes (see \cite{SasVajKum_arxiv}). 

Let us first consider the $n \alpha$ symbols $\{U(x,y,\uz) \ | \ (x,y) \in \Zs \times \Zt, \uz \in \Zst \}$  of an uncoupled code $\calu$ where each code symbol $U(x,y,\cdot)$ is a vector of $\alpha$ symbols in $\fq$. These $n\alpha$ symbols can be organized to form a three-dimensional (3D) data cube (see Fig.\ref{fig:datacube_uncoupled}), where $(x,y) \in \Zs \times \Zt$ is the node index and where $\uz \in \Zst$ serves to index the contents of a node.  For fixed $\uz \in \Zst$, we think of the symbols $\{U(x,y,\uz) \ | \ (x,y) \in \Zs \times \Zt\}$ as forming a plane or a layer and thus the value of $\uz$ may be regarded as identifying a plane or layer. The symbols in each layer of the uncoupled data cube form an $(n,k)$ MDS code.

\begin{figure}[!ht]
	\centering
	\begin{minipage}[c]{0.3\textwidth}
		\centering
		\includegraphics[width=2.2in]{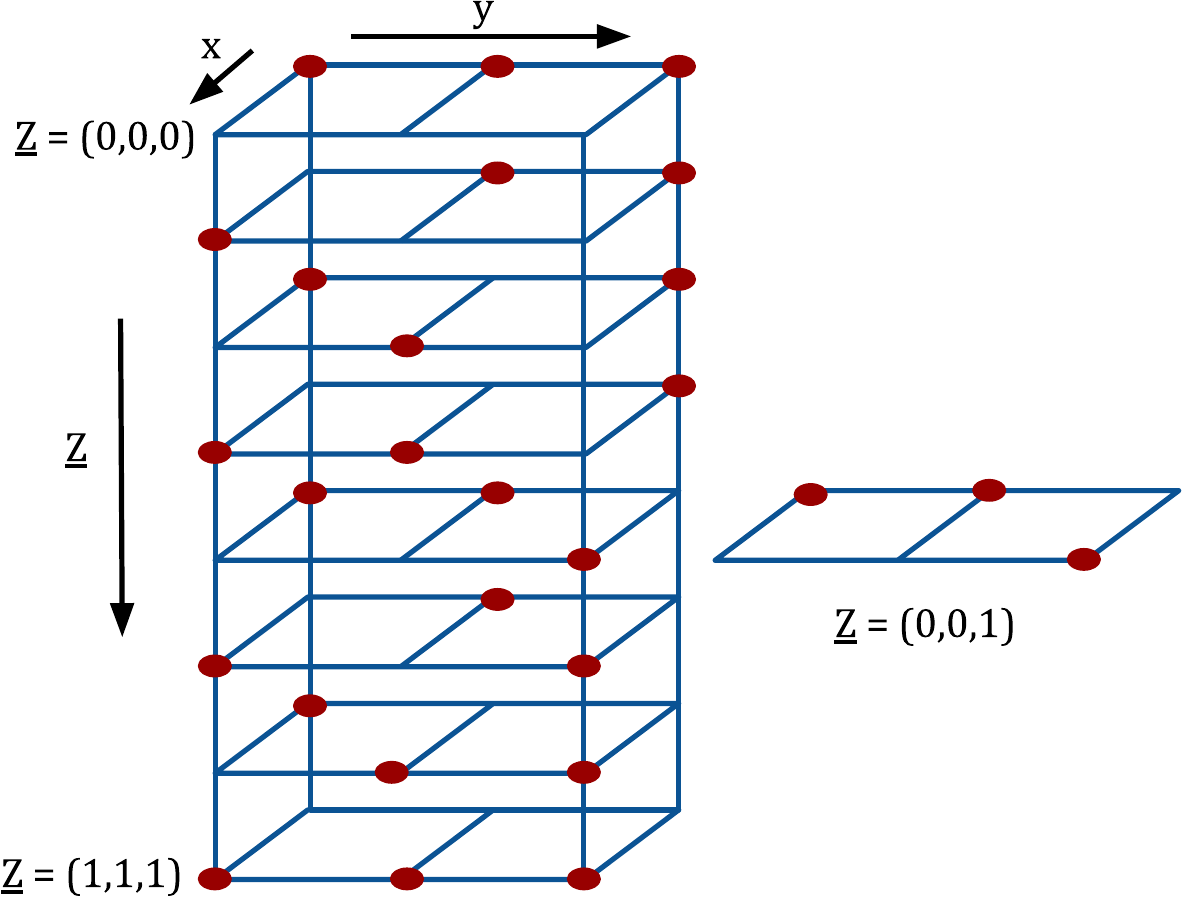}
	\end{minipage}
	\hspace{0.05\textwidth}
	\begin{minipage}[c]{0.6\textwidth}
		\centering
		\includegraphics[width=2.2in]{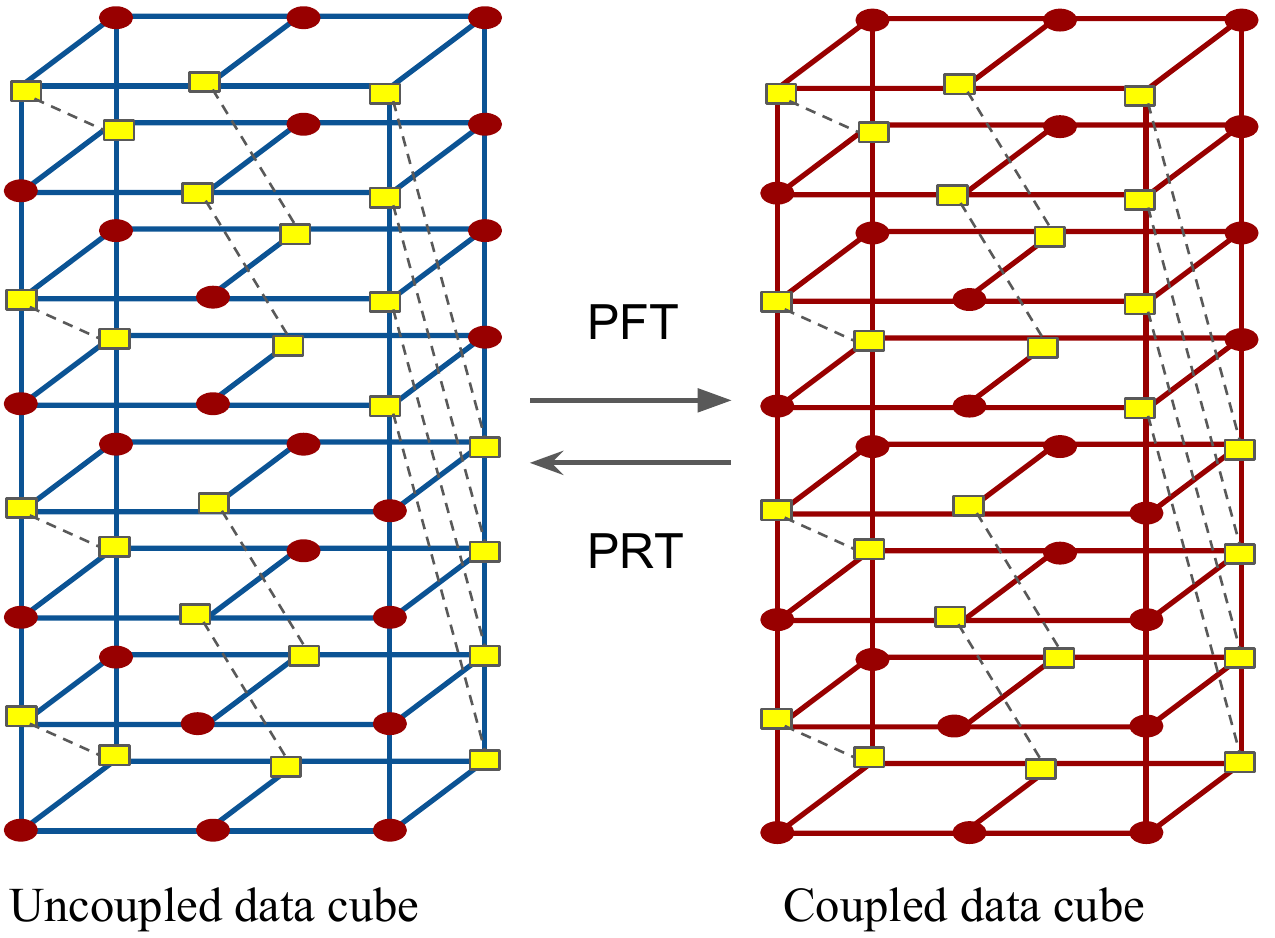}
	\end{minipage}\\[3mm]
	\begin{minipage}[t]{0.3\textwidth}
		\centering
		\caption{Uncoupled data cube for $s=2$, $t=3$. The red dots represent plane-index $\uz$.}
		\label{fig:datacube_uncoupled}
	\end{minipage}
	\hspace{0.05\textwidth}
	\begin{minipage}[t]{0.6\textwidth}
		\centering
		\caption{Paired symbols are shown using yellow rectangles connected by dotted lines. Uncoupled symbols are transformed using PFT to get the coupled symbols in the coupled data cube.}
		\label{fig:datacube_ucoupled_pairs}
	\end{minipage}
\end{figure}

Let, $\Theta$ be the $\left( (n-k) \times n \right)$ parity check (p-c) matrix of an arbitrarily chosen $(n,k)$ scalar MDS code defined over $\fq$. Let $\theta_{x,y}(\ell)$ denote the element of $\Theta$ lying in the $\ell$th row, and $(x,y)$th column. Then the symbols of the uncoupled code satisfy the p-c equations:
\bea
\label{eq:uncoupled_code}
\sum\limits_{(x,y) \in \Zs \times \Zt} \theta_{x,y}(\ell) U(x,y;\uz) &=& 0, \ \ \forall \ell \in [0, n-k-1], \ \forall \uz \in \Zst.
\eea

Next, consider an identical data-cube (see Fig.~\ref{fig:datacube_ucoupled_pairs}) containing the $n\alpha$ symbols 
\bean 
\{C(x,y,\uz) \ | \ (x,y) \in \Zs \times \Zt, \uz \in \Zst \}
\eean corresponding to the coupled-layer code. This data-cube will be referred to as the coupled data cube.  The symbols of the coupled data cube are derived from the symbols of the uncoupled data cube as follows.  
Let $\gamma$ be an element in $\fq \setminus \{0\}$, $\gamma^2 \ne 1$.   Let us define $\uz(y,x) = (z_0, \cdots, z_{y-1}, x, z_{y+1}, \cdots, z_{t-1})$. Each symbol $C(x,y,\uz)$ which is such that $z_y \ne x$ is paired with a symbol $C(z_y,y,\uz(y,x))$.  The values of the symbols so paired, are derived from those of their counterparts in the uncoupled data cube as per the $(2 \times 2)$ linear transformation given below, termed as the PFT: 
\bea
\label{eq:pft}
\left[ \begin{array}{c}
	C(x,y,\uz)\\
	C(z_y,y,\uz(y,x))
\end{array} \right] = \left[\begin{array}{cc}
1 & \gamma\\
\gamma & 1
\end{array}\right]^{-1} \left[\begin{array}{c}
U(x,y,\uz)\\
U(z_y,y,\uz(y,x))
\end{array}\right]. \eea
In the case of the symbols $C(x,y,\uz)$ when $z_y=x$, the relation between symbols in the two data cubes is even simpler and given by: $C(x,y,\uz) = U(x,y,\uz)$. The pairwise reverse transform (PRT) is simply the inverse of the PFT and is used to obtain the uncoupled symbols $U(\cdot)$ from the coupled symbols $C(\cdot)$.
The p-c equations satisfied by the coupled-layer code can be derived using the p-c equations \eqref{eq:uncoupled_code} satisfied by the symbols in the uncoupled data cube and the PRT :
\bea
\scalebox{0.9}{$
\label{eq:coupled_code}
\sum\limits_{(x,y) \in \Zs \times \Zt} \theta_{x,y}(\ell) C(x,y,\uz) + \sum\limits_{y \in \Zt} \sum\limits_{x \ne z_y} \gamma \theta_{x,y}(\ell) C(z_y,y,\uz(y,x)) = 0, \ \ \forall \uz \in \Zst, \ \ell \in [0,n-k-1].$}
\eea
\underline{Node Repair:} Let $(x_0,y_0)$ be the failed node. To recover the symbols $\{C(x_0, y_0, \uz) \ | \  \uz \in \Zst \}$, each of the remaining nodes $(x,y) \ne (x_0, y_0)$ sends helper information:
$\{C(x,y, \uz) \ | \ \uz \in \Zst,\ z_{y_0} = x_0 \}$.  Focusing on \eqref{eq:coupled_code} for $\uz$ such that $z_{y_0} = x_0$ and retaining on the left side the unknown symbols, leads to equations of the form: 
\bea
\label{eq:repaireqns}
\theta_{x_0, y_0}(\ell) C(x_0, y_0, \uz) + \sum\limits_{x \ne x_0} \gamma \theta_{x,y_0}(\ell) C(x_0,y_0,\uz(y_0, x)) &=& \kappa^*, \ \ \forall \ell \in [0, n-k-1],
\eea
where $\kappa^*$ is a known value.  These equations can be solved for the contents of the replacement node. 

\underline{Data Collection:} Please refer to \cite{SasVajKum_arxiv} for the proof of data collection property. 



\paragraph{Ye-Barg Codes \cite{YeBar_1}:} 


In \cite{YeBar_1} the authors present two constructions, for non optimal-access MSR and optimal-access MSR codes respectively. These are the only known MSR constructions that are explicit and yield MSR codes for any parameter set $(n,k,d)$. The same codes are also optimal for the repair of multiple nodes. We describe here, for simplicity, the construction of $(n,k,d)$ MSR codes having parameters: $(n, k, d), (\alpha=s^n, \beta=s^{n-1}), q \ge sn $ where $s=d-k+1$, defined over finite field \fq\ for $s \ge 1$. Let $\{C(i,\uz) \ | \ i \in [n], \uz \in \Zsn \}$ be the collection of $n\alpha$ symbols of a codeword, where $i$ is the node index and $\uz$ is the scalar symbol index. The code is defined via the p-c equations given below:
\bea
\label{eq:yebarg_parity}
\sum\limits_{i \in [n]} \lambda_{i, z_i}^{\ell} C(i; \uz) &=& 0, \ \ \forall \uz \in \Zsn, \ \ell \in [0, n-k-1],
\eea
where the $\{ \lambda_{i,j}, i \in [n], j \in [0,s-1]\}$ are all distinct, thereby requiring a field size $q \ge sn$.

\underline{Node Repair:} Let $f$ be the failed node, $D$ be the set of $d$ helper nodes. The helper information sent by a node $i \in D$ is given by: $\{ \mu_f^i(\uz)=\sum\limits_{j=0}^{s-1} C(i, \uz(f,j)) \ | \ \uz \in \Zsn, z_f = 0 \}$.  Next, fixing $z_i, \ \forall i \in [n] \setminus \{f\}$ and summing equations \eqref{eq:yebarg_parity} over the values of $z_f$, we get:
\bea
\label{eq:yebarg_repair1}
\sum\limits_{z_f=0}^{s-1} \lambda_{f,z_f}^{\ell} C(f, \uz) + \sum\limits_{i \in [n] \setminus \{f\} } \lambda_{i,z_i}^{\ell} \mu_f^i(\uz) &=& 0, \ \ \forall \ell \in [0,n-k-1].
\eea
It can be shown that the collection of symbols $\{\mu_f^i(\uz) | i \in [n] \setminus \{f\} \}$ form an $[n-1, d]$ MDS code. Therefore, all the $\mu_f^{i}(\uz)$ can be computed from the known $d$ values supplied by  the helper nodes and 
the symbols $\{C(f,\uz) \ | \ \uz \in \Zsn\}$ can thus be recovered from \eqref{eq:yebarg_repair1}.

\underline{Data Collection:} For every $\uz \in \Zsn$, the collection $\{C(i,\uz) | i \in [n]\}$ forms an $(n,k)$ MDS code. Therefore, any $(n-k)$ erased symbols can be recovered.\\ 

{\bf Multiple Node Repair} Let $1\le t \le n-k$ be the number of erasures to be recovered. It was shown in \cite{CadJafMalRamSuh} that the minimum repair bandwidth required to repair $t$ erasures in an MDS code having sub-packetization level $\alpha$ is lower bounded by $\gamma_t \ge \frac{t(n-t)\alpha}{n-k}$. Given that $ k \le d \le n-t$ is the number of helper nodes that need to be contacted during the repair of $t$ nodes, $\gamma_t$ is lower bounded by: $\gamma_t \ge \frac{td\alpha}{d+t-k}$. The Ye-Barg code presented above achieves this bound~\cite{YeBar_1}. The $t$ node repair discussed here assumes a centralized repair setting whereas an alternate, cooperative repair approach is discussed in Section \ref{subsec:coop_rep}. \\

{\bf Adaptive Repair} Adaptive-repair $(n,k)$ MSR codes are MSR codes that can repair a failed node by connecting to any $d$ nodes, for any $d \in [k, n-1]$ and can reconstruct the failed node by downloading $\frac{\alpha}{d-k+1}$ symbols each from the $d$ helper nodes. Constructions of MSR codes with adaptive repair can be found in \cite{GopFazVar, YeBar_1, MahMohKhis_MSRAdapt}.

\section{On the Storage-Repair Bandwidth Tradeoff under Exact Repair \label{sec:on-srb}}

We distinguish between the S-RB tradeoffs for exact and functional-repair RG code, by referring to them as the ER and FR tradeoff respectively. The file size $B$ under exact repair cannot exceed that in the FR case since ER may be regarded as a trivial instance of FR.  However, unlike in the case of functional-repair codes, the data collection problem in the exact-repair setting, cannot be identified with a multicast problem simply because each replacement node for a failed node acts as a sink for a different set of data.     Thus it is not clear that the cut-set bound for FR can be achieved under ER, leaving the door open for an S-RB tradeoff in the case of ER that lies strictly above and to the right of the FR tradeoff in the $(\alpha, \beta)$-plane.  There do exist constructions of exact-repair MBR and MSR codes meeting the cut-set bound with equality, showing that the ER tradeoff coincides with the FR tradeoff at the extreme MSR and MBR points. 



\subsection{The Non-existence of ER Codes Achieving FR tradeoff\label{sec:non-exist}}

The first major result on the ER tradeoff was the result in \cite{ShaRasKumRam_rbt}, showing that apart from the MBR point and a small region adjacent to the MSR point, there do not exist ER codes whose $(\alpha, \beta)$ values lie on the interior point of the FR tradeoff.  We set $\alpha_{\text{\tiny MSR}} = \beta(d-k+1)$ to be the value of $\alpha$ at the MSR point.  
\bthm \label{thm:shah_non_exist} For any given values of $(n,k \geq 3,d)$, ER codes having parameters $(\alpha, \beta, B)$ corresponding to an interior point on the FR tradeoff do not exist, except possibly for $\alpha$ in the range 
\bea \label{eq:exception}
\alpha_{\text{\tiny MSR}} & \leq \ \alpha \ \leq & \alpha_{\text{\tiny MSR}} 
\left( 1 + \frac{1}{   \alpha_{\text{\tiny MSR}} (\alpha_{\text{\tiny MSR}}+1)     }\right) ,
\eea
corresponding to a small region in the neighborhood of the MSR point.
\ethm

\begin{figure}[ht!]
	\centering
	\includegraphics[width=2.5in]{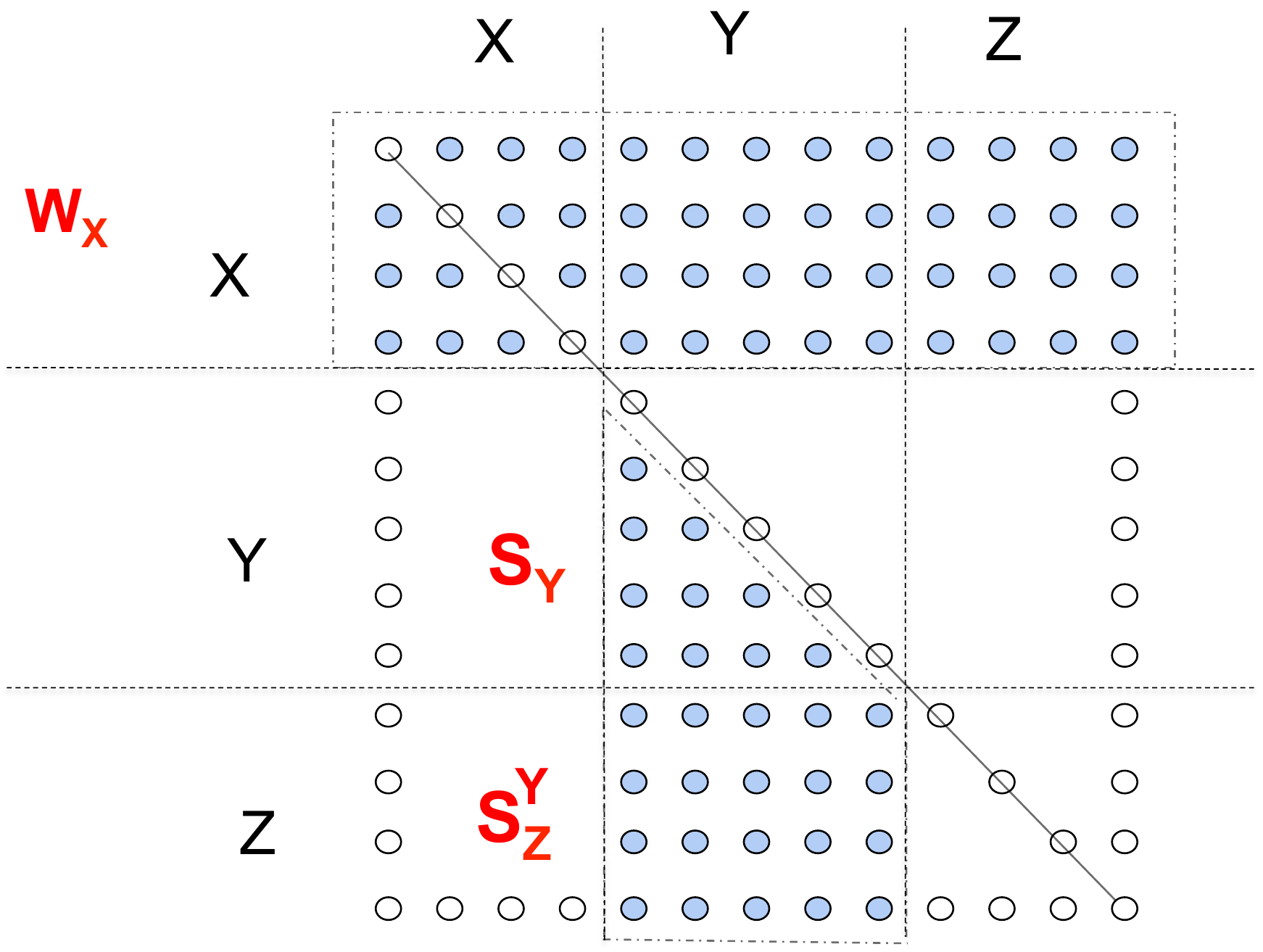}
	\caption{The repair matrix \label{fig:repmatrix}}  
\end{figure}


\bpf (Sketch) By restricting attention to any $(d+1)$ symbols of an RG code having parameter set $(n,k,d, (\alpha, \beta), B)$ one obtains a second RG code with parameter set $((d+1),k,d, (\alpha, \beta), B)$ in which all the remaining nodes participate in the repair of a failed node.  This simplifies the analysis of the repair setting and with this in mind, in the proof, we set $n=(d+1)$. 
When the message vector $\underline{u}$ is picked uniformly at random, we have associated nodal random variables $\{W_i \mid i \in [n]\}$ and repair data variables $\{S_i^j \mid i \in [n] \setminus j \}$, where $S_i^j$ denotes the data passed from node $i$ to replacement node $j$. The repair matrix $\mathbb{S}$ (see Fig.~\ref{fig:repmatrix}) is an $(n \times n)$ matrix whose $(i,j)$th entry $i \neq j$, is $S_i^j$.  The diagonal elements of $\mathbb{S}$ do not figure in the discussion and maybe set equal to $0$.  Given subsets $H, N \subset [n]$, we 
set $W_N = \{W_i \mid i \in N\}$, $S_H^N= \{S_i^j \mid i \in H, j \in N\}$. 
%
%
%
%
%
We introduce the index sets $X = \{1, 2, \ldots, m\}$, $Y = [k]\setminus X$ and $Z = [k+1, \ d+1]$ for $m \leq k$. The file size $B$ can be expressed in terms of the joint entropy of the node and repair-data variables (with logs computed to base $q$):
\bea
B & = & H(W_X, W_Y) \ = \ H(W_X, S_{Y \cup Z}^Y) \\
& = & H(W_X) + H(S_{Y \cup Z}^Y \mid W_X) \le  H(W_X) + \sum\limits_{j=m+1}^k H(S_{[m+1,j-1] \cup Z}^j \mid W_X) \label{eq:cond_entropy} \\
& \leq & m\alpha + \sum_{i=m+1}^{k} (d-i+1)\beta \ := \ B_m, \  \ m = 0,1,\cdots, k. \label{eq:ER_cutset}
\eea
The cut-set bound in \eqref{eq:cutsetbound} corresponds to the inequalities: $B \leq \min_{m=0,1,\ldots, k} B_m$. For the  bound to hold with equality, the joint random variables $S_Y^Y$ and $S_Z^Y$ must have maximum entropy. However it can be shown that the entropy of a row in the repair matrix is limited by $\beta$ if the cut-set bound holds with equality. This leads to a contradiction, concluding the proof. 
\epf
Theorem~\ref{thm:shah_non_exist} does not however, rule out the possibility of an ER code having tradeoff approaching the FR tradeoff asymptotically i.e., as the file size $B \rightarrow \infty$. 

\subsection{The S-RB Tradeoff for $(4,3,3)$}

It is possible that the entropies of 
the random variables involved satisfy Shannon inequalities other than the ones we have noted and which shed light on the ER tradeoff. For the particular case $(n,k,d)=(4,3,3)$, Tian \cite{Tia} was able to identify such an inequality with the help of a modified version of the Information Theory Inequality Prover (ITIP) \cite{ITIP,Yeu}.  
\begin{figure}[ht!]
	\bc
	\includegraphics[width=2.5in]{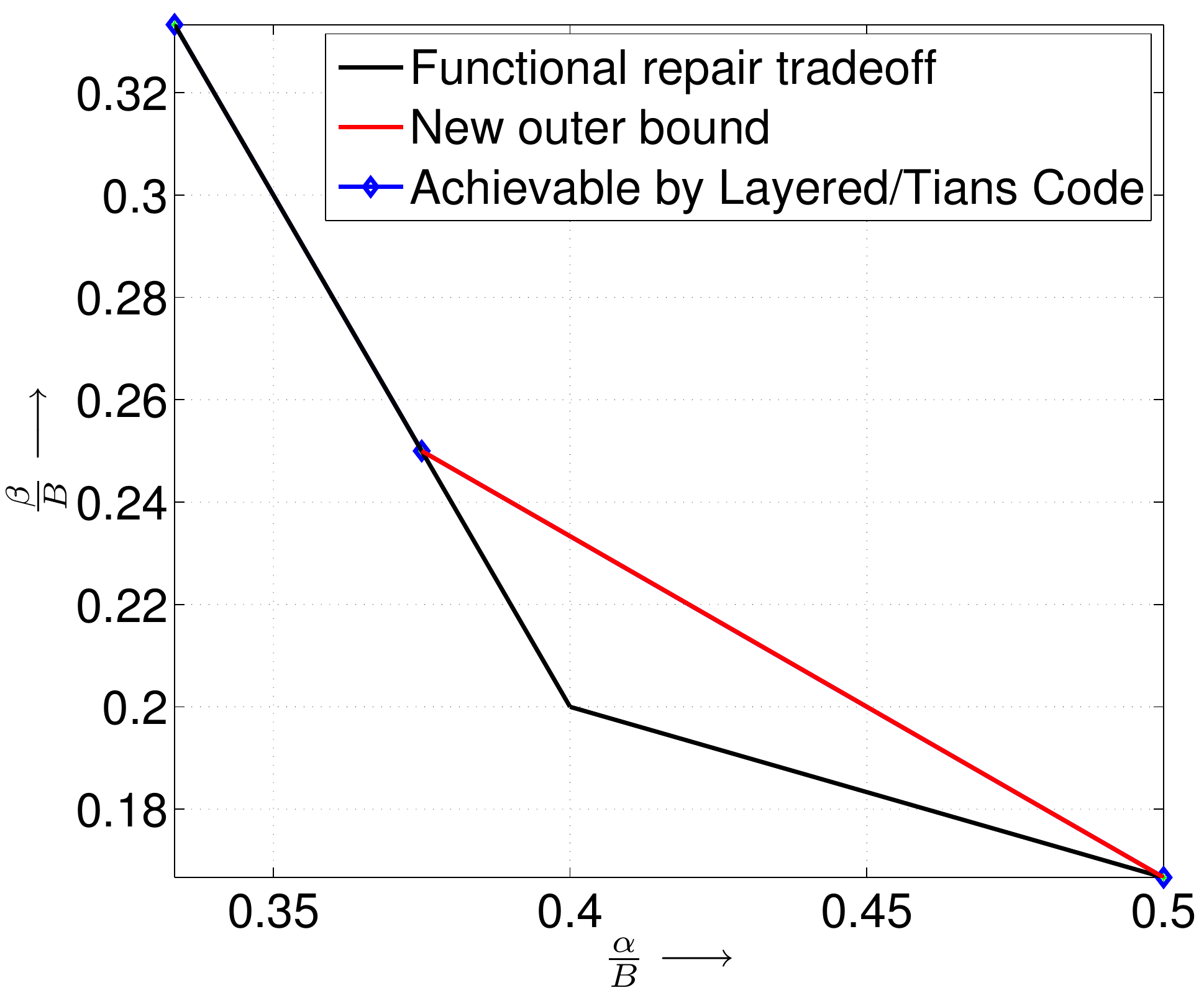}
	\caption{The (4,3,3) normalized tradeoff. \label{fig:433}}  
	\ec
\end{figure}
Let $\bar{\alpha} = \frac{\alpha}{B}$, $\bar{\beta} = \frac{\beta}{B}$ represent the {\em normalization} of $\alpha$ and $\beta$ with respect to file size $B$. A point $(\bar{\alpha}, \bar{\beta})$ is said to be achievable if for any $\epsilon > 0$, there exists an ER-RG code whose $(\bar{\alpha}_1, \bar{\beta}_1)$ is $\epsilon$-close to $(\bar{\alpha}, \bar{\beta})$.  The normalized tradeoff, i.e., the tradeoff expressed in terms of $\bar{\alpha}$ and $\bar{\beta}$ allows comparison of codes across file sizes $B$.  In the limit as $B \rightarrow \infty$, the S-RB tradeoff 
becomes a smooth curve. 
Let $C_1$, $C_2$ be RG codes over $\mathbb{F}_q$ having respective parameter sets $(n,k,d, (\alpha_1, \beta_1), B_1)$ and \linebreak $(n,k,d, (\alpha_2, \beta_2), B_2)$.  Consider a codeword array $\underline{c}$ obtained by vertically stacking  $M_1$ codeword arrays of $C_1$ and $M_2$ codeword arrays of $C_2$. The code $C$ comprising of all such arrays is said to be the {\em space-shared code} of $C_1$ and $C_2$. Then $C$ is also an RG code with parameter set $(n,k,d, (M_1\alpha_1+M_2\alpha_2, M_1\beta_1+M_2\beta_2), M_1B_1+M_2B_2)$. The notion of space-sharing clearly extends to multiple codes.


\bthm For $(n,k,d)=(4,3,3)$, the achievable region ${\cal R}$ is given by
\bea
{\cal R} & = & \left\{ (\bar{\alpha}, \bar{\beta}) \middle\vert 3\bar{\alpha} \ \geq \ 1, \ 2\bar{\alpha} + \bar{\beta}  \ \geq \ 1, 6\bar{\beta} \ \geq \ 1, 4\bar{\alpha} + 3\bar{\beta} \ \geq \ 1  \right\}. \ \ \ \ \ \ \ \ \ \ \ \label{eq:433bound}  
\eea
\ethm 
\bpf Of the four inequalities listed, the first $3$ follow the entropy constraints listed in \eqref{eq:ER_cutset} above.  The last inequality $4\bar{\alpha} + 3\bar{\beta} \ \geq \ 1$ does not follow from \eqref{eq:ER_cutset}, and was found in \cite{Tia} using an ITIP.  It remains to construct a code that operate on points on the $(\bar{\alpha}, \bar{\beta})$-plane, satisfying the inequalities with equality. A $[4,3]$ single parity-check code serves as an MSR code $C_1$ for $(4,3,3)$. A $(4,3,3)$ MBR code $C_2$ can be constructed using the polygonal construction described in Sec.~\ref{sec:mbr}. A hand-crafted code $C_3$ operating at the interior point of deflection (see Fig.~\ref{fig:433}) is given in \cite{Tia}. Every point on the lines determined by equality in \eqref{eq:433bound} is achieved by a code obtained by space-sharing among $C_1, C_2$ and $C_3$.
\epf
\subsection{Layered Codes for Interior Points}
\begin{figure}[ht!]
	\centering 
	\begin{tabular}{|c|c|c|c|c|} 
		\hline 
		\textcolor{blue}{1} & \textcolor{blue}{2} & \textcolor{blue}{3} & \textcolor{blue}{4} &   \textcolor{blue}{5} \\
		\hline 
		$c_{11}$ & $c_{12}$ & $c_{13}$  & & \\
		\hline 
		&  $c_{21}$ & $c_{22}$ & $c_{23}$ &  \\
		\hline 
		& & & & \\
		\hline 
		$c_{m1}$ &  & $c_{m2}$ & $c_{m3}$  & \\
		\hline 
		\vdots  & \vdots & \vdots & \vdots & \vdots \\
		\hline 
	\end{tabular}
	\caption{An $(n=5,k=4,d=4)$ canonical layered code. \label{fig:layered}}  
\end{figure}
A simple code-construction technique based on the layering (see Fig.~\ref{fig:layered} for an example) of MDS codes turns out to provide codes that perform well with respect to file size in the interior region of the S-RB tradeoff.   Let ${\cal C}_{\text{\tiny MDS}}$ be an MDS code having parameters $[w+\gamma,w,\gamma+1]$.  Let $n$ be such that $w+\gamma\leq n$ and $L={n \choose w+\gamma}$. Let $\{S_i \subset [n] \mid i =1,2,\cdots L\}$ denote an ordering of the collection of all possible $(w+\gamma)$ subsets of $[n]$.  Let $\underline{u}_i \in \mathbb{F}_q^w$, $i = 1,2,\cdots L$ be $L$ message vectors, not necessarily distinct, and $\underline{c}_i$ be the codeword in ${\cal C}_{\text{\tiny MDS}}$ associated with $\underline{u}_i$. We create an $(L \times n)$ array in which we place the symbols of codeword $\underline{c}_i$ in the location specified by subset $S_i$.  It turns out that this array represents an array code which possesses the data collection property of an RG code, but not the repair property.  By replicating the array a certain number $V$ of times, it turns out that one obtains a regenerating code with parameters $(n, k = n-\gamma, d=k, B_0=LVw)$, operating between the MSR and MBR points. Further details can be found in \cite{TiaSasAggVaiKum}.  We will refer to this code as the canonical layered code ${\cal C}_{\text{can}}$. The canonical layered-code construction has been extended to construct codes with $k < d$ by making use of an outer code designed using linearized polynomials.  An alternate generalization of the canonical code to the case of $k < d$ involved adding additional layers consisting of carefully designed parity symbols. Such an approach leads to the {\em improved layered codes} in \cite{SenSasKum_ita}, that turn out to be optimal for the set of parameters  $(n,k=3,d=n-1)$.

\subsection{ER Tradeoff Strictly Away from FR Tradeoff for all $(n,k,d)$}

In \cite{SasSenKum_isit}, it was shown that the ER tradeoff cannot approach the FR tradeoff even when $B \rightarrow \infty$ for any value of $(n,k,d)$. This was established by deriving a positive lower bound $0 < \delta < \beta$ on the gap between the ER and FR tradeoffs. 
\bthm \label{thm:away} The ER tradeoff between $\bar{\alpha}$ and $\bar{\beta}$ for any exact-repair regenerating code, with $k \geq 3$ is strictly separated from the FR tradeoff, apart from the MSR and MBR endpoints as well as the region surrounding the MSR point appearing in \eqref{eq:exception}. 
%
%
\ethm
The proof the theorem involves identifying contradicting bounds on the entropy of various trapezoidal-shaped subsets within the repair matrix. Subsequent papers \cite{Duu1},\cite{Duu2} derive better bounds, thereby improving the gap $\delta$ to go beyond $\beta$. In \cite{MohTan_isit}, the authors adopt a different approach by first providing three different expression for the entropy $B$ of the data file  involving mutual information between various repair-data variables, and taking a linear combination of these expressions that leads to a significantly tighter bound on $B$:
\bea
\label{eq:soh_final}B &\le & \min_{0 \leq p \leq  k} \frac{(3k-2p)\alpha + \frac{p(2(d-k)+p+1)\beta}{2} + (d-k+1)\min\{\alpha, p\beta\} }{3}.
\eea
The authors in \cite{SasPraKriVajSenKum} improve upon the result in \eqref{eq:soh_final} using repair-matrix techniques, in combination with the bound in Thm.~\ref{thm:away}, leading to the best-known outer bound on the ER tradeoff. For the case of $(n,k=3,d=n-1)$, the outer bound is achieved by the improved layered codes, thus characterizing the ER tradeoff. The bound also characterizes certain interior points when $k=4$ \cite{SasSenKum_isit}. 

\subsection{Determinant Codes for Interior Points}

The construction given in \cite{ElyMoh} has parameters $\alpha = {k \choose m}$, $\beta = {k-1 \choose m-1}$ and file size $B = m{k+1 \choose m+1}$, where $m \in \{1, 2, \ldots, k\}$ is an auxiliary parameter. The message symbols are first precoded to obtain $k {k \choose m}$ symbols, and these are then arranged in a data matrix $M$ of size $(k \times \alpha)$ in a particular manner. The codeword array is then obtained as in the case of the Product-Matrix framework introduced in \cite{RasShaKum_pm}, by setting $C_{n \times \alpha} = \psi_{n \times k} M_{k \times \alpha}$, where $\psi_{n \times k}$ is a Vandermonde matrix. The data collection and repair properties of the code are proved by making use of the Laplace expansion of determinants, and the codes for this reason, are called {\em determinant codes}. The codes achieve an outer bound discussed in the next subsection, and thus form an optimal family of codes for parameters $(n,k,k)$.  An extension of the construction to include the parameter set $(n,k,d=k+1)$ can be found in \cite{ElyMoh_Allerton}.


%

\subsection{ER Tradeoff under Linear Setting}

In \cite{PraKri}, \cite{Duu2}, \cite{ElyMohTan} the authors characterize the ER tradeoff for $(n,k=n-1,d=n-1)$ for the subclass of linear codes, using an approach that involves lower bounding the rank of the parity-check matrix of an RG code. The upper bound in \cite{PraKri} holds in general for any $(n,k,d=k)$.
\begin{thm} \label{thm:rankH_new_bound_k_eq_d} Consider an ER linear regenerating code with parameters $\{(n\geq 4, k, d), (\alpha, \beta)\}$ and file size $B = n\alpha - \rho$. Then
	\begin{eqnarray} \label{eq:bound_rank_H_gen}
	\rho & \geq & \left \{ \begin{array}{cl} \left \lceil \frac{2rn\alpha - n(n-1)\beta}{r^2+r}\right \rceil, & \frac{d\beta}{r} \leq \alpha \leq \frac{d\beta}{r-1}, \ \ 2 \leq r \leq n - 2 \\ 
	2\alpha - \beta, & \frac{d\beta}{n-1} \leq \alpha \leq \frac{d\beta}{n-2} \end{array} \right. .
	\end{eqnarray}
\end{thm}
The corresponding bound on file size $B$ coincides with the achievable region of layered codes when $k=d=(n-1)$.  Determinant codes achieve the above bound in general for $(n,k,k)$, thus characterizing the linear ER tradeoff in this case.
\begin{open} Characterization of ER tradeoff for general $(n,k,d)$ in both the linear and non-linear settings.
\end{open}

\section{Variations on the Theme of Regenerating Codes}\label{sec:variations}  

\subsection{Cooperative repair\label{subsec:coop_rep}}
{\it This subsection was contributed at the request of the authors, by Kenneth Shum.} 
The potential benefit of allowing data exchange among the nodes being regenerated while repairing multiple node failures simultaneously, was first investigated by Hu et al. in \cite{HuXuWanZhaLi}.  The cooperative-repair process consists of two phases. In the first phase, each of the new nodes selects a set of $d$ surviving nodes, and downloads a total of $d\beta_1$ symbols from them.  In the second phase, a new node downloads $\beta_2$ symbols from each of the other new nodes. If  $t$ new nodes are re-built at the same time, the repair bandwidth per new node is
$d\beta_1 + (t-1) \beta_2. $  As in the non-cooperative case, there is a tradeoff between the amount of data stored in a node and the repair bandwidth. In the following, we denote the repair bandwidth per new node by $\gamma$. The minimum-storage cooperative regenerating (MSCR) point and minimum-bandwidth cooperative regenerating (MBCR) point are determined in~\cite{KerScoStr} and \cite{ShuHu}, and are given by 
\begin{align*}
	(\alpha_{\text{MSCR}}, \gamma_{\text{MSCR}}) = \big(\frac{B}{k}, \frac{B(d+t-1)}{k(d+t-k)} \big),\ \ \ \ (\alpha_{\text{MBCR}},\gamma_{\text{MBCR}}) = \frac{B(2d+t-1)}{k(2d+t-k)} \ (1,1),
\end{align*}
where $t$ is the number of nodes to be repaired simultaneously. When $t=1$, they reduce to the corresponding operative points for single-node repair. The full FR tradeoff curve between storage and repair bandwidth per node is derived in~\cite{ShuHu}. 
\begin{table}[ht!]
	\caption{Parameters of explicit constructions of cooperative regenerating codes.}
	\label{table:CRC}
	\begin{center}
		\begin{tabular}{|c|c|c|} \hline
			Type & Code Parameters & Ref. \\ \hline \hline
			MBCR & $n, k$, $ k \le d \le (n-t)$, $t\geq 1$ & \cite{WanZhaCoopRegen} \\
			\hline
			MSCR & $n = d+2$, $k=t=2$ & \cite{Sco12} \\ \hline
			MSCR & $n = 2k$, $d=n-2$, $k\geq 2$, $t=2$ & \cite{ShuChe} \\ \hline
			MSCR & $n = 2k$, $d=n-t$, $k\geq 2$, $k\geq t\geq 2$ & \cite{ShuChe} \\
			& (repair of systematic nodes only)&\\ \hline
			MSCR & $n, k$, $ k \le d \le (n-t)$, $t\geq 1$ & \cite{YeBarg_Coop} \\ \hline
		\end{tabular}
	\end{center}
\end{table}
In the case of exact repair, the explicit construction of cooperative regenerating codes for all parameters at the minimum-bandwidth point was first presented in~\cite{WanZhaCoopRegen}. The construction  in~\cite{WanZhaCoopRegen} is presented in an alternate way in~\cite{ShuChe}. Constructions for minimum-storage cooperative codes are relatively rare (see e.g. \cite{Sco12} and \cite{ShuChe}). Table \ref{table:CRC} summarizes the existing constructions of MSCR and MBCR codes. We note that the MSCR codes in \cite{ShuChe} share the same encoding method as in \cite{SuhRam} and \cite{ShaRasKumRam_ia}. It is shown in \cite{ShuChe} that with the MSR codes in \cite{SuhRam} and \cite{ShaRasKumRam_ia}, we can repair multiple systematic nodes with repair bandwidth achieving the MSCR point. In \cite{YeBarg_Coop}, the authors present constructions for any $(n,k, k \le d \le n-t, t)$ MSCR codes.

The cooperative repair model was extended to partial cooperative repair in~\cite{LiuOgg_a}. The first phase of repair is the same as described above. Each of the $t$ new nodes contacts $d$ other nodes and download a total of $\beta_1$ data packets. In the second phase, a new node exchanges $\beta_2$ data packets with $t-s$ other new nodes, where $s$ is a system parameter between 1 and $t$. When $s=t$, it is the original single-loss repair model. When $s=1$, it reduces to the cooperative repair model. The minimum-storage and minimum bandwidth point are derived in~\cite{LiuOgg_a}. With partial collaboration, the minimum-storage and minimum-bandwidth operating points are given respectively by  
\bean
(\alpha,\gamma) = \big(\frac{B}{k}, \frac{B(d+t-s)}{k(d-k+t-s+1)}\big) & \text{ and } & (\alpha,\gamma) =  \frac{B(2d+t-s)}{k(2d-k+t-s+1)} \ (1,1).
\eean
Two explicit codes for partial collaborative repair are presented in~\cite{LiuOgg_b}. The code construction in~\cite{ShuChe} for MBCR codes can be extended to achieve all minimum-bandwidth points with partial collaboration. The security of cooperative regenerating codes is investigated in \cite{KoyRawVish,HuaParXia}.

\subsection{MDS Codes with Repair Capability} 

We discuss in this subsection, vector MDS codes that are not MSR, which nevertheless offer some savings in repair bandwidth in comparison to the conventional repair of RS codes while keeping the sub-packetization level $\alpha$ small. The piggybacking framework introduced in \cite{RasShaRam_Piggyback}, was one of the first such efforts.  In \cite{GurRaw_MDS}, the authors introduce codes that offer a choice of sub-packetization levels, namely, $\alpha = r^p$ for $1 \le p < \left\lceil\frac{n}{r}\right\rceil$. The corresponding repair download from each helper node is given by $\beta = (1+\frac{1}{p})r^{p-1}$. When $p=\left\lceil\frac{n}{r}\right\rceil$ these codes coincide with the construction in \cite{SasAgaKum}. A similar approach was followed by the authors of \cite{KralevskaGligoroski} where they provide constructions for MDS codes for any given $1 \le \alpha \le r^{\left\lceil\frac{k}{r}\right\rceil}$. However, the constructions here are restricted to systematic node repair and the bandwidth needed from each helper node is not uniform. 
These constructions are motivated by the systematic MSR code with $\alpha=r^{\left\lceil\frac{k}{r}\right\rceil}$ appearing in \cite{AgaSasKum}. In more recent work \cite{RawTamGur_epsilonMSR}, the $\epsilon$-MSR framework was introduced to construct MDS codes that somewhat surprisingly, have sub-packetization $\alpha$ that is {\em logarithmic} in $n$ for a modest increase in repair bandwidth by a multiplicative factor $(1+\epsilon)$.\\

{\bf Piggybacking framework:}  The piggybacking framework \cite{RasShaRam_Piggyback} begins with a collection of $\alpha$ codewords drawn form an MDS code and proceeds to modify the code symbols as described below.
Let $\mathcal{C}$ be an MDS code and let $(f_1(\textbf{u}),f_2(\textbf{u}),\dots,f_n(\textbf{u}))$ represent the codeword  corresponding to message $\textbf{u}$. Next, consider codewords of $\mathcal{C}$ corresponding to $\alpha$ distinct messages, $\textbf{u}_1,\dots,\textbf{u}_\alpha$.  The $\alpha$ code symbols $f_j(\textbf{u}_i)$, $i=1,2,\cdots,\alpha$ are stored on node $j$. 
%
We first modify the code by adding a function $g_{ij}(\textbf{u}_1,\dots,\textbf{u}_{i-1})$ to the $j$-th symbol of $i$-th codeword $f_j(u_i)$, for all $i \in \{2,\dots,\alpha\}, j \in \{1,\dots,n\}$.  The values so added are termed as piggybacks. This modification does not affect our ability to decode the code, if the codewords are decoded in sequence.  Applying an invertible linear transform $T_i$ to the $\alpha$ code symbols in the $i$th node, similarly does not affect our ability to decode the $\alpha$ codewords, nor a node's ability to serve as a helper node.  
By carefully choosing the piggybacking functions and the set $T_i$ of invertible linear transformations it possible to reduce the repair bandwidth for the collective repair of the $\alpha$ MDS codewords in comparison with the repair bandwidth needed for the conventional repair of $\alpha$ MDS codewords. 
Three families of piggybacking-based MDS codes with reduced repair bandwidth and disk read are constructed in \cite{RasShaRam_Piggyback}. The piggybacking framework typically provides savings between $25\%$ to $50\%$ depending up on the parameters and choice of piggybacking functions. For example, Fig.~\ref{PB} shows modification of a $[4,2]$ MDS code with sub-packetization level $2$ in such a way that the systematic nodes can be repaired by reading $3$ symbols (instead of the $4$ symbols required for MDS decoding), resulting in a $25\%$ repair bandwidth and disk read saving.
\begin{figure}[!ht]
	\centering
	\includegraphics[width=6.5in]{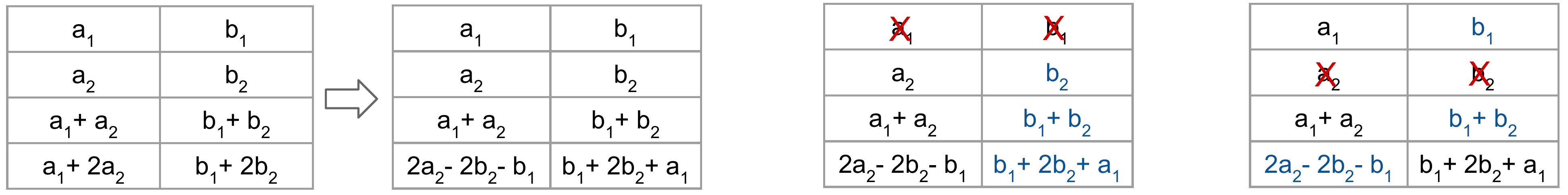}
	\caption{Here two codewords of a [4,2] MDS code are piggybacked. The first systematic node can be repaired by reading $ b_2$, $b_1+b_2$ and $b_1+2b_2+a_1$, whereas the second systematic node repair requires $b_1$, $b_1+b_2$ and $2a_2-2b_2-b_1$}. 
	\label{PB} 
\end{figure}


{\bf $\epsilon$-MSR framework:} The motivation for constructing $\epsilon$-MSR codes \cite{RawTamGur_epsilonMSR} is the larger sub-packetization level of an MSR code, which could possibly prove to be a hurdle in its practical implementation. 
The authors of \cite{RawTamGur_epsilonMSR} provide a generic way  to transform an MSR code into an $\epsilon$-MSR code.
\begin{definition}
	An MDS code $\mathcal{C}$ with sub-packetization $\alpha$ over a finite field $\mathbb{B}$  is said to be an $(n,k,d=n-1,\alpha)_\mathbb{B}$ $\epsilon-$MSR code, $\epsilon > 0$, if for every $i \in [n]$ there exists a linear repair scheme for the code symbol $c_i$ which downloads $\beta_{ij} \le (1+\epsilon)\frac{\ell}{n-k}$ symbols over $\mathbb{B}$ from the $(n-1)$ nodes storing code symbols $c_j$, for $j \in [n]\setminus \{i\}$.
\end{definition}
The construction of an $\epsilon$-MSR code presented in \cite{RawTamGur_epsilonMSR} combines a short block-length MSR code with a code having large minimum distance.
Let $\mathcal{C}_I$ be an $(n=k+r,k,d=n-1,\alpha)_{\mathbb{B}}$ MSR code having parity check matrix,
\begin{equation*}
	H=
	\begin{bmatrix}
		H_{1,1} & H_{1,2} &\dots& H_{1,n} \\
		\vdots & \vdots & \dots & \vdots \\
		H_{r,1} & H_{r,2} &\dots& H_{r,n}
	\end{bmatrix},
\end{equation*}
where the sub-matrices $H_{i,j}$ are of size $(\alpha \times \alpha)$. Next, let $\mathcal{C}_{II}$ be a (not necessarily linear) code having block length $N$, size $M$ and minimum distance $D=\delta N$ over an alphabet $\mathbb{G}$ of size $|\mathbb{G}|\le n$. Let us associate with every codeword $c=(c_1,\dots,c_N)$ of $\mathcal{C}_{II}$, an $(rN\alpha \times N\alpha)$ matrix:
\begin{equation*}
	\mathcal{H}_c=
	\begin{bmatrix}
		u_{1,c}\ \text{Diag}(H_{1,c_1},\dots,H_{1,c_N})\\
		\vdots\\
		u_{r,c}\ \text{Diag}(H_{r,c_1},\dots,H_{r,c_N})\\
	\end{bmatrix},
\end{equation*}
where the $\{u_{i,c}\}$ are non-zero coefficients, drawn from $\mathbb{B}$. Next, using the fact that the number of codewords in $\mathcal{C}_{II}$ is $M$, let us form an $(rN\alpha \times MN\alpha)$ matrix $\mathcal{H}$ with each of the $M$ `thick' columns $\mathcal{H}_c$ corresponding to a different codeword $c \in  \mathcal{C}_{II}$. It can be shown that the code having $\mathcal{H}$ as its parity-check matrix is an $(M,M-r,d=M-1,N\alpha)_{\mathbb{B}}$  $\epsilon$-MSR code, where $\epsilon=(r-1)(1-\delta)$. 
Ensuring this requires judicious selection of the base MSR code $\mathcal{C}_{I}$ as well as the non-zero scalars $\{u_{i,c}\}$. An additional requirement is that for a given $\epsilon>0$, the code $\mathcal{C}_{II}$ should be chosen such that the parameter $\delta$ satisfies $\delta \ge 1-\frac{\epsilon}{r-1}$. The $\epsilon$-MSR codes constructed using this approach can have sub-packetization level scaling logarithmically in the block length. \\
In \cite{RawTamGur_epsilonMSR}, $\epsilon$-MSR codes are constructed by picking the non-optimal-access MSR constructions in \cite{YeBar_1} as $\mathcal{C}_I$. For instance, using $\mathcal{C}_I$ with parameters $(n=3,k=1,d=2,\alpha=2^{3}=8)$ and $\mathcal{C}_{II}$ with parameters $N=20, M=27$ and $D=13$ over $\mathbb{F}_3$ one can construct a $(M=27,M-r=25,M-1=26,N\alpha=160)$ $\epsilon$-MSR code. Note that the MSR code $\mathcal{C}_I$ with parameters $(n=27,k=25,d=26)$ requires a sub-packetization level of $2^{27}$, whereas this $\epsilon$-MSR code has sub-packetization level of 160 $(\ll 2^{27})$ and repair bandwidth is within $1.35$ times that of the MSR code.  
\subsection{Fractional Repetition Codes}
\begin{figure}[ht!]
	\centering
	\includegraphics[width=1.8in]{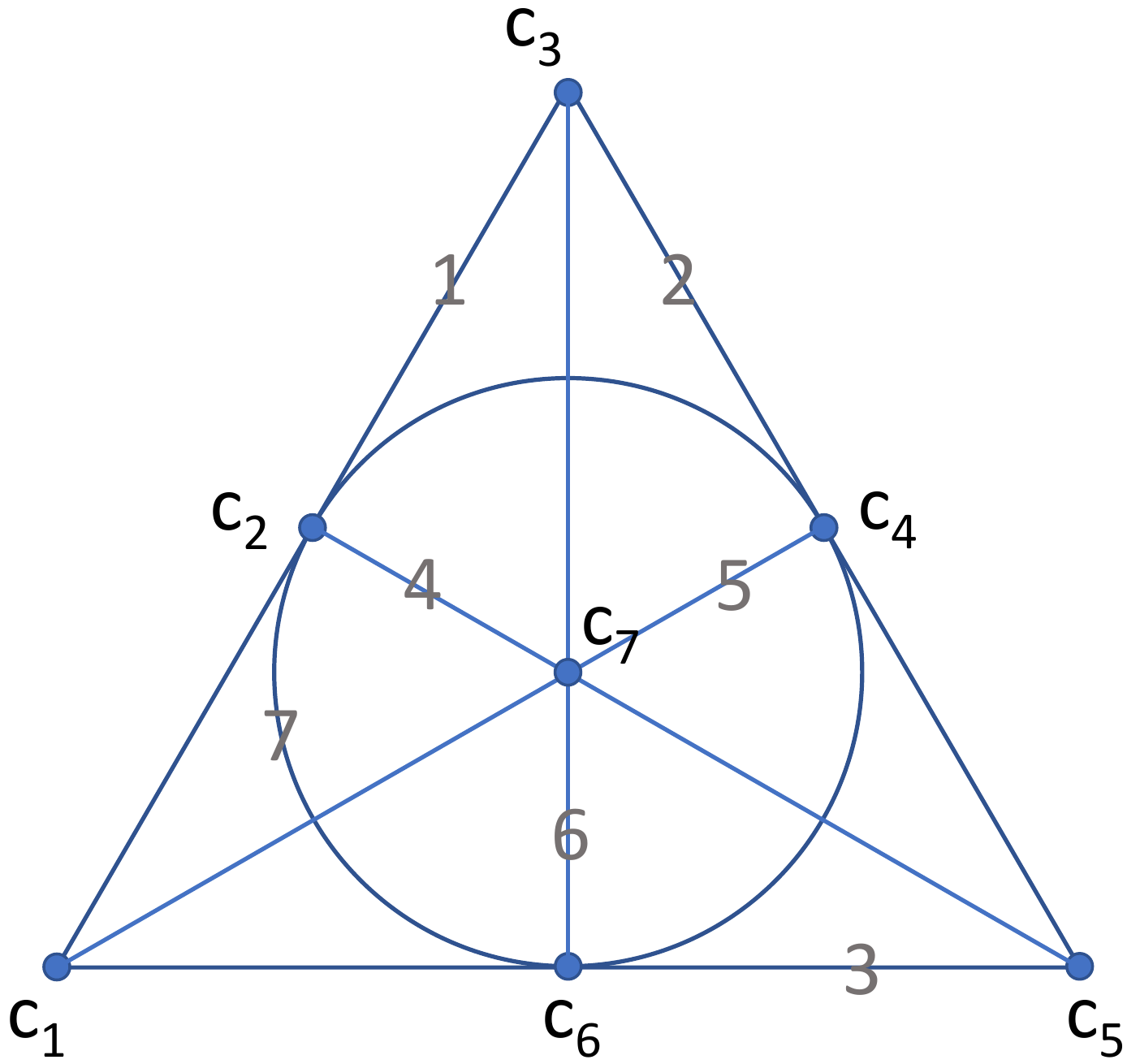}
	\caption{Each of the seven lines in the Fano plane indicates a node and points within a line denote the code symbols stored in the corresponding node. For instance, $N_1=\{c_1,c_2,c_3\}$.   }\label{fig:fano_plane}  
\end{figure}
Fractional Repair (FR) Codes, introduced in \cite{RouRamFR10}, are regarded as codes that generalize the RBT MBR construction in \cite{RasShaKumRam_allerton09}.   An FR code is associated with the parameter set $\{n,k,\alpha,\rho\}$, where $n$ is the number of nodes and $k$ is the smallest number such that one can retrieve the entire data file from connecting to any set of that many nodes. Let $K$ be the file size of the FR code. To encode and store data, an FR code begins by encoding a collection $\{u_1,\cdots, u_K\}$ of message symbols drawn from a finite field $\mathbb{F}_q$ using a scalar $[N,K]$ MDS code ${\cal A}$, also referred to as the DRESS code in \cite{PawNoorRouRamDressCode11}.  Let $(v_1,v_2,\cdots,v_N)$ denote the symbols of a codeword in ${\cal A}$.  Each of the $N$ scalar code symbols is replicated $\rho$ times and the resultant $\rho N$ symbols are stored across the $n$ nodes in such a way that there are $\alpha$  symbols per node and each code symbol is present in precisely $\rho$ distinct nodes. Combinatorial techniques such as $t$-designs are used to make such an assignment possible. For this to happen, we must have that $n\alpha= N\rho$. 
In order to be able to recover the entire data file by connecting to any $k$ nodes we must clearly have that: $R_\mathcal{C}(k)\triangleq\min_{{J\subseteq [n]}:|J|=k}|\cup_{j\in J} N_j|  \geq   K$,
where $N_j$ indicates the set of $\alpha$ code symbols stored in $j^\text{th}$ node, $j\in[n]$. Note that $R_\mathcal{C}(k)$ is defined with respect to a given collection $\{N_j\}_{j=1}^n$. Let $C_{FR}(n,k,\alpha,\rho)$ denote the maximum  $R_\mathcal{C}(k)$ possible across all possibilities of $\{N_j\}_{j=1}^n$, which conform to the parameters $n,\alpha$ and $\rho$. Hence an FR code is said to be {\it $k$-optimal} \cite{SilbEtzOptimalFR15}, if it satisfies: $
K=C_{FR}(n,k,\alpha,\rho)$.

In contrast to an MBR code, an FR code requires the existence of just a single set of $d=\alpha$ helper nodes to perform RBT. However it follows naturally from the $\rho$-replication of code symbols that such a set of $d$ helper nodes is available, even in the presence of $(\rho-1)$ node failures. 

\begin{eg}[\cite{RouRamFR10}]
	Consider an FR code $\mathcal{C}$ with parameters $n=7, k=3, d=3, \rho=3$. The code is described using the Fano plane as shown in Figure \ref{fig:fano_plane}. Here $R_\mathcal{C}(k)=6$. By choosing the outer MDS code to be the $[7,6]$ single parity check code, data collection property follows. As each symbol is shared by three lines, $\rho=3$ and hence $\mathcal{C}$ permits RBT up to $2$ node failures.
\end{eg}

The following bound on the maximum rate of an FR code with parameters $(n,k,\alpha,\rho)$, is derived in \cite{RouRamFR10}.	
\bean
C_{FR}(n,k,\alpha,\rho) & \leq & \min \left\{ \bigg\lceil\frac{n\alpha}{\rho}\bigg(1-\frac{\binom{n-\rho}{k}}{\binom{n}{k}}\bigg)\bigg\rfloor, \ g(n,k,\alpha,\rho) \right\},\\
\text{ where } \ g(n,1,\alpha,\rho)=1,  &  \text{ and } &  g(n,k+1,\alpha,\rho)=g(n,k,\alpha,\rho)+\alpha-\bigg\lceil\frac{\rho g(n,k,\alpha,\rho)-k\alpha}{n-k}\bigg\rceil.
\eean

The paper \cite{SilbEtzOptimalFR15} considers FR codes with parameters $\alpha \geq k$, $\beta=1$ and provides several $k$-optimal constructions. The work \cite{OlmRamFRCodes16} considers FR codes with parameter $\beta\geq 1$ and also introduces a certain notion of locally recoverable FR codes where the parameter $\alpha<k$. In  \cite{KooGillScaleFR11}, the authors study FR codes that have $\alpha$ much larger than replication degree, $\rho$. In \cite{ExistenceFR}, the authors identify necessary and sufficient conditions for the existence of FR codes.

\subsection{Secure Regenerating Codes} 

Three secrecy models in the context of a regenerating code are introduced in \cite{PawarRouayRam}: (a) a passive eavesdropper model, where the eavesdropper can read the contents of any $\ell$ nodes but cannot modify the content of these nodes, (b) an active omniscient adversary model, where the adversary can read the content of $\ell=k$ nodes and can also modify the content of $b$ nodes where $2b \le k$ and (c) an active limited-knowledge adversary model, where the adversary can read the content of $\ell < k$ nodes and can modify the content of $b \le \ell$ nodes. In the case of a passive eavesdropper, the {\em secrecy capacity} ($B_s$)
is the maximum amount of information that can be stored without any information being revealed to the eavesdropper.  In the active eavesdropper model, the {\em resiliency capacity} ($B_r$) is the maximum amount of information that can be stored such that it can be reliably made available to a legitimate data collector, in spite of the tampering on the data in $b$ nodes done by the eavesdropper. In \cite{PawarRouayRam}, the following upper bound on secrecy capacity of the passive eavesdropper model was derived:
\bea
\label{eq:securerc_passive} 
B_s(\alpha, \gamma = d\beta) \le \sum\limits_{i=\ell+1}^k \min \{(d-i+1)\beta, \alpha\}.
\eea If $\alpha$ is not constrained, then the resultant {\em bandwidth-limited secrecy capacity} $B_{s,\text{BL}}$ becomes a function of $(k, d,\beta)$ alone. The value of $B_{s,\text{BL}}$ is determined \cite{PawarRouayRam} for $d=(n-1)$ by providing a bound and an optimal construction. It was also shown that the resiliency capacity satisfies $B_r(\alpha, \gamma) \le \sum\limits_{i=i_0}^k \min \{(d-i+1)\beta, \alpha\}$,
where $i_0$ is equal to $2b+1$ for omniscient case and $b+1$ for the limited knowledge case. 

In an alternate setting, Rashmi et al. in \cite{RashmiShahRamKum_errors} assume a noisy channel for transmission of data during repair and reconstruction, and introduce the notion of an $(s,t)$-resilient regenerating code that can correct up to $t$ errors and $s$ errors during both repair and reconstruction. The model is aligned with the active eavesdropper model where the eavesdropper can tamper the contents of $b$ nodes. An $(s,t)$-resilient regenerating code is shown to satisfy $B \le \sum\limits_{i=1}^{k} \min \{(d-i+1)\beta, \alpha \}$ where, $d = \Delta - 2t-s$, $k = \kappa - 2t-s$ and $\Delta, \kappa$ are the number of nodes contacted during repair and reconstruction respectively.  Constructions of MSR and MBR codes that are $(s,t)$ resilient are also provided in \cite{RashmiShahRamKum_errors}. In \cite{YeBar_1}, the authors extend this model to the repair of multiple nodes and provide MSR constructions that are resilient to $t$ errors during repair. In \cite{ShahRashKum}, the authors extend the passive eavesdropper model to the setting where out of the $\ell$ nodes accessed, the eavesdropper can read the contents of $\ell_1$ nodes and can observe the information passed on for the repair of $\ell_2 = \ell - \ell_1$ nodes. The upper bound in \eqref{eq:securerc_passive} also holds for this extended case. In the case of an MBR code, since the amount of data stored equals the amount of data received for node repair, the breakup between $\ell_1,\ell_2$ is immaterial.   

However in the case of an MSR code, $d\beta > \alpha$. In \cite{ShahRashKum}, the authors provide explicit, secure MBR, and low-rate MSR code constructions that achieve the upper bound \eqref{eq:securerc_passive} for $\ell_2=0$. The secure MSR construction from \cite{ShahRashKum} provides a lower bound to the secure file size of an MSR code: $B_s \ge  (k-\ell)(\alpha - \ell_2 \beta)$ for $\ell_2 > 0$.

The upper bound on secure MSR file size $B_s \le (k-\ell)\alpha$ given by \eqref{eq:securerc_passive} is improved in \cite{TandonAmuruClancyBuehrer,RawatKoyluSilberVish,GopaRouayCladerPoor,HuangParamXian}. In \cite{Rawat_secrecy}, Rawat established that the secrecy capacity of an MSR codes is given by $B_s = (k-\ell)(1-\frac{1}{n-k})^{\ell_2}\alpha$ by providing an MSR construction. An upper bound that matches with Rawat's construction is proved by Goparaju et.al in \cite{GopaRouayCladerPoor} under the constraint of linearity. In \cite{KadheSprinston}, secure MSR codes with smaller field sizes for all parameters were constructed. In \cite{YeShumYeung,ShaoLiuTianShen} the ER tradeoff is studied for secure regenerating codes. 




\section{Locally Recoverable Codes} \label{sec:lrc} 

The earliest-known appearance of locally recoverable codes can be found in \cite{HanMon,HuaChenLi}. A construction for a code with locality appears in \cite{OggDat}. A formal treatment of codes with locality with a bound on minimum distance (discussed below) appears in \cite{GopHuaSimYek}. The extension to the non-linear case for all-symbol and information-symbol locality appear in \cite{PapDim} and \cite{ForYek} respectively.

Let \calc\ be an $[n,k]$ linear code over $\fq$. For a subset $S \subseteq [n]$, we use $\calc |_S$ to denote the {\em restriction} of \calc\ to the coordinates in $S$.
Let $G$ be a $(k \times n)$ generator matrix for \calc\ having columns $\{\ug_i\}_{i=1}^n$, i.e.,
$G = [\ug_1, \ug_2, \cdots, \ug_n]$. 
An information set $E=\{e_1, e_2, \cdots, e_k\}$ is any subset of $[n]$ of size $k$ satisfying:
$\text{rk}(G|_E) = \text{rk}[\ug_{e_1}, \cdots, \ug_{e_k}] = k$.
An $[n,k]$ code \calc\  is said to have $(r, \delta)$ information-symbol (IS) locality if there is an information set $E = \{e_1, e_2, \cdots, e_k\}$ such that for every $e_i \in E$, there exists a subset $S_i \subseteq [n]$, with $e_i \in S_i$, with 
\bea
\label{eq:rdeltaloc}
\dim(\calc |_{S_i} ) \le r, \ d_{\min}(\calc |_{S_i} ) \ge \delta. 
\eea
$\calc$ is said to have $(r,\delta)$ all-symbol (AS) locality if for every coordinate $i \in [n]$, there exists a subset $S_i \subseteq [n]$ with $i \in S_i$, such that \eqref{eq:rdeltaloc} holds. Clearly, a code with AS locality also possesses IS locality.

\subsection{Bound on Minimum Distance} 

A major result in the theory of LR codes is the minimum distance bound derived in \cite{GopHuaSimYek}, which in the context of the theorem below, was derived for $\delta=2$. An analogous  proof for $\delta=2$ and nonlinear codes can be found in \cite{PapDim},\cite{ForYek}. The bound in \cite{GopHuaSimYek} was extended adopting the same approach as in \cite{GopHuaSimYek}, to the general case $\delta >2$ in \cite{PraKamLalKum} and appears in Theorem~\ref{thm:lrcdminbound} below.   The extension to codes over a vector alphabet can be found in \cite{SilRawKoyVis}.  

\bthm \cite{PraKamLalKum} \label{thm:lrcdminbound} 
Let \calc\ be an $[n,k]$ linear code over \fq\ having $(r,\delta)$ IS locality. Then
\bea
\label{eq:lrcdminbound}
d_{\min} \le (n-k+1) - (\left\lceil \frac{k}{r} \right\rceil - 1) (\delta -1). \label{Singleton}
\eea
\ethm
Our proof will make use of the following lemma.

\blem \label{lem:n_minus_s}
Let \calc\ be an $[n,k]$ code and let $S \subseteq [n]$ such that $\text{rk}(G|_S) \le k-1$. Then $d_{\min}(\calc) \le n - |S|$.
\elem
\bpf
Since rk$(G|_S) \le k-1$, it follows that there exists a nonzero message vector \uu\ such that $\uu^T G|_S = 0$. Let $\uc = \uu^T G$, then $0 < \text{wt} (c) \le n - |S|$ and the result follows.
\epf

\bpf (of Theorem~\ref{thm:lrcdminbound})
Let $E = \{e_1, e_2, \cdots, e_k\}$ be the information set with respect to which $\calc$ has information-symbol locality. Let the subsets $S_i \subseteq [n]$, $1 \le i \le k$, be such that $e_i \in S_i$, $\calc|_{S_i}$ is an $(r,\delta)$ code, i.e, $\dim(\calc|_{S_i}) \le r$, $d_{\min}(\calc|_{S_i}) \ge \delta$.  Let $V_i$ denote the column space of $G|_{S_i}$.  Next, over the course of several iterations, we incrementally build up a set $S$, beginning with $S = \phi$.  We use $j$ to indicate the iteration number and begin with $j=1$.  On the $j$-th iteration, $j \geq 1$, we first search for an index $i$ such that $V_i  \not \subset \text{Col}(G\mid_{S})$ ($\text{Col}(A)$ refers to the column space of $A$).  This will always be possible, as we always ensure rk$(G|_{S}) \leq k-1$. Having found such an index $i$, we next examine the rk$(G|_{S \cup S_i})$.  If rk$(G|_{S \cup S_i}) \leq k-2$, we set 
\bea \label{eq:normal_iteration} 
a_j = |S \cup S_i| - |S|, \ \ \gamma_j = \text{rk}(G|_{S\cup S_i}) - \text{rk}(G|_S), \ S = S \cup S_i, \ j=j+1, 
\eea
in order from left to right, and repeat the procedure in $(j+1)$th iteration by searching for an index $i$ such that $V_i  \not \subset \text{Col}(G\mid_{S})$.  
If at the $j$-th iteration, for any $j$, we find that
\begin{enumerate}
	\item[]  Case (i): rk$(G|_{S \cup S_i}) = k-1$, we then replace the procedure in \eqref{eq:normal_iteration} with the steps below: 
	\bean
	a_j = |S \cup S_i| - |S|, \ \ \gamma_j = \text{rk}(G|_{S\cup S_i}) - \text{rk}(G|_S), \ S = S \cup S_i, \ m=j, 
	\eean
	and terminate the program. 
	\item[] Case (ii):  rk$(G|_{S \cup S_i}) = k$. In this case, we replace the procedure in \eqref{eq:normal_iteration} by selecting a subset $T_i \subseteq S_i$ such that rk$(G|_{S \cup T_i}) = k-1$ (this can always be done), and then setting
	\bean
	a_j = |S \cup T_i| - |S|, \ \ \gamma_j = \text{rk}(G|_{S\cup T_i}) - \text{rk}(G|_S), \ S = S \cup T_i, \ m=j, 
	\eean
	and then terminating the program. 
	\een
	Thus $m$ indicates the number of iterations that took place before the program was terminated.
	Note that since for every $i$, rk$(G|_{S_i}) \le r$, we have that $\gamma_j \le r$. Let $j \geq 1$. At the $j$-th iteration, let $i$ be the index chosen such that  $V_i  \not \subset \text{Col}(G\mid_{S})$ and Let $R_i \subseteq S_i \setminus S$ be such that $|R_i|=\gamma_j-1$ and rk$(G|_{R_i})=\gamma_j-1$. Since the code having generator matrix $G|_{S_i}$ has minimum distance $\ge \delta$ and since rk$(G|_{(S \cap S_i) \cup R_i}) \leq r-1$, by Lemma \ref{lem:n_minus_s}, $\delta \leq |S_i|-|(S \cap S_i) \cup R_i|$ = $ |S_i|-|(S \cap S_i)|- | R_i| $ = $|S_i \setminus S|- (\gamma_j-1) $. It follows from this that $a_j \ge \gamma_j + (\delta -1)$.
	
	\begin{enumerate}
		\item[$\bullet$] Algorithm terminates under Case(i): Since the incremental rank is at most $r$, it follows that the number of iterations $m$ satisfies $m \ge \left\lceil \frac{k-1}{r} \right\rceil$. We thus have
		\bean
		|S| = \sum\limits_{j=1}^m a_j \ge \sum\limits_{j=1}^m \left( \gamma_j + \delta-1 \right)  = (k-1)+ (\delta-1) m \ge (k-1) + \left\lceil \frac{k-1}{r} \right\rceil (\delta-1).
		\eean
		\item[$\bullet$] Algorithm terminates under Case(ii): $\text{Arguing similarly, we have that $m \geq \lceil \frac{k}{r} \rceil$ and }$
		\bean
		|S| &=& \sum\limits_{j=1}^m a_i \ge \sum\limits_{j=1}^{m-1} (\gamma_j + \delta-1) + \gamma_m\\
		&=& (k-1)+(\delta-1) (m-1) \ge (k-1) + \left(\left\lceil \frac{k}{r} \right\rceil -1\right) (\delta-1).
		\eean
		\een
		Case (ii) leads to a smaller lower bound on $|S|$.  Hence from Lemma~\ref{lem:n_minus_s}  it follows that
		\bean
		d_{\min} 
		& \leq & (n-k+1) - \left(\left\lceil \frac{k}{r} \right\rceil -1\right) (\delta-1).
		\eean
		\epf
		We note the following: 
		\begin{enumerate}
			\item Setting $\delta=1$ (i.e., no locality constraint) in \eqref{Singleton}, one recovers the classical Singleton bound.  For this reason, the bound in \eqref{Singleton} is commonly referred to in the context of locality as the Singleton bound. 
			\item The Pyramid-Code Construction in \ref{sec:pyramid} provides a general construction of codes with IS locality that achieves the Singleton bound for all parameters $(n,k,r,\delta)$.  
			\item For many parameter sets, one can construct codes with AS locality that achieve the bound in \eqref{Singleton}, include all cases where $(r+1)|n$, see Section~\ref{sec:TB} below.   
			\item For $\delta=2$, bounds for AS locality that are tighter than the Singleton bound for IS locality appearing in \eqref{Singleton}, can be found in \cite{PraLalKum,WanZhaI,ZhaWanGe,MehArd}.  Constructions for codes achieving the tightened bound in \cite{WanZhaI} for the case of $n_1 > n_2$ where $n_1 = \lceil \frac{n}{r+1} \rceil$, $n_2=n_1(r+1)-n$ and having exponential field size can also be found there. 
			\item It is shown in \cite{TamBar_LRC} that one can construct codes with AS locality and field size of order $n$ whose minimum distance is within $1$ of the bound in \eqref{Singleton} provided $ r \nmid k$, $n \neq 1 \pmod{r+1}$.  In \cite{ErnWesHol}, it is shown that this can be achieved for any parameter set if one permits the field size to be exponential in $n$.  
		\end{enumerate}

		\subsection{Constructions} 
		
		\subsubsection{Pyramid Code Construction} \label{sec:pyramid} 
		
		The pyramid code construction technique which appeared in \cite{HuaChenLi}, allows us to construct for any given parameter set $\{n, k, r, \delta\}$ a code with $(r,\delta)$ IS locality achieving the $d_{\min}$ bound in \eqref{eq:lrcdminbound}. We sketch the construction for the case $k=2r$. The general case $k=ar$, $a>2$ or even when $r \nmid k $, follows along similar lines. The construction begins with the systematic generator matrix \gmds\ of an $[n_1,k]$ scalar MDS code $\calc_{\text{\tiny MDS}}$ having block length $n_1 = n-(\delta-1)$.  It then reorganizes the sub-matrices of \gmds\ to create the generator matrix \gpyr\ of the pyramid code: 
		\bean
		\gmds\  = \left[ \begin{array}{cccc}
			I_r & & P_1 & Q_1 \\ 
			& I_r & \underbrace{P_2}_{ (r \times (\delta-1))} & \underbrace{Q_2}_{(r \times s)}
		\end{array} \right]
		& \Rightarrow & 
		G_{\text{\tiny PYR}} = \left[\begin{array}{ccccc}
			I_r & P_1 & & & Q_1\\ & & I_r & P_2 & Q_2
		\end{array}\right].
		\eean
		where $s = n_1 - 2r - (\delta-1)$. 
		It is not hard to show that the $[n,k]$ code $\calc_{\text{\tiny PYR}}$ generated by $G_{\text{\tiny PYR}}$ has $(r, \delta)$ IS locality 
		and that $d_{\min}(\calc_{\text{\tiny PYR}}) \ \ge \ d_{\min}(\calc_{\text{\tiny  MDS}})$. 
		It follows that
		$d_{\min}(\calc_{\text{\tiny PYR}})  \ge  d_{\min}(\calc_{\text{\tiny  MDS}}) = n_1 - k+1 \ = \ (n-k+1)-(\delta-1)$, 
		and the code $\calc_{\text{PYR}}$ is thus optimal w.r.t the $d_{\min}$ bound in \eqref{eq:lrcdminbound}.
		
		\subsubsection{The Tamo-Barg Construction } \label{sec:TB}
		
		\begin{figure}[ht!]
			\centering
			\includegraphics[width=2.5in, height=1.5in]{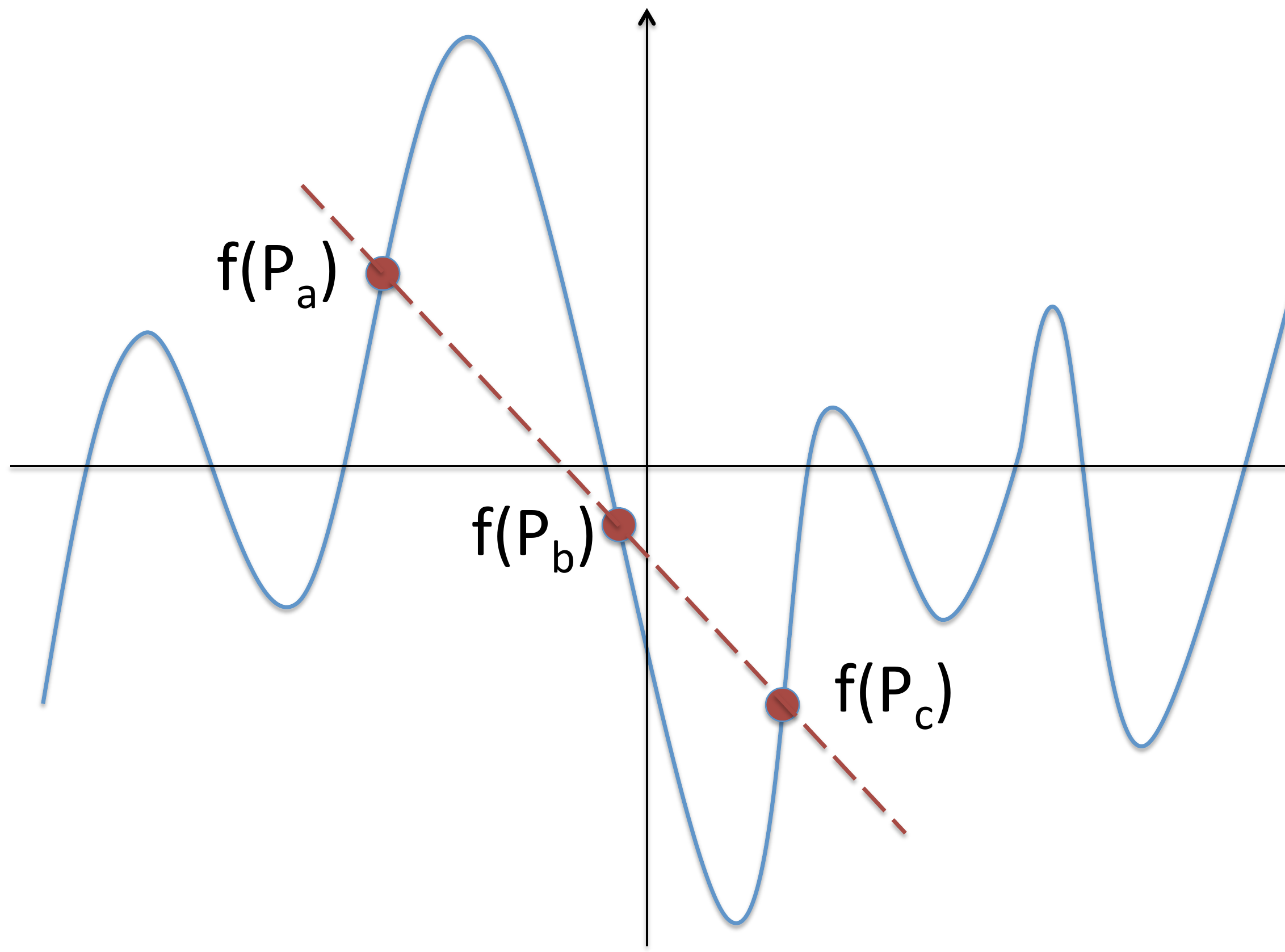}
			\caption{In the T-B construction, code symbols in the local codes of length $(r+1)$ correspond to evaluations of polynomials of degree $\leq (r-1)$. Here,  $r=2$ implying evaluation at $3$ points of a linear polynomial.  }  
		\end{figure} 
		
		The construction below by Tamo and Barg~\cite{TamBar_LRC}, provides a construction for LR codes with AS locality.  While for simplicity, we present the construction for the case $\delta=2$, the construction has a natural extension to the general case $\delta>2$ (see \cite{TamBar_LRC}).  We will refer to the construction in the sequel as the Tamo-Barg (T-B) construction. 
		\begin{thm} Let \fq\ be a finite field of size $q$, let $r \geq 2$, $n =m(r+1) \leq q$, with $ m \geq 2$ and $ 2 \leq k \leq (n-1)$.  Set $k=ar+b, 0 \leq b \leq (r-1)$.  Let $A=\{\theta_1,\theta_2,\cdots,\theta_n\} \subseteq \mathbb{F}_q$  and $A_i \subset A, 1 \leq i \leq m$, $|A_i|=(r+1), A_i \cap A_j = \phi, i \neq j$ represent a partitioning $A = \cup_{i=1}^m A_i$ of $A$.  
			Let $g(x)$ be a `good' polynomial, by which is meant, a polynomial over \fq\ that is constant on each $A_i$ and of degree $(r+1)$.  
			Let 
			\bea 
			f(x) & = & \sum_{j=0}^{a-1} \sum_{i=0}^{r-1} a_{ij}[g(x)]^jx^i \ + \   \sum_{j=a} \sum_{i=0}^{b-1} a_{ij}[g(x)]^jx^i   , \label{eq:poly_TB} 
			\eea
			where the $a_{ij} \in \fq$ are the message symbols and where the second term is vacuous for $b=0$, i.e., when $r \mid k$.  
			Consider the code ${\cal C}$ of block length $n$ and dimension $k$ where the code symbols are obtained through evaluation of the above collection of polynomials at the elements in $A$.    Then \calc\ is an $(r,\delta)$ AS locality code with $\delta=2$ and is optimal with respect to the $d_{\min}$ bound in \eqref{Singleton}. The $i$-th local code has support set $A_i$. 
		\end{thm}
		\begin{proof}
			In \eqref{eq:poly_TB} it can be checked that by varying $\{a_{ij}\}$, one obtains a collection of $k$ linearly independent polynomials and since $k <n$, it follows that the code has dimension $k$.  Let $g(\theta)=\gamma_{\ell}$, all $\theta \in A_{\ell}$.  Next, let $\theta \in A_{\ell}$. Then we have
			\bean
			f(x)|_{\theta \in A_{\ell}} & = & \sum_{j=0}^{a-1} \sum_{i=0}^{r-1} a_{ij}[\gamma_{\ell}]^jx^i \ + \   \sum_{j=a} \sum_{i=0}^{b-1} a_{ij}[\gamma_{\ell}]^jx^i ,
			\eean
			which is a polynomial of degree $\leq (r-1)$ and hence the corresponding evaluation code, when restricted to $A_i$ has $d_{\min}\geq 2$, leading to the desired locality and ability to recover from a single erasure.   To determine $d_{\min}$, assume $b \geq 1$.  
			The maximum degree of a polynomial $f(x)$ then equals 
			\bean
			a(r+1)+b-1 & = & (ar+b) + (a-1) \ = \  k + \lceil \frac{k}{r} \rceil -2 .
			\eean
			When $b=0$ and hence $k=ar$, the maximum degree equals: 
			\bean
			(a-1)(r+1) + (r-1) & = & (ar) + (a-2)  \ = \  k + \lceil \frac{k}{r} \rceil -2 .
			\eean
			It follows that the code is optimal as 
			$
			d_{\min} \ \geq \  (n-k+1) -( \lceil \frac{k}{r} \rceil -1).
			$
		\end{proof} 
		
		An example of how good polynomials may be constructed is given below, corresponding 
		to the annihilator polynomial of a multiplicative subgroup $G$ of $\mathbb{F}_q^*$.    
		\begin{example}
			Let $H < G \le \fqstar$ be a chain of cyclic subgroups, where $|H|=(r+1), |G|=n$ so that $(r+1)$  $|$  $n$  $|$ $(q-1)$.  Let $n=(r+1)t$.  Let $\{A_i = \gamma_i H \mid i \in \{1,2,\cdots, t\}\}$ be the $t$ multiplicative cosets of $H$ in $G$, with $\gamma_1$ being the multiplicative identity so that $A_1=H$.  It follows that 
			\bean
			\prod_{\beta \in A_i} (x-\beta) \ = \ x^{r+1}-\gamma_i^{r+1},
			\eean
			so that $x^{r+1}$ is constant on all the cosets of $H$ in $G$ and may be selected as the good polynomial $g(x)$ i.e., $g(x)=x^{r+1}$ is one possible choice of good polynomial based on multiplicative group $H$. 
		\end{example}
		Further examples may be found in \cite{TamBar_LRC,Sihem_Mesnager,KolBarTamYad}. For constructions meeting the Singleton bound with field size of $O(n)$ and more flexible value of $r$, please see \cite{LinLimCha}. A construction of LR codes achieving the Singleton bound with field size of $O(n)$ closely resembling construction based on parity splitting for a restricted set of parameters can be found in \cite{BalKUmMax}.

		\subsection{Alphabet-Size Dependent Bounds on Code Rate} 
		
		\subsubsection{General Bound} The bound in Theorem~\ref{thm:lrcdminbound} as well as the bounds for non-linear and vector codes derived in \cite{PapDim,SilRawKoyVis} hold regardless of the size $q$ of the underlying finite field.  The theorem below takes the size $q$ of the code symbol alphabet into account and provides a tighter upper bound on the dimension of a code with locality that is valid even for nonlinear codes. The `dimension' of a nonlinear code \calc\ over an alphabet $\mathbb{Q}$ of size $q=|\mathbb{Q}|$ is defined to be the quantity $k=\log_q(| \calc |)$.
		
		\begin{thm} \label{thm:CadMaz}
			\cite{CadMaz} For any $(n,k,d)$ code  \calc\ that is an LR code with parameter $r$ over an alphabet $\mathbb{Q}$ of size $q=|\mathbb{Q}|$,
			\bea
			k \  \leq \  \min_{t \in \mathbb{Z}_{+}} [ tr+ k_{\text{opt}}^{(q)}(n-t(r+1),d) ],  \label{CadMaz}
			\eea
			where $k_{\text{opt}}^{(q)}(n-t(r+1),d)$ is the largest possible dimension of a code over $\mathbb{Q}$ having block length $(n-t(r+1))$ and minimum distance $d$.
		\end{thm}
		\bpf (Sketch of proof) 
		The bound holds for linear as well as nonlinear codes. In the linear case, with $\mathbb{Q}=\mathbb{F}_q$, the derivation proceeds as follows.  Let $G$ be a $(k \times n)$ generator matrix of the locally recoverable code ${\cal C}$.  Then it can be shown that for any integer $t > 0$,  there exists an index set ${\cal I}$ such that $|{\cal I}|=\min(t(r+1),n)$ and $\text{rank} \left( G \mid_{\cal I} \right) \ = \ s \leq tr$. This implies that \calc\ has a generator matrix of the form (after permutation of columns): 
		\bean
		G & = & 
		\left[ \begin{array}{cc} 
			\underbrace{A}_{(s \times |{\cal I}|)} & B \\ 
			\left[ 0 \right]  & D
		\end{array} \right] . 
		\eean
		In turn, this implies that the rowspace of $D$ defines an $[n-t(r+1),k-s \geq k-tr,d]$ code over $\mathbb{F}_q$, if $k-tr > 0$.  
		It follows that $k \leq tr + k_{\text{opt}}^{(q)}(n-t(r+1),d)$ and the result follows.  Note that the row space of $D$ corresponds to a shortening $\calc^{S}$ of \calc\ with respect to the coordinates ${\cal I} \subseteq [n]$.  The proof in the general case is a (nontrivial) extension to the nonlinear setting.  \epf 
		
		\begin{note}
			The above bound was obtained by showing that shortening of an $[n,k,d]$ LR code with parameter $r$, leads to an $[n-t(r+1),\geq k-tr,d]$ code.  Classical bounds on coding theory can be applied to this shortened code, to yield ``lifted'' bounds on the parent code having locality. This {\em shortening} approach, presented for the first time in \cite{CadMaz}, has since been employed in subsequent papers in the literature, see \cite{HuaYaaUchSie,BalKum}. 
		\end{note}
		
		\noindent An alphabet-size-dependent bound on {\em $d_{\min}$ } (based on the shortening approach in \cite{CadMaz}), and which uses upper bounds on generalized Hamming weights \cite{Wei} of the dual code derived in \cite{PraLalKum},  appears in \cite{BalKum}.  The approach in \cite{BalKum}, can also be used to derive the following upper bound on dimension which is in general tighter than \eqref{CadMaz}: 
		\bea
		k \  \leq \  \min_{\{i: e_i < n-d+1 \}} [ e_i-i + k_{\text{opt}}^{(q)}(n-e_i,d) ] . \label{GHW_Bound}
		\eea
		The integers $\{e_i\}_i$ appearing here can be recursively computed for a given $(n,r)$, and represent upper bounds on the generalized Hamming weights (GHW) of the dual code (see Section~\ref{sec:available_dmin}). 
		A bound on the dimension of a binary LR code for a given $(n,r,d_{min})$ based on the Hamming bound for $d_{min} \geq 5$ and $2 \leq r \leq \frac{n}{2}-2$ appears in \cite{WanZhaLin}. This bound is shown to be tighter than \eqref{CadMaz} for some cases including $5 \leq d_{min} \leq 8$ for $n$ large.  
		In \cite{HuaYaaUchSie}, the authors employ the shortening approach to derive an alphabet-size-dependent bound on the minimum distance and dimension of codes having IS locality.  An example comparison of the bounds on dimension for linear LR codes in \eqref{GHW_Bound}, \eqref{CadMaz} and the Hamming-bound based bound in \cite{WanZhaLin} is presented in Table~\ref{tab:bd_compare}.
		
		\begin{table}[ht!]
			\caption {A comparison of upper bounds on the dimension $k$ of binary LR code, for given $(n,d_{\min},r,q)$} \label{tab:bd_compare}
			\begin{center}
				\begin{tabular}{ | c | c | c | c | c| c | } 
					\hline
					\multicolumn{6}{|c|}{$n=31$, $q=2$, $d_{min}=5$}   \\
					\hline
					$r$ (locality) & 2 & 3 & 4 & 5 & 6 \\ 
					\hline
					Bound \eqref{CadMaz} & 17 & 19 & 20 & 20 & 20 \\ 
					\hline
					Bound in \cite{WanZhaLin} & 15 & 18 & 20 & 22 & 23 \\
					\hline
					Bound \eqref{GHW_Bound} & 16 & 18 & 19 & 20 & 20 \\
					\hline
				\end{tabular}
			\end{center}
		\end{table} 
		
		\subsubsection{Bounds with Disjoint Repair Groups}
		
		Bounds on the dimension of a binary LR code $\calc$ for a given $n,r,d_{min}$ under the assumption that the local codes ($\calc|_{S_i}$) have pairwise disjoint support  appear in \cite{WanZhaLin,MaGe,AgaBargHuMazTam}. The bound in \cite{WanZhaLin} make use of the Hamming bound and is shown to be tighter than \eqref{CadMaz} for some cases. A tightening of this bound appears in \cite{MaGe}.   The tightest known bounds for this setting appear in \cite{AgaBargHuMazTam} and are based on Linear Programming.
		
		\subsubsection{Bounds on the Dimension of Cyclic LR Code} \label{sec:cyclic_code_bound}
		
		A linear-programming-based upper bound on the dimension of cyclic LR codes appears in \cite{TambarGopCal}.   Other bounds can be found in \cite{GopCal,ZehYak}.
		
		
		\subsubsection{Asymptotic Bounds}
		
		Upper bounds on asymptotic rate $R^q(r,1,\Delta)$ (see Section \ref{sec:asymptotic_rate} for a definition) for a given fractional minimum distance of a binary LR code appear in \cite{AgaBargHuMazTam}, that represent a slight tightening of the asymptotic version of the bound in \eqref{CadMaz}. 
		An achievable asymptotic Gilbert-Varshamov type lower bound for LR code appear in \cite{TamBarFro} to be:
		\bea
		R^q(r,1,\Delta) \geq 1- \min_{0 <s \leq 1} (\frac{1}{r+1} \log_q((1+(q-1)s)^{r+1}+(q-1)(1-s)^{r+1} )-\Delta \log_q(s)). \label{GV}
		\eea
		Constructions achieving the lower bound \eqref{GV} can also be found in \cite{CadMaz}.
		An improved lower bound obtained via a construction that makes use of algebraic-geometric codes based on the Garcia-Stichtenoth curves  appear in \cite{BarTamVla}:
		\bean
		R^q(r,1,\Delta) \geq \frac{r}{r+1} (1-\Delta-\frac{\sqrt{q}+r}{q-1}) \text{ for } (r+1) | (\sqrt{q}+1).
		\eean
		Constructions based on algebraic geometry and covering a wider range of parameters can be found in \cite{LiMaXin}.
		The algebraic-geometry-based constructions improve upon the GV-type bound in \eqref{GV} for some selected range of parameters.
		
		\subsection{Small-Alphabet Constructions}  
		
		\subsubsection{Construction of Binary Codes} 
		
		Constructions for binary codes that achieve the bound on dimension given in \eqref{CadMaz} for binary codes, appear in \cite{NamSo,SilZeh,HaoXiaChe}. While \cite{HaoXiaChe} and \cite{NamSo} provide constructions for $d_{min}=4$ and $d_{\min}=6$ respectively, the constructions in \cite{SilZeh} handle the case of larger minimum distance but have locality parameter restricted to $r \in  \{ 2,3 \}$.  In \cite{HuaYaaUchSie}, the authors give optimal binary constructions with information and all symbol locality with $d_{min} \in \{3,4\}$. The construction is optimal w.r.t a bound similar to \eqref{CadMaz} derived in \cite{HuaYaaUchSie}.  Constructions achieving the bound on dimension appearing in \cite{WanZhaLin} and the further tightened bound for disjoint repair groups given in \cite{MaGe} for binary codes, appear respectively, in \cite{WanZhaLin,MaGe}. These constructions are for the case $d_{min}=6$. 
		In \cite{HaoXiaChe}, the authors present a characterization of binary LR codes that achieve the Singleton bound \eqref{Singleton}. In \cite{ShaKhaArd}, the authors present constructions of binary codes meeting the Singleton bound. These codes are a subclass of the codes characterized in  \cite{HaoXiaChe} for the case $d_{min} \leq 4$.
		
		\subsubsection{Constructions with Small, Non-Binary Alphabet }

		In \cite{HaoXiaChe1}, the authors characterize ternary LR codes achieving the Singleton bound \eqref{Singleton}.  In \cite{HaoXiaChe,ShaKhaArd, HaoXia}, the authors provide constructions for codes over a field of size $O(r)$ that achieve the Singleton bound in \eqref{Singleton} for $d_{min} \leq 5$.   Some codes from algebraic geometry achieving the Singleton bound \eqref{Singleton} for restricted parameter sets are presented in \cite{LiMaXin1}.
		
		\subsubsection{Construction of Cyclic LR Codes}
		
		Cyclic LR codes can be constructed by carefully selecting the generator polynomial $g(x)$ of the cyclic code.  We illustrate a key idea behind the construction of a cyclic LR code by means of an example.  
		\begin{figure}[h!]
			\centering 		\includegraphics[width=4in]{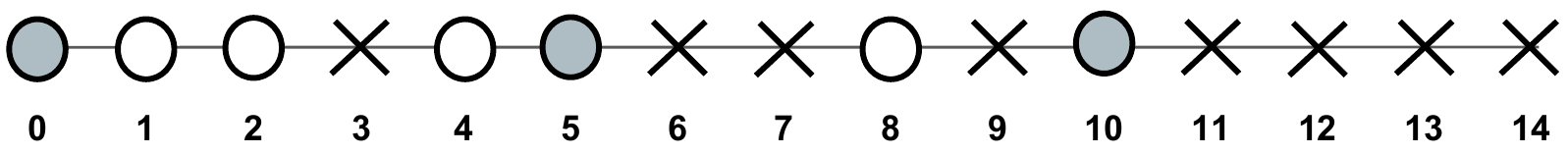}
			\caption{Zeros of the generator polynomial $g(x)=\frac{g_1(x)g_2(x)}{(x+1)}$ of the cyclic code in Example~\ref{eg:zero_train} are identified by circles. The unshaded circles along with the shaded circle corresponding to $\alpha^0=1$ indicate the zeros $\{1,\alpha,\alpha^2,\alpha^4,\alpha^8\}$ of $g_1(x)$ selected to impart the code with $d_{\min} \geq 4$.  The shaded circles indicate the periodic train of zeros $\{1,\alpha^5,\alpha^{10}\}$ introduced to cause the code to be locally recoverable with parameter $(r+1)=5$. The common element $1$ is helpful both to impart increased minimum distance as well as locality.}
			\label{fig:zero_train} 
		\end{figure}
		%
		
		\begin{eg} \label{eg:zero_train} 
			Let $\alpha$ be a primitive element of $\mathbb{F}_{16}$ satisfying $x^4+x+1=0$.  Let $\calc_1$ be a cyclic $[n=15,k=10]$ code having generator polynomial $g_1(x)=(x+1)(x^4+x+1)$.    
			Since the consecutive powers $\{1,\alpha,\alpha^2\}$ of $\alpha$ are zeros of $g_1(x)$, it follows that $d_{\min}(\calc)\geq 3+1=4$ by the BCH bound.  Suppose we desire to ensure that a code \calc\ having generator polynomial $g(x)$ has $d_{\min} \geq 4$ and in addition, is locally recoverable with parameter $(r+1)=5$, then we do the following.  Set $s=\frac{n}{(r+1)}=3$.   Let $g_2(x) \ = \  \prod_{l=0}^{s-1=2} (x-\alpha^{5l})$ and $g(x) = \text{lcm} \{g_1(x), g_2(x)\} = g_1(x)g_2(x)/(x+1)$.
			It follows that $\sum_{t=0}^{14} c_t \alpha^{5lt} \ = \  0, \ \ l=0,1,2$. Summing over $l$ we obtain:
			\bean
			\sum_{l=0}^2 \sum_{t=0}^{14} c_t \alpha^{5lt} \ =  \ 0  & \Rightarrow & 
			\sum_{t: t= 0 \pmod{3}}c_t \ = \ 0.  
			\eean
			It follows that the symbols $\{c_t \mid t= 0 \pmod{3}\}$ of \calc\ form a local code as they satisfy the constraint of an overall parity-check.  Since the code \calc\ is cyclic the same holds for the code symbols $\{c_{t+\tau} \mid t= 0 \pmod{3}\}$, for $\tau = 0,1,2$.   Thus through this selection of generator polynomial $g(x)$, we have obtained a code that has both locality and $d_{\min} \geq 4$.   The zeros of $g(x)$ are illustrated in Fig.~\ref{fig:zero_train}. The code \calc\ has parameters $[n=15,k=8,d_{min} \geq 4]$ and $r=4$. Note that the price we pay for introduction of locality is a loss in code dimension, equal to the degree of the polynomial $\frac{g_2(x)}{ \text{gcd} \{g_1(x), g_2(x)\}}$.  Thus an efficient code will choose the zeros of $g_1(x),g_2(x)$ for maximum overlap. 
		\end{eg}
		The above idea of constructing cyclic LR code was introduced in \cite{GopCal} and extended in \cite{TambarGopCal,ZehYak,KimNo,LuoXinYua}. In \cite{KriPurKumTamBarISIT17}, the use of locality for reducing the complexity of decoding a cyclic code is explored. The same paper also makes a connection with earlier work \cite{VardyBeery94} that can be interpreted in terms of locality of a cyclic code.
		%
		%
		In \cite{GopCal} a construction of binary cyclic LR codes for $r=2$ an $d_{min} \in \{2,6,10\}$  achieving a bound derived within the same paper for binary codes is provided.   In \cite{KimNo}, the authors give constructions of optimal binary, ternary codes meeting the Singleton bound \eqref{Singleton} for $d_{min} = 4, r \in \{1,3\}$ and $d_{min}=6, r=2$ as well as a construction of a binary code meeting the bound given in \cite{WanZhaLin} for $d_{min}=6,r=2$  based on concatenating cyclic codes.  A discussion on the locality of classical binary cyclic codes as well as of codes derived from them through simple operations such as shortening, can be found in \cite{HuaYakUchSie,HuaYaaUchSie}. The principal idea here is that any cyclic code has locality $d^{\perp}-1$ where $d^{\perp}$ is the minimum distance of the dual code $\calc^{\perp}$. In \cite{ZehYak}, the authors construct optimal cyclic codes under the constraint that the local code is either a Simplex code or else, a Reed-Muller code.   In \cite{TambarGopCal}, the authors provide a construction of cyclic codes with field size $O(n)$ achieving the Singleton bound \eqref{Singleton} and also study the locality of subfield subcodes as well as their duals, the trace codes. In \cite{LuoXinYua}, constructions of cyclic LR codes with $d_{min} \in \{3,4\}$ for any $q$ and flexible $n$ are provided. 
		
		\subsection{Maximal Recoverable Codes} 
		
		An $[n,k]$ MDS code can recover from any pattern of $(n-k)$ erasures.  Maximal Recoverable (MR) codes \cite{Chen} are codes that operate under some pre-specified linearity constraints and which can recover from any pattern of $(n-k)$ erasures that is not precluded by the pre-specified linearity constraints imposed. In the context of locality, these constraints are the ones imposed on the local codes. A different perspective of MR codes based on $k$-core subsets (defined below) is given in \cite{GopHuaSimYek}.
		
		\begin{defn}
			Let $H_0$ be an $(\rho \times n)$ matrix over \fq\ whose row space has $m=q^{\rho}-1$ nonzero vectors with respective support sets $A_i \subseteq [n], i=1,2, \cdots, \ m$.  We view $H_0$ as the matrix that imposes locality constraints.   Let us define a subset $S \subset [n]$ to be a $k$-core with respect to $H_0$ if $|S|=k$ and $| A_i \cap S^c | \geq 1, \text{ for all } i =1,2,\cdots,m$.  Then with respect to $H_0$, an MR code is an $[n,k,H_0,q]$ code ${\cal C}$ possessing a $(k \times n)$ generator matrix $G$ with $k \leq n-\rho$ satisfying the property that $H_0G^T =[0]$ and for any $k$-core $S$,  
			\bea \label{eq:mrc} 
			\text{rank} \left( G \mid_S \right) & = & k. 
			\eea
		\end{defn}
		\begin{note}
			Let $H = \left [ \frac{H_0}{H_1} \right] $ denote the parity-check matrix of the MR code, where $H_1$ represents the additional parity-checks that need to be imposed to satisfy the requirements of an MR code. It could happen that the elements of $H_0$ belong to a small base field $\mathbb{B}$ and over that field it is not possible to find a matrix $H_1$ which will result in an MR code.  It turns out that in such instances, one can always choose the elements of $H_1$ to lie in a suitable extension field \fq\ of $\mathbb{B}$, resulting in an MR code over \fq. 
			
		\end{note}
		\begin{note}
			The condition in \eqref{eq:mrc}  imposed on the $k$-core subsets $S$ is equivalent to the following condition: 
			Let $B \subseteq [n]$ be such that $\mid B^c \cap A_i \mid \geq 1$,   $\forall i=1,2,\cdots,m$. Then $G|_B$ is a generator matrix of an $[n=|B|,k]$ MDS code.  This follows since any $k$ columns of $G |_B$ are required to be linearly independent. 
		\end{note}
		
		\subsubsection{General Construction with Exponential Field Size} 
		The following construction is based on parity check matrix. There is an equivalent construction based on generator matrix which is presented in \cite{GopHuaSimYek}. Saying that $S$ is a $k$-core is equivalent to saying that $S$ is an information set since the $k$ underlying message symbols can be uniquely recovered from the $k$ code symbols $\{c_i \mid i \in S\}$.  From the perspective of the parity check matrix $H$, $S$ is a $k$-core if and only if $\text{rk} \left(H \mid_{S^c}\right) = (n-k)$.  This suggests a construction technique.  Setting $H = \left[ \begin{array}{c} H_0 \\ H_1 \end{array} \right]$ as earlier, we regard the symbols in the $((n-k-\rho) \times n)$ matrix $H_1$ as variables.   We need to select $H_1$ such that any $(n-k) \times (n-k)$ sub-matrix of $H$ corresponding to the complement $S^c$ of a $k$-core, has nonzero determinant.  Let $P(H_1)$ be the polynomial in the symbols of $H_1$ obtained by taking the product of these determinants. Note that the definition of a $k$-core ensures that each of these determinants are non-zero polynomials.  The product polynomial is a polynomial in the entries (variables) of the matrix $H_1$ and each variable appears with degree at most ${n-1 \choose n-k-1}$.  By the Combinatorial Nullstellensatz \cite{Alon_Comb}, it follows that there is a field of size  $q > {n-1 \choose n-k-1}$ such that this product of determinants can be made nonzero.    Thus an MR code always exists of field size $q > {n-1 \choose n-k-1}$. The interest is of course, in explicit constructions of MR codes having low field size $q$.   It is also possible to use linearized polynomials to construct MR codes, but while this results in an explicit construction, the field size is still in general,  of exponential size. 
		
		\subsubsection{Partial MDS Codes} 
		
		In the literature, the focus motivated by practical considerations, is on the following subclass of MR codes, also sometimes termed as Partial MDS (P-MDS) codes \cite{Bla_Haf_Het}.
		\begin{defn} \label{MR_PMDS}
			An $(r,\delta,s)$ MR code or partial MDS code is defined as an $[n=m(r+\delta),k=mr-s]$ code over $\mathbb{F}_q$ in which the $n$ code symbols can be arranged as an array of $(m \times (r+\delta) \ )$ code symbols in such a way that each row in the array forms a $[r+\delta,r,\delta+1]$ MDS code and upon puncturing any $\delta$ code symbols from each row  of the array, the resulting code becomes an $[mr,mr-s]$ MDS code.
		\end{defn}
		A tabular listing of some constructions of partial-MDS codes appears in Table \ref{tab:pmds}. 
		\begin{table}[h!] 
			\caption{Constructions for partial MDS codes.} 
			\label{tab:pmds} 
			\begin{center}
				\begin{tabular}{ | m{2cm} | m{2cm}| m{12cm} |} 
						\hline \hline 
						\multicolumn{1}{|c|}{Reference} &  \centering Parameters of MR Code &  \multicolumn{1}{|c|}{Field Size} \\
						\hline \hline 
						\multicolumn{3}{|c|}{General $r,\delta,s$} \\ \hline 
						\hline
						\cite{Cal_Koy} & $(r,\delta,s)$ &  $(q')^{mr}$ where $q'$ is a prime power $\geq r+\delta$. \\
						\hline
						\cite{Gab_Yak_Bla_Sie} & $(r,\delta,s)$ &  $\geq \max((q')^{\delta+s}m^{s-1},(q')^{s(\delta+s)})$ with $q'$ a prime power $\geq r+\delta$. \\
						\hline \hline 
						\multicolumn{3}{|c|}{$\delta=1$} \\   \hline 
						\hline 
						\cite{Bla_Haf_Het} & $(r,1,s)$ & $O(2^n)$ \\
						\hline 
						\cite{GopHuaJenYek} &  $(r,1,s)$ & $O(m^{\lceil (s-1)(1-\frac{1}{2^r}) \rceil})$ or $\geq n^{\frac{m+s}{2}}$ for $m+s$ even and  $\geq 2n^{\frac{m+s-1}{2}}$ for $m+s$ odd, when $r+1$ and $m$ are powers of $2$. \\
						\hline 
						\cite{Hu_Yek} & $(r,1,s)$ &  $\geq (q')^{\lfloor (1-\frac{1}{m})s \rfloor +m-1}$ ($q'$ is prime power $\geq n $) and for some special case, the field size of their construction is $\geq (q')^{\lfloor (1-\frac{1}{m})s \rfloor +m-2}$. For  $m=2$, $4 | s$, $\geq (q')^{\frac{s}{2}}$ where $q'\geq n$ is a power of $2$.\\
						\hline
						\cite{Gab_Yak_Bla_Sie} & $(r,1,s)$ &  $\geq {2^{\ell}}^{(1+(s-1)\lceil \log_{2^{\ell}}(m) \rceil)}$ where $\ell = \lceil \frac{s+1}{2} \rceil \lceil \log_2(r+\delta) \rceil$. \\
						\hline \hline 
						\multicolumn{3}{|c|}{$s=1$}  \\ \hline 
						\hline 
						\cite{Bla_Haf_Het} & $(r,\delta,1)$ & $O(\max(m,r+\delta))$ \\ 
						\hline 
						\cite{Che_Shu_Yu_Sun} &  $(r,\delta,1)$ & $O(r+\delta)$ \\
						\hline \hline 
						\multicolumn{3}{|c|}{$s=2$}  \\ \hline 
						\hline 
						\cite{Blaum_1} & $(r,1,2)$ & $O(n)$ \\
						\hline 
						\cite{Blaum_3} &  $(r,\delta,2)$ & $\geq m((\delta+1)(r-1)+1)$ $\approx$ $\delta \times n$ \\
						\hline 
						\cite{BalKUmMax} & $(r,\delta,2)$ & $O(n)$ \\
						\hline \hline 
						\multicolumn{3}{|c|}{$s=3$}  \\ \hline 
						\hline 
						\cite{GopHuaJenYek} &  $(r,1,3)$ & $O(k^{\frac{3}{2}})$ \\
						\hline 
						\cite{Gab_Yak_Bla_Sie} &  $(r,\delta,3)$ & if $m < (r+\delta)^3$ then $O((r+\delta)^{3(
							\delta+3)})$ otherwise $O((r+\delta)^{\delta+3} m^{1.5})$ \\
						\hline \hline 
						\multicolumn{3}{|c|}{$s=4$}  \\ \hline  \hline 
						\hline 
						\cite{GopHuaJenYek} &  $(r,1,4)$ & $O(k^{\frac{7}{3}})$ \\
						\hline       \hline                  
					\end{tabular}
			\end{center}
		\end{table}
		In \cite{LalLok}, the authors characterize the weight enumerators and higher support weights of an $(r,1,s)$ MR code.

\section{LR Codes for Multiple Erasures}  \label{sec:multiple_erasures} 

We begin with an overview of the different classes (see Fig.~\ref{fig:LRCClassification}) of LR codes that are capable of recovering from multiple erasures proposed in the literature.   All the codes defined in this section are over the finite field \fq.

\subsection{Various Classes of Multiple-Erasure LR Codes} 


\noindent {\em Sequential-Recovery LR Codes:}  An $(n,k,r,t)$ sequential-recovery LR code (abbreviated as S-LR code) is an $[n,k]$ linear code $\calc$ having the following property: Given a collection of $s \le t$ erased code symbols, there is an ordering $(c_{i_1}, c_{i_2}, \cdots, c_{i_s})$ of these $s$ erased symbols such that for each index $i_j$, there exists a subset $S_j \subseteq [n]$ satisfying (i) $|S_j| \le r$ , (ii) $\ S_j \cap \{i_j, i_{j+1}, \cdots, i_s \} = \phi$,  and 
\begin{eqnarray} \label{eq:locality}
\text{(iii)  \  }  \ c_{i_j} & = & \sum \limits_{\ell \in S_j} u_{\ell} c_{\ell}, \ u_{\ell} \in \fq  .
\end{eqnarray} 
It follows from the definition that an $(n,k,r,t)$ S-LR code can recover from the erasure of $s$ code symbols $c_{i_1}, c_{i_2}, \cdots, c_{i_s}$, for $1 \leq s \leq t$ by using \eqref{eq:locality} to recover the symbols $c_{i_j}, \ j=1,2,\cdots,s$, in succession. 

\begin{figure}[ht!]
	\centering
	\includegraphics[width=4.5in]{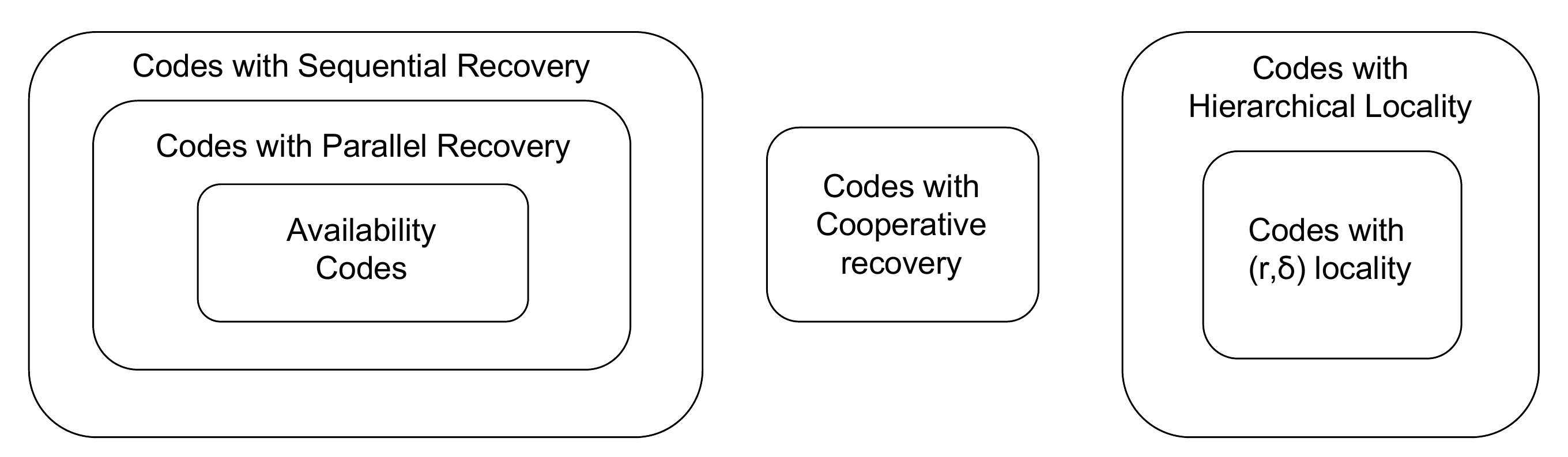}
	\caption{The various code classes corresponding to different approaches to recovery from multiple erasures.}
	\label{fig:LRCClassification}
\end{figure}

\noindent {\em Parallel-Recovery LR Codes:} \ \ If in the definition of the S-LR code, we replace the condition (ii) in \eqref{eq:locality} by the more stringent requirement $S_j \cap \{i_1, i_2, \cdots, i_s \} = \phi,$
then the LR code will be referred to as a \noindent {\em parallel} recovery LR code, abbreviated as P-LR code.  Clearly the class of P-LR codes is a subclass of S-LR codes.  From a practical perspective, P-LR codes are preferred since as the name suggests, the erased symbols can be recovered in parallel. However, this will in general come at the expense of storage overhead.  We note that under parallel recovery, depending upon the specific code, this may require the same helper (i.e., non-erased) code symbol to participate in the repair of more than one erased symbol $c_{i_j}$.

\noindent {\em Availability Codes:} An $(n,k,r,t)$ \noindent {\em availability} LR code, is an LR code having the property that in the event of a single but arbitrary erased symbol $c_i$, there exist $t$ recovery sets $\{R^i_j\}_{j=1}^t$ which are pair-wise disjoint and of size $|R^i_j| \leq r$ with $R^i_j \subseteq [n]-\{i\}$ such that for each $j, 1 \leq j \leq t$, $c_i$ can be expressed in the form:
\bean
c_i = \sum\limits_{\ell \in R^i_j} a_{i\ell} c_{\ell}, \ \ a_{i\ell} \in \fq.
\eean
An $(n,k,r,t)$ availability code is also an $(n,k,r,t)$ P-LR code.  This follows because the presence of at most $t$ erasures implies, that there will be at least one recovery set for each erased code symbol all of whose symbols remain unerased. If the $t$ disjoint recovery sets are available only for code symbols corresponding to an information set, the code is said to be an information-symbol (IS) availability code as opposed to the all-symbol (AS) availability implicit in the previous definition. 

\noindent {\em $(r,\delta)$ Codes:} Recovery from $t$ erasures can also be accomplished by using the codes with $(r,\delta)$ locality introduced in the previous section, Section~\ref{sec:lrc}, if one ensures that the code has $d_{\min} \geq t+1$.  However in this case, repair is local only in those cases where the erasure pattern is such that the number of erasures $e_i$ within each local code satisfies $e_i \leq \delta-1$. 
Thus one may regard $(r,\delta)$ codes as offering probabilistic guarantees of local recovery in the presence of $\leq t$ erasures in exchange for a potential increase in code rate.  Of course, one could always employ an $(r,\delta)$ locality with each local code being an MDS code and $\delta \geq t+1$, but this would result in a significant rate penalty. 


\noindent {\em Cooperative Recovery Codes:} A \noindent {\em cooperative} recovery $(n,k,r,t)$ LR (C-LR) code is an LR code such that if a subset $(c_{i_1}, c_{i_2}, \cdots, c_{i_s})$, $1 \leq s \le t$ of symbols are erased, then there exists a subset $\{c_{j_1}, c_{j_2}, \cdots, c_{j_r}\}$ of $r$ other code symbols 
(i.e., $i_a \ne j_b$ for any $a, b$) such that for all $a \in [s]$,
$c_{i_{a}} = \sum\limits_{b=1}^r \theta_{a,b} c_{j_b}, \  \theta_{a,b} \in \fq.  $
Clearly an $(n,k,r,t)$ C-LR code is also an $(n,k,r,t)$ P-LR code, but the $r$ in the case of a C-LR code will tend to be significantly larger.  One may regard C-LR codes as codes that seek to minimize the number of unerased symbols contacted per erased symbol on average, rather than insist that each code symbol be repaired by contacting $r$ other code symbols.

\subsection{Availability Codes}

\subsubsection{Bounds on Code Rate} 

The following upper bound on the rate of an availability code was given in \cite{TamBarFro}. 
\begin{thm}[\cite{TamBarFro}] If \calc\ is an $(n,k,r,t)$ availability code, then its rate $R$ must satisfy: 
	\bea
	R \ = \ \frac{k}{n} \leq \frac{1}{\prod_{j = 1}^{t}(1+\frac{1}{jr})}. \label{TamoBargRate}
	\eea
\end{thm} 
The parity check matrix of an availability code can be written in the form  $H^T=[H_a^T \ H_b^T]$ where the rows of $H_a$ are the distinct parity checks associated with the recovery sets $R^i_j$, $\forall i \in [n], j \in [t]$ and where the matrix $H_b$ contains all the remaining parity checks.  Clearly the Hamming  weight of each row of $H_a$ is $\leq (r+1)$ and the column weight $\geq t$. 

\noindent {\em Codes with Strict Availability:} Codes with Strict Availability (SA-LR codes) are simply the subclass of availability codes where each row of $H_a$ has weight equal to $(r+1)$ and each column of $H_a$ has weight equal to $t$.  Thus the number $m$ of rows of $H_a$ must satisfy $m(r+1)=nt$.  Further, if the support sets of the rows in $H_a$ having a non-zero entry in the $i^\text{th}$ column are given respectively by $S^{(i)}_j, j=1,2,\cdots t$, then we must have by the disjointness of the recovery sets, that $S^{(i)}_j \cap S^{(i)}_l =   \{i \} , \forall\ 1 \leq j \neq l \leq t$.       Each code symbol $c_i$ in an SA-LR code is thus protected by a collection of $t$ `orthogonal' parity checks, each of weight $(r+1)$. 
\begin{thm} \cite{BalKum} Let $R=\frac{k}{n}$  be the maximum possible rate of an $(n,k,r,t)$ SA-LR code.  Then $R$ must satisfy the upper bound: 
	\bea
	R &\leq& 1- \left( \frac{t}{r+1}\right) + \left( \frac{t}{r+1} \right) \left( \frac{1}{\prod_{j=1}^{r+1}(1+\frac{1}{j(t-1)})}\right).  \label{Rate_2}
	\eea
\end{thm}
The above bound \eqref{Rate_2}, derived in \cite{BalKum}, is tighter than \eqref{TamoBargRate} as $r$ increases for any fixed $t$.  An upper bound on rate of an $(n,k,r=2,t)$ SA-LR code over $\mathbb{F}_2$ that for large $t$, becomes tighter in comparison with the bounds in either \eqref{TamoBargRate} or \eqref{Rate_2}, is presented in \cite{KadCal}.  Also contained in \cite{KadCal}, is an upper bound on the rate of an $(n,k,r,3)$ SA-LR code over $\mathbb{F}_2$ which is tighter than the bound in either \eqref{Rate_2} or \eqref{TamoBargRate} for $r > 72$ and which makes use of the ``transpose''-based rate equation appearing in \cite{BalKum}.  

\subsubsection{Constructions}

\noindent {\em The Product Code:} Consider the $[(r+1)^t, r^t]$ product code in $t$ dimensions. Clearly this is an $(n=(r+1)^t,k=r^t,r,t)$ availability code, having rate $R=(\frac{r}{r+1})^t$. 

\noindent {\em The Wang et al. Construction:} For any given parameter pair $(r,t)$,  Wang et al. \cite{WanZhaLiu} provide  a construction for an $(n,k,r,t)$ availability code which is defined through its parity-check matrix.  Let $S$ be a set of $m=(r+t)$ elements.  Then in the construction, each row of $H$ corresponds to a distinct subset of $S$ of cardinality  $(t-1)$ and each column, to a distinct subset of $S$ of cardinality $t$.  We set $h_{ij}=1$ if the $i$-th $(t-1)$-subset belongs to the $j$-th $t$-subset and zero otherwise.  Thus $H$ is of size ${m \choose t-1} \times {m \choose t} $.  It is easy to verify that each row of $H$ has constant row weight $(r+1)$ and each column of $H$ has constant weight $t$.  It turns out that the rank of $H$ is given by ${m-1 \choose t-1}$ and that $H$ defines an $(n,k,r,t)$ availability code, having parameters: $n= {m \choose t}, k= {m \choose t} - {m-1 \choose t-1}$ and rate $R= \frac{r}{r+t}$. 
Thus this code provides improved rate in comparison with the product code.  Since $ {r+t \choose t} \ <  \ (r+1)^t$,  the code has smaller block length as well. 

\noindent {\em Direct-Sum Construction:} It is shown in \cite{KadCal} that the direct sum of $m$ copies of the $[7,3]$ Simplex code yields an SA-LR code with parameters $(7m,3m,2,3)$ having maximum possible rate for $n=7m,r=2,t=3,q=2$.  

\subsubsection{Bounds on Minimum Distance} \label{sec:available_dmin} 

Let $d_{\text{min}}(n,k,r,t) $ be the maximum possible minimum distance of an $(n,k,r,t)$ availability code. In \cite{WanZha}, the following bound on the minimum distance of an information symbol availability code (and hence applicable to the case of all-symbol availability codes as well) was presented:
\bea
d_{\text{min}}(n,k,r,t) \leq n-k+2 - \left \lceil \frac{t(k-1)+1}{t(r-1)+1} \right \rceil .   \label{WangDmin}
\eea
This bound was derived by adopting the approach employed in Gopalan et al. \cite{GopHuaSimYek} to bound the minimum distance of an availability code.   An improved minimum-distance estimate appears in \cite{TamBarFro}:
\bea
d_{\text{min}}(n,k,r,t) \leq n-\sum_{i=0}^{t} \left \lfloor \frac{k-1}{r^i} \right \rfloor. \label{TamoBargDmin}
\eea

\noindent {\em Approach via Minimum Support Weights:} The next bound on minimum distance relies upon an easy-to-compute sequence that  represents upper bounds on the GHW of the dual of an availability code.  Let there be $b$ subsets $\{S_1,\cdots,S_b\}$ of $[n]$, each of size at most $r+1$. We assume that $[n] = \cup_{i=1}^b S_i$. Let $f_i$ be the minimum size of the union of any $i$ out of the $b$ subsets i.e.,  $f_i = \min_{\{T: T \subseteq [b] : |T|=i\}} |\cup_{j \in T} S_j|$. Then $f_i \leq e_i$ \cite{PraLalKum} where the $\{e_i\}_{i=1}^b$ are recursively calculated in the reverse direction as follows:
set $e_b=n$, and for $2 \leq i \leq b$, set 
\bea
e_{i-1} = \min \{e_i, e_i-\left \lceil \frac{2e_i}{i} \right \rceil + r+1\}. \label{Sequence}
\eea
From the definition of $e_j$, it is clear that $e_j$ is an upper bound on the $j$-th minimum support weight or $j$-th GHW of a code containing $b$ linearly independent codewords with the $i$-th codeword having support $S_i$, $i \in [b]$.  We will refer to the sequence $\{e_i\}$ associated with a given parameter set $(n,r,b)$as the {\em minimum-support-weight (msw) sequence} associated to $(n,r,b)$.   The bound below in \eqref{balkumT}  appeared in \cite{BalKum} and  makes use of the fact that shortening of an $(n,k,r,t)$ availability code results in a second availability code with parameters $(n-\Delta_n,k-\Delta_k,r,t)$ having the same or larger $d_{\min}$.  By applying the bound in \eqref{TamoBargDmin} to the shortened code, one often obtains a bound on the original code (i.e., the parent code before shortening) that is significantly tighter.  To estimate $(\Delta_n,\Delta_k)$, the bound makes use of the msw sequence discussed above. 
\begin{thm} \cite{BalKum}
	Let $b  \ =  \ \left \lceil n(1-\rho(r,t)) \right \rceil$ and $e_i$ be calculated as per \eqref{Sequence}, where
	\bea
	\rho(r,t) \ = \ \begin{cases}
		\frac{r}{r+t}, & \text{ if } t \in \{1,2\} \\ 		
		\frac{r^2}{(r+1)^2}, & \text{ if } t=3, \\
		\frac{1}{\prod_{j=1}^{t}(1+\frac{1}{jr})}, & \text{ if } t>3	. 				
	\end{cases} \label{Rate_Bound} 
	\eea
	\bea
	\text{Then, }	d_{\text{min}}(n,k,r,t)  & \leq  &   \min_{\ \ 1 \leq i \leq b,  \ e_i-i <k   \ } \   \left\{ n-k-i+1-\sum_{j=1}^{t} \left \lfloor \frac{k+i-e_i-1}{r^j} \right \rfloor \right\}. \label{balkumT} 
	\eea
\end{thm} 
The calculation of $\rho(r,t)$ for $t=1$ was not explicitly stated in \cite{BalKum} but is well known.   Also contained in  \cite{BalKum} is an improved upper bound on $d_{\text{min}}$ in the case of codes with strict availability.

\subsubsection{Alphabet-Size Dependent Bounds on $d_{\min}$}  
Let $d^q_{min} (n,k,r,t)$ be the maximum possible minimum distance of an $(n,k,r,t)$ availability code over $\mathbb{F}_q$.
In \cite{HuaYaaUchSie}, the authors provide a bound on minimum distance of an $(n,k,r,t)$ IS availability code (the bound thus also applies to AS  availability codes as well) that depends on the size $q$ of the underlying finite field $\mathbb{F}_q$:
\bea
d^q_{min} (n,k,r,t) \leq \min_{\scriptsize \begin{array}{c} 1 \leq x \leq \lceil \frac{k-1}{(r-1)t+1} \rceil, \\ x \in \mathbb{Z}^{+} , y \in [t]^x, A(r,x,y) < k \end{array}} d^q(n-B(r,x,y),k-A(r,x,y)), \label{dminAlp}
\eea
where $A(r,x,y) = \sum_{j=1}^{x} (r-1)y_j + x$, \ $B(r,x,y) = \sum_{j=1}^{x} r y_j + x$ and $d^q(n,k)$ is the maximum possible minimum distance of a classical (i.e., no locality necessary) $[n,k]$ block code over $\mathbb{F}_q$.
There is a similar bound on the dimension of an availability code with parameters $n,r,t,d_{\min}$ over $\mathbb{F}_q$.

The following bound on the minimum distance of an $(n,k,r,t)$ availability code over $\mathbb{F}_q$ that is tighter than the bound in \eqref{dminAlp} appears in \cite{BalKum} and is currently the tightest-known bound on $d^q_{\text{min}}(n,k,r,t)$:
\bea
d^q_{\text{min}}(n,k,r,t) \leq  \min_{i \in S} \ d^q_{\text{min}}(n-e_i,k+i-e_i,r,t), \ \label{Mindist_Bound1}
\eea
where $S=\{ i: e_i-i < k, 1 \leq i \leq b \}$ and $b  \ =  \ \left \lceil n(1-\rho(r,t)) \right \rceil$ and $e_i$ is calculated as per \eqref{Sequence}.  This bound is also based on the shortening approach introduced in \cite{CadMaz}. 


\subsubsection{Asymptotic Bounds on Rate} \label{sec:asymptotic_rate}
Let $R^q(r,t,\Delta) = \limsup_{n \rightarrow \infty} \frac{\log_q(A_q(n,r,t,\lceil \Delta n \rceil))}{n}$, where $A_q(n,r,t,d)$ is the maximum number of codewords in an  availability code with parameters $(n,r,t)$ with minimum distance $d$ over $\mathbb{F}_q$. The only known upper bounds on $\sup_{q} R^q(r,t,\Delta)$ are based on converting the minimum distance bounds appearing in \eqref{WangDmin}, \eqref{TamoBargDmin} and \eqref{balkumT} into asymptotic bounds. There are constructions which provide lower bounds on $ R^q(r,t,\Delta)$.  A lower bound on $\sup_q R^q(r,t,\Delta)$ for any $r\geq t$ is provided in \cite{TamBarFro}.  A lower bound on $\sup_q R^q(r,t,\Delta)$ appears in \cite{KruFro}. For the specific case $t=2$, \cite{TamBarFro} provides lower bounds on $R^q(r,2,\Delta)$: 
\bea
R^q(r,2,\Delta) \geq \frac{r}{r+2} - \min_{0 <s \leq 1} \left( \frac{1}{{r+2 \choose 2}} \log_q(g^{(2)}_q(s)) - \Delta \log_q(s) \right) & & \text{ valid for any $q$}, \\
g^{(2)}_2(s)=\frac{1}{2^{r+2}} \sum_{i=0}^{r+2} {r+2 \choose i} (1+s)^{{r+2 \choose 2} - i(r+2-i)} (1-s)^{i(r+2-i)} & & \text{ valid only for $q=2$}. 
\eea
The reader is referred to \cite{TamBarFro} for an expression for $g^{(2)}_q(s)$ for general $q$ as well as a lower bound on $\sup_q R^q(r,t,\Delta)$ for any $r\geq t$.  A further lower bound on $R^q(r,t,\Delta)$ for the case $t=2$ and based on algebraic geometry codes appears in \cite{BarTamVla}.

\subsection{Codes with Sequential Recovery}

Somewhat surprisingly, the maximum possible rate of an $(n,k,r,t)$ S-LR code has been precisely determined via a tight upper bound and a matching construction. The case $t=2,3$ is respectively settled in \cite{PraLalKum} and \cite{SonYue}, where the authors derive the respective bounds:
\bean
n\geq k+\left\lceil\frac{2k}{r}\right\rceil  \ \ \text{ for  $t=2$, \ \ } n\geq k+\left\lceil\frac{2k+\lceil\frac{k}{r}\rceil}{r}\right\rceil  \text{ for  $t=3$},
\eean
and provide matching constructions in each case. Matching constructions for the $t=2$ case can be derived either from complete graphs or Turan graphs \cite{PraLalKum}.  Interestingly, the construction based on Turan graphs turns out to be optimal with respect to GHW as well.
The general $t \geq 4$ case was settled in \cite{BalKinKum_ISIT,BalKinKum} and is presented below.  				
\begin{thm} \cite{BalKinKum_ISIT,BalKinKum}
	Let $\mathcal{C}$ be an $(n,k,r,t)$ S-LR code over a finite field $\mathbb{F}_q$. Let $r \geq 3$. Then 
	\bea
	\frac{k}{n} & \leq  & \left\{ \begin{array}{rl} \frac{r^{\frac{t}{2}}}{r^{\frac{t}{2}} + 2 \sum_{i=0}^{\frac{t}{2}-1} r^i}, & t \text{ even}, \\ 
		\frac{r^{s}}{r^{s} + 2 \sum_{i=1}^{s-1} r^i + 1}, & \text{for $t$ odd,} \end{array} \right. , \label{eq:seq_rate_bound}
	\eea
	where $s = \frac{t+1}{2}$.  Moreover, there exist \noindent {\em binary} codes (i.e., codes over $\mathbb{F}_q$ with $q=2$) that achieve this bound. 
\end{thm}

The rate bound given in equation \eqref{eq:seq_rate_bound} proves a conjecture given in \cite{SongCaiYue_L} for maximum achievable rate of an $(n,k,r,t)$ S-LRC. 
The proof of the bound \eqref{eq:seq_rate_bound} given in \cite{BalKinKum_ISIT,BalKinKum}, shows that a code achieving the above rate bound must have a parity check matrix  (upto a permutation of rows and columns) with a specific, sparse, staircase structure. An example of this for the case $t=8$ is shown below. 
\bean
H = \left[ \begin{array}{c|c|c|c|c}
		D_0 & A_1 & 0 & 0 & 0\\ \hline
		0 & D_1 & A_2 & 0 & 0\\ \hline
		0 & 0 & D_2 & A_3 & 0\\ \hline
		0 & 0 & 0 & D_3 & C\\ 
\end{array}	\right]	
\eean
Therefore, it can be shown that a binary code achieving the rate bound  \eqref{eq:seq_rate_bound} must be based on a tree-like graph with girth $\geq t+1$ with  degree $r+1$ for most nodes, where each edge of the graph represents a code symbol and each node represents a parity check of the code symbols incident on it. 
Codes achieving the rate bound \eqref{eq:seq_rate_bound} appeared in \cite{BalKinKum,BalKinKum_ISIT,BalKinKum_NCC} and are based on constructing these tree-like graphs with girth $\geq t+1$.

We note that a construction of codes based on $(r+1)-$regular bipartite graphs having girth $t+1$ and achieving rate close to \eqref{eq:seq_rate_bound} was suggested earlier in \cite{RawMazVis}. It was noted that these codes have rate $\geq \frac{r-1}{r+1}$.   It is not hard to show that these codes have rate equal to  $\frac{r-1}{r+1} + \frac{1}{n}$, see \cite{BalKinKum_ISIT}.  For certain $n$, the resultant codes achieve the rate bound in \eqref{eq:seq_rate_bound}.  However these values of $n$ correspond to the existence of Moore graphs of degree $r+1$, and girth = $t+1$ with that number $n$ of edges. For $r \geq 2$, Moore graphs exist only for $t \in \{2,3,4,5,7,11\}$ (see \cite{DynCageSur}).  
%
%
%

\begin{figure}[ht!]
	\centering
	\includegraphics[width=3.8in]{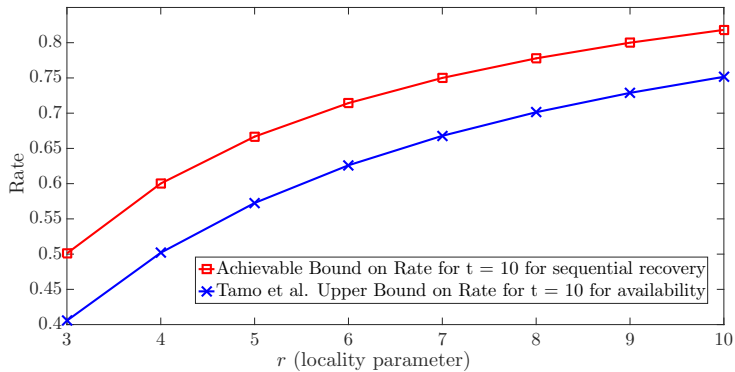}
	\caption{Comparison of rate bounds on codes with sequential recovery \eqref{eq:seq_rate_bound} and codes with availability \eqref{TamoBargRate} for $t=10$.}
	\label{fig:Comparison_Tamo}
\end{figure}

In Fig \ref{fig:Comparison_Tamo}, we compare the tight bound in \eqref{eq:seq_rate_bound} on the rate of an S-LR code with the upper bound in \eqref{TamoBargRate}, due to Tamo et al. on the rate of a code with availability.   The plots suggest that codes with sequential recovery offer a significant rate advantage.  

%

\subsection{$(r,\delta)$ Codes}

The (Singleton) bound on the minimum distance of a code with $(r,\delta)$ locality was presented above in \eqref{Singleton}.  We collect together in this subsection, other results on this class of codes that have appeared in the literature. 

\subsubsection{Constructions and Characterization of Distance Optimal $(r,\delta)$ Codes}

We focus here only on optimal constructions having low field size. A construction achieving Singleton bound with field size of $O(n)$ based on parity splitting appears in \cite{PraKamLalKum} for a restricted set of parameters. A detailed investigation of codes which achieve the Singleton bound on minimum distance of a code with $(r,\delta)$ locality for all symbols appears in \cite{SonDauYueLi} (see in particular, Fig. 2 of \cite{SonDauYueLi} which provides a characterization of the existence of codes achieving the Singleton bound). In \cite{TamBar_LRC}, a construction of codes achieving \eqref{Singleton} with field size $O(n)$ for the case $(r+\delta-1) | n$ is provided.  A construction of cyclic codes with $(r,\delta)$ locality achieving the bound \eqref{Singleton} for $(r+\delta-1) | n$ and field size of $O(n)$ appears in \cite{BinShuJieFan}. 
\subsubsection{$(r,\delta)$ Codes with Small Alphabet Size}

\noindent {\em Upper Bounds on Dimension:} Several alphabet-size dependent bounds on dimension for a code with $(r,\delta)$ AS locality and given minimum distance $d_{\min}$ appear in \cite{AgaBargHuMazTam}.  The bounds take on the form:
\bean
k \leq \left(\left \lceil{\frac{n-d+1}{r+\delta-1}} \right \rceil + 1 \right) \log_q(B(r+\delta-1,\delta)),
\eean
where $B(r+\delta-1,\delta)$ is an upper bound on the number of codewords in a code of block length $(r+\delta-1)$ and minimum distance $\delta$ and is log-convex in the block length. The different bounds are obtained by substituting various bounds for $B(r+\delta-1,\delta)$.  The authors also present bounds for disjoint local codes derived based on association schemes and linear programming which provide the tightest-known bounds in the literature on codes with $(r,\delta)$ locality with disjoint local codes. 

\noindent {\em Binary Codes with $(r,\delta)$ locality:}
In \cite{HaoXiaChe2}, distance-optimal (codes achieving the Singleton bound) binary codes are characterized and the authors of \cite{HaoXiaChe2}, prove that there are only 2 classes of binary, distance-optimal codes for $\delta > 2$. They make use of the fact in their proof that since the code is binary and achieves the Singleton bound on minimum distance, the code after shortening a sufficient number of selected symbols must be an $[\ell,1,\ell]$ MDS code for some $\ell <n$. 

\subsubsection{Achievability Results on Asymptotic Rate}
In \cite{BarTamVla} the following GV-type bound is derived:
\bean
R^q(r,\delta,\Delta) \geq \frac{r}{r+\delta-1} - \min_{0 <s \leq 1} \left(\frac{\log_q(b_{\delta}(s))}{r+\delta-1} - \Delta \log_q(s) \right),
\eean
where $\Delta$ denotes the fractional minimum distance and $\delta$ is the parameter associated with $(r,\delta)$ locality and where 
\bean
b_{\delta}(s)=1+(q-1)\sum_{w=\delta}^{r+\delta-1} {r+\delta-1 \choose w} s^w q^{w-\delta} \sum_{j=0}^{w-\delta} {w-1 \choose j} (-q)^{-j}.
\eean
A second lower bound, based on a construction appearing in \cite{BarTamVla} applies whenever $r+\delta-1=\sqrt{q}$ and improves upon the above GV-type bound in some parameter range: 
$R^q(r,\delta,\Delta) \geq \frac{r}{r+\delta-1} (1-\Delta-\frac{3}{\sqrt{q}+1}). $

\subsection{Codes with Hierarchical Locality}

\begin{figure}[ht!]
	\begin{center}
	\includegraphics[width=3.5in]{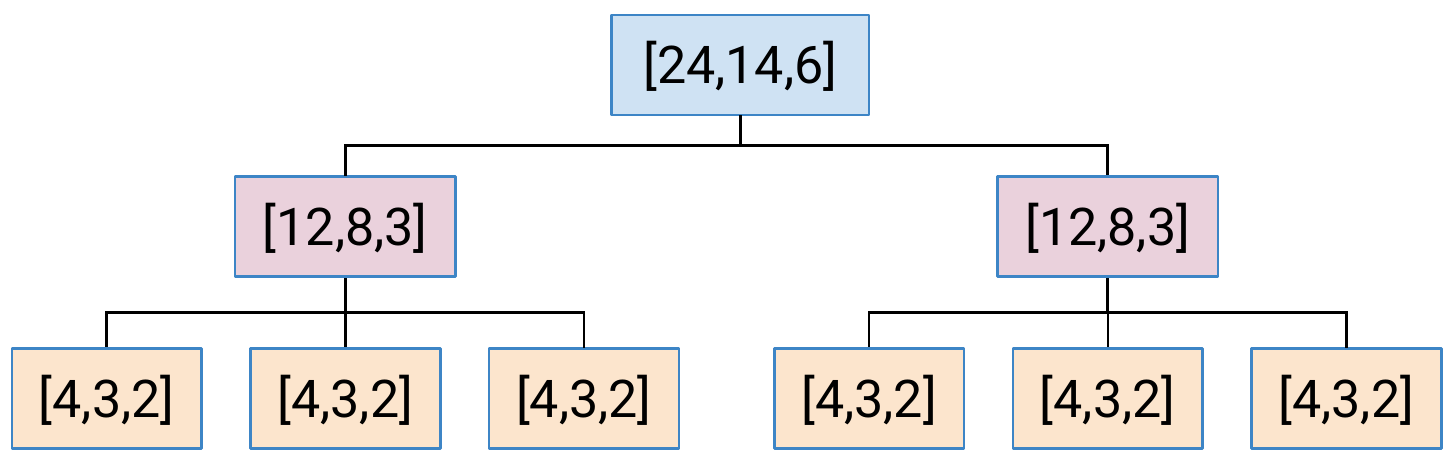}
	\caption{Illustration of a code with hierarchical locality. Each code symbol is protected by a $[4,3,2]$ local code. Each local code is contained in a $[12,8,3]$ middle code. \label{fig:chl}}  
\end{center}
\end{figure}
Codes with hierarchical locality are codes proposed in \cite{SasAgaKum_loc} having multiple tiers of locality.  We restrict the discussion for simplicity here to $2$ tiers.  The motivation here is that in a code with $2$-tier locality, the higher probability single-erasure event can be repaired with the help of a short local code, while the lower-probability, multiple-erasure event can be handled by accessing a larger number of symbols from the next level local code, termed here  as the `middle' code.  A hierarchical topology of local codes as illustrated by the example shown in Fig.~\ref{fig:chl} is proposed in \cite{SasAgaKum_loc} and a bound on the minimum distance derived for the general case. The bound for a two-level hierarchy is presented below.

\bthm \label{thm:bound} Let $\mathcal{C}$ be an $[n,k,d]$-linear code with hierarchical locality with the local and middle codes having  dimensions at most $r_1, r_2$ respectively, and minimum distances at least $\delta_1$, $\delta_2$ respectively. Then
\bea \label{eq:bound}
d \leq n-k+1 - \left( \left\lceil \frac{k}{r_2} \right\rceil -1 \right)(\delta_2 -1) - \left(\left\lceil \frac{k}{r_1} \right\rceil -1 \right)(\delta_1 - \delta_2).
\eea
\ethm
Optimal constructions are provided in \cite{SasAgaKum_loc,BalBar}.  We note that in the context of a practical distributed-storage system, the authors in \cite{DumBie} had previously suggested the topology of hierarchical codes and compared hierarchical codes with Reed-Solomon codes in terms of repair-efficiency using real data. 

\subsection{LR Code with Cooperative Recovery (C-LR code)}
Let $d_{\min}(n,k,r,t)$ be the maximum possible minimum distance of a C-LR code with parameters $(n,k,r,t)$. In \cite{RawMazVis}, the authors introduce the notion of cooperative local repair and provide the following bound on minimum distance for both linear as well as non-linear codes:
\bean
d_{\min}(n,k,r,t) \leq n-k+1-t \left \lfloor \frac{k-t}{r} \right \rfloor .
\eean
They also give a second bound for $r \geq t$.  The paper also contains the following alphabet-size dependent bound on dimension:
\bean
k \leq \min_{\gamma \leq \min(\lfloor \frac{n}{r+t} \rfloor,\lfloor \frac{k-1}{r} \rfloor)} rt+ \log_q(A_q(n-\gamma (r+t),d)),
\eean
where $A_q(n,d)$ is the maximum size of a $q$-ary code of block length $n$ and minimum distance $d$.\\


\begin{open}[\textbf{Codes for Multiple Erasures}]
	
	\begin{enumerate}
		\item[]
		\item For a given $(n,k,r,\delta)$, what is the maximum achievable minimum distance of codes having $(r,\delta)$ locality for a given constraint on field size ? 
		\item For a given $(n,k,r)$, what is the minimum field size over which we can construct a code with locality ($\delta=2$) meeting the Singleton bound ? 
		\item The construction of codes with locality ($\delta=2$) over a field $\mathbb{F}_q$ of size $q=O(1)$ for a larger range of $(d_{min},r)$ (say large $d_{min},r$) which are $d_{min}$ optimal over $\mathbb{F}_q$.
		\item The construction of MR codes with smaller field size for a wide range of parameters. 
		\item What is the maximum achievable rate $\frac{k}{n}$ for a given $(r,t)$ of codes with availability and C-LR codes ?
		\item For a given $(n,k,r,t)$, what is the maximum achievable minimum distance of a S-LR code, a code with availability, or a C-LR code ?
		\item Questions 5 and 6 when restricted to a finite field $\mathbb{F}_q$.
		\item All the above questions on minimum distance can be rephrased as a question on maximum achievable dimension for a given $(n,d_{\min},r,t)$ over a finite field $\mathbb{F}_q$.
	\end{enumerate} 
\end{open}

\section{Locally Regenerating Codes}  \label{sec:lrgc} 

As is clear from the discussion in the preceding sections, while RG codes aim to minimize the repair bandwidth, LR codes focus in keeping the repair degree low.  It is natural to ask if it is possible to construct codes that possess both low repair bandwidth and repair degree.  
The class of Locally Regenerating (LRG) codes introduced independently in \cite{KamPrakLalKumLRGC14} and \cite{RawKoySilbVishLRGC12}, answers this question in the affirmative.  These codes are perhaps best viewed as codes with locality in which the local codes are regenerating codes.   

\subsection{Locality in Vector Codes}

We begin by studying the notion of locality in a vector code, i.e., a code over a vector alphabet.  Let $\mathcal{C}$ be an $[[n,K,d_{\text{min}},\alpha]]$ vector code over the vector alphabet $\mathbb{F}_q^\alpha$ having block length $n$ and minimum Hamming distance $d_{\min}$.  Let $K$ be the dimension of the code viewing the code as a vector space over $\mathbb{F}_q$.  Let $\mathcal{C}_s$ be the scalar code of length $n\alpha$ obtained from $\mathcal{C}$ by replacing each vector symbol by the corresponding $\alpha$ scalar symbols. Let $\mathbf{G}$ be a generator matrix for $\mathcal{C}_s$, where the first $\alpha$ columns correspond to the first vector code symbol of $\mathcal{C}$ and so on. For $1\leq i\leq n$, we use the terminology $i$-th {\it thick column} to denote the set of columns $[(i-1)\alpha+1, i\alpha]$ of $\mathbf{G}$ corresponding to the $i$-th vector code symbol of $\mathcal{C}$. Clearly, the scalar code $\mathcal{C}_s$ has dimension $K$. 

For a subset $S\subseteq [n]$, of indices, let $\mathcal{C}|_S$ denote the vector code obtained by restricting the code $\mathcal{C}$ to the thick columns associated with the indices in $S$.  We similarly define $\mathbf{G}|_S$ to be the restriction of $\mathbf{G}$ to the thick columns associated to $S$.  The definition below is a natural extension of the notion of locality to a code over  vector alphabet. 

\begin{defn}For $i\in[0,n-1]$ and $\delta\geq 2$, the $i$-th vector code symbol of $\mathcal{C}$ is said to have {\it $(r,\delta)$ locality} if there exists a set $S_i\subseteq [n]$ such that $i\in S_i$, $|S_i|\leq r+\delta-1$ and  $d_{\text{min}}(\mathcal{C}|_{S_i})\geq \delta$. The restriction of $\mathcal{C}$ to $S$, i.e., code  $\mathcal{C}|_{S_i}$ will be referred to as the {\it local code associated to $S_i$}.
\end{defn}

\begin{defn} A vector code $\mathcal{C}$ is said to have {\it $(r,\delta)$ information-symbol locality} if there exists $\mathcal{I}\subseteq [n]$ such that $\text{rank}(\mathbf{G}|_{\mathcal{I}})=K$ and the $i$-th vector code symbol of $\mathcal{C}$ has $(r,\delta)$ locality for all $i\in\mathcal{I}$.
	
\end{defn}
$\mathcal{C}$ is said to have {\it$(r,\delta)$ all-symbol locality} if $\mathcal{I}$ can be set to be $[n]$ in the definition above. If for a code having $(r,\delta)$ all-symbol locality,  $S_i=S_j$ or $|S_i\cap S_j|=0$, for all $i\neq j$, then the code is said to have {\it disjoint locality}.

\begin{defn}
	An $[[n,K,d_\text{min},\alpha]]$ vector code $\mathcal{C}$ is said to have the {\it Uniform Rank Accumulation (URA)} property if there exists a sequence $\{a_i\}_{i=1}^{n}$ of non-negative integers satisfying: (i) $a_1=\alpha$ (ii) $\text{rank}(\mathbf{G}|_{\mathcal{I}})=\sum_{j=1}^ia_i$, $\forall \mathcal{I}\subseteq [n]:|\mathcal{I}|=i$. The integer sequence $\{a_i, i \in [n]\}$ is referred to as the {\it rank profile} of $\mathcal{C}$.
\end{defn}

\begin{rem}\label{rem:msr_mbr_rankprofile}
	It is shown in \cite{ShahHBTNonexistence13} that both MSR and MBR codes possess the URA property.  The rank profile in the case of $((n,k,d),(\alpha,\beta),K)$ MSR, MBR codes, are respectively given by: 
	\begin{equation*}
		\underbrace{a_i}_{\text{MSR}} = \begin{cases}
			\alpha &  1\leq i\leq k\\
			0 & (k+1)\leq i\leq n
		\end{cases}, \ \ \ \
		\underbrace{a_i}_{\text{MBR}} = \begin{cases}
			\alpha-(i-1)\beta &  1\leq i\leq k\\
			0 & (k+1)\leq i\leq n
		\end{cases}.
	\end{equation*}
	
\end{rem}

\begin{defn}\label{defn:URA_loc}
	An $[[n,K,d_\text{min},\alpha]]$ vector code $\mathcal{C}$ is said to have {\it URA locality}, if the code  has either information or all-symbol locality and if in addition, local codes are $[[n_\ell,K_\ell,d_\ell,\alpha]]$ vector codes having the URA property with identical rank profiles.
\end{defn}

Consider the vector code $\mathcal{C}$ having  URA locality with parameters as in Definition \ref{defn:URA_loc}.  The rank profile for any given $[[n_\ell,K_\ell,d_\ell,\alpha]]$ local code is denoted by $\{a_i\}_{i=1}^{n_\ell}$. Let  $\{b_i\}_{i=1}^\infty$ be a periodic sequence, where $b_i=a_i$ for $1\leq i\leq n_\ell$ and $b_{n_\ell+j}=b_{j}$ for $j\geq 1$. Define
$P(s)\triangleq\sum_{i=1}^sb_i:\ s\geq1$. For $x\geq 1$, set $P^{\text{(inv)}}(x)$ to be the smallest integer $y$ such that $P(y)\geq x$, i.e.,  $P^{\text{(inv)}}(x)=y$.
\begin{thm}[\cite{KamPrakLalKumLRGC14}] \label{thm:URA_min_dist_bound}
	Let $\mathcal{C}$ be an $[[n,K,d_\text{min},\alpha]]$ code with URA locality, where the local codes have parameter set $[[n_\ell,K_\ell,d_\ell,\alpha]]$. Then, we have  
	$d_{\text{min}}(\mathcal{C}) \leq n -P^{(\text{inv})}(K) + 1.$
\end{thm}
\begin{cor}
	Consider the case of a vector with locality, where the local codes are $((n_\ell,r,d)$, $(\alpha,\beta),K_\ell)$ MSR codes. Using Remark \ref{rem:msr_mbr_rankprofile} and Theorem \ref{thm:URA_min_dist_bound}, it follows that \cite{KamPrakLalKumLRGC14} \cite{RawKoySilbVishLRGC12}:
	\begin{equation*}\label{eq:msr_locality_bound}
		d_{\text{min}}(\mathcal{C})\leq n-\bigg\lceil\frac{K}{\alpha}\bigg\rceil+1-\bigg(\bigg\lceil\frac{K}{\alpha r}\bigg\rceil-1\bigg)(\delta-1).
	\end{equation*}
\end{cor}
In \cite{KamPrakLalKumLRGC14}, the authors give minimum-distance bounds for general vector codes with locality and a tighter bound for the case when the local codes have the URA property. LRG codes with MSR or MBR all-symbol locality, and information-symbol locality that meet the minimum-distance bound, are provided for various parameters. The field-size requirement here is at least $O(n^2)$ for the all-symbol locality code constructions.  In \cite{RawKoySilbVishLRGC12}, the authors present an explicit construction of a vector code with MSR all-symbol locality, that requires a field-size that is exponential in $n$. In \cite{GligKralJenSimHashTagLRGC17}, the authors construct a related family of vector codes with information-symbol locality, where the local codes are vector MDS codes with near-optimal bandwidth and small sub-packetization ($\alpha$) levels.  In \cite{Hollmann14}, \cite{AhmWanLocRepMeetRGC15}, the authors consider vector codes with locality featuring functional repair and achieving a reduction in repair bandwidth by carefully choosing for each failed node, a set of $r\leq k$ helper nodes. In \cite{KriNarKumLRGC18}, the authors provide linear, field-size constructions for LRG codes with all-symbol locality, where the local codes are either MSR or MBR.

\subsection{Codes where local codes are MSR/MBR}

It is possible to construct LRG codes which are minimum-distance optimal where the local codes are MSR or MBR using the Tamo-Barg (T-B) construction of optimal scalar LR codes. 

\begin{eg}[\cite{KamPrakLalKumLRGC14}]
	An LRG code $\mathcal{C}$ having parameters $[[n=15,K=20,d_{\text{min}}=5,\alpha=4]]$ where the local codes are $((n_\ell=5,r=3,d=4),(\alpha=4,\beta=1),K_\ell=9)$ MBR codes, can be constructed as follows. Let $N_\ell\triangleq{n_\ell \choose 2}=10$, $\delta'=N_\ell-K_\ell+1=2$, $\nu=\frac{n}{n_\ell}=3$. Take a minimum-distance optimal $[\nu N_\ell=30,K=20,9]$ scalar T-B code $\mathcal{C}'$ with $(K_\ell=9,\delta'=2)$ all-symbol locality. Note that each local code of $\mathcal{C}'$ is a $[N_\ell=10,K_\ell=9]$ MDS code. The LRG code  with the required parameters is obtained by mapping each such local MDS code to an MBR code, using the polygonal MBR construction. The resultant code (see Fig.~\ref{fig:gen_example_n_eq_8}) is shown to be minimum-distance optimal in \cite{KamPrakLalKumLRGC14}.
	\begin{figure}
		\centering
		\includegraphics[width=5.5in]{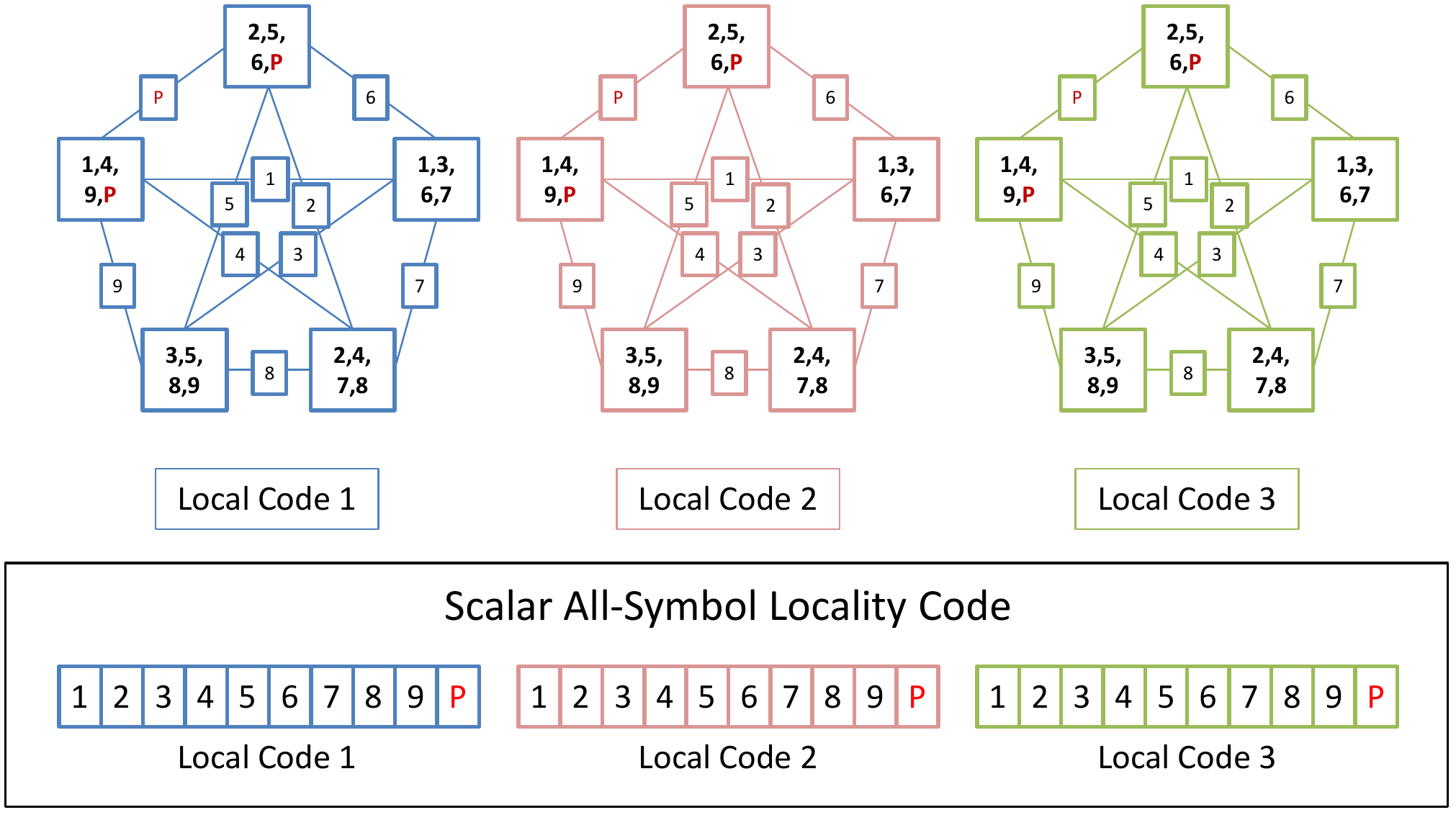}
		\caption{An $[[n=15,K=20,d_{\text{min}}=5,\alpha=4]]$ LRG code $\mathcal{C}$ where local codes are MBR codes. Here the local codes are $((n_\ell=5,r=3,d=4),(\alpha=4,\beta=1),K_\ell=9)$ MBR codes.}
		\label{fig:gen_example_n_eq_8}
	\end{figure}
\end{eg}

\begin{eg}[\cite{KriNarKumLRGC18}]
	From the discussion in Section \ref{sec:TB}, it can be inferred that each local code in a T-B code is an MDS code. Let $(n_\ell-r)\mid r$ and $n_\ell\mid n$. In order to construct a code with MSR local regeneration, we initially stack $\alpha=(n_\ell-r)^{\frac{n_\ell}{n_\ell-r}}$ independent layers of codewords from an $[n,k,d_\text{TB}]$ T-B code with $(r,\delta)$ all-symbol locality. We then perform the pairwise forward transform (introduced in Section \ref{par:clay}) independently, for each local code. This results in an $[[n,K=k\alpha,d_\text{min},\alpha]]$ LRG code $\mathcal{C}$ where local codes are $((n_\ell,r,d),(\alpha,\beta),K_\ell)$ MSR codes, with $d=n_\ell-1$. Let $d_\text{TB}$ denote the (optimal) minimum-distance of the underlying T-B code. The code will be minimum-distance optimal if $d_\text{TB}\leq 2(n_\ell-r+1)$.
\end{eg}

\section{Repairing RS codes} \label{sec:rs_repair} 

The conventional repair of an $[n,k]$ scalar MDS code treats each code symbol as an indivisible unit and leads to a total repair bandwidth of $k$ times the amount of data stored in the failed node, where $k$ is the dimension of the code.  Over the past couple of years, new techniques have surfaced that present a different picture for the repair of scalar MDS codes, particularly for RS codes. These techniques realize that the code symbols (say, over $\mathbb{F}_q$) of a scalar MDS code can be viewed as vectors whose entries are over some subfield, $\mathbb{B}\subseteq\mathbb{F}_q$.  For example, consider the $[16,8]$ RS code obtained by evaluating message polynomials of degree $\leq 7$ over all the elements in $\mathbb{F}_{2^4}$. Under the traditional repair, the repair bandwidth will be $8$ code symbols over $\mathbb{F}_{2^4}$, which is equivalent to $32$ bits. As we will shortly see, it is possible to perform single-node repair in this instance, by downloading just $1$ bit from each of the fifteen surviving nodes. This results in a repair bandwidth of $15$ bits, which is a clear improvement over the $32$ bits downloaded under the conventional scheme. This line of work which vectorizes scalar MDS codes and performs repair operations over a suitable subfield $\mathbb{B}$ for bandwidth gains, began with the pioneering work of Shanmugam et al. \cite{ShanPapDimCaiMDSRepair14} who showed the existence of an efficient repair scheme for systematic node repair, when $k=n-2$, that improved up on the traditional repair bandwidth. In a  subsequent paper, Guruswami and Wootters \cite{GuruWootRSRepair17} consider Generalized Reed-Solomon (GRS) codes and all-node repair. There have been other papers since as well.  

Let $t$ be the degree of the field extension $[\mathbb{F}_q:\mathbb{B}]$.  Clearly, through vector representation over the subfield $\mathbb{B}$ of over $\mathbb{F}$, $t$ can be regarded as the sub-packetization level of the MDS code. Traditional RS codes have code lengths typically on the order of $|\mathbb{F}_q|$ corresponding to a sub-packetization level which is logarithmic in code-length. On the other hand, there are fundamental bounds (see Section \ref{sec:intro_msr}) that require the sub-packetization to be exponential in code length (for a fixed $r$) in order to achieve the cut-set bound. This leads to the natural and interesting question: what is the least possible repair bandwidth that can be achieved in a low-sub-packetization-level setting ?

\subsection{Linear Repair Schemes for Scalar MDS Codes}
In this section, we consider the single-node repair of linear, scalar, MDS codes over $\mathbb{F}_q$, where $q=p^t$ for $p$ a prime power and $t$ a positive integer. Let $\mathbb{B}$ be a subfield of $\mathbb{F}_q$ of size $|\mathbb{B}|=p$. In this setting, by linear repair scheme, we will mean that all repair operations correspond to linear operations over $\mathbb{B}$. For $i\in[n]$, let $b_i$ denote the least possible repair bandwidth (measured by the number of $\mathbb{B}$-symbols downloaded) to repair the $i$-th code symbol. The repair bandwidth $b$ is then defined as: $b\triangleq \max_{i\in[n]}b_i$.  In the discussion below, by dimension we will throughout mean dimension as a vector space over $\mathbb{B}$.
\begin{thm}[\cite{GuruWootRSRepair17}]\label{thm:rs_linear_repair}
	Let $\mathcal{C}$ be a scalar MDS code.  Then a linear repair scheme for $\mathcal{C}$ with repair bandwidth $b$ exists iff for each code coordinate $i\in[n]$, there exists a subset $\mathcal{A}_i$ of $t$ codewords in the dual code $\mathcal{C}^\perp$ such that:
	\bean
	\dim(<a_i, \ \underline{a}\in\mathcal{A}_i>)=t  \ \text{ \ \ and \ \ }  \  \max_{i\in[n]}\bigg(\sum_{j\in[n]\setminus i}\dim(<a_j, \ \underline{a}\in\mathcal{A}_i>)\bigg)\leq b.
	\eean
\end{thm}
It is easy to see the `if' part above.  The {\em trace function} from $\mathbb{F}_q$ to $\mathbb{B}$ is the $\mathbb{B}$-linear map  given by: $T(\gamma)=\sum_{m=0}^{t-1}\gamma^{p^m}$.   Given a basis $\{\rho_m\}_{m=1}^t$ for $\mathbb{F}_q$ over $\mathbb{B}$, it is known~\cite{McWSlo} that there always exists a second basis $\{\gamma_m\}_{m=1}^t$ for $\mathbb{F}_q$ over $\mathbb{B}$, termed the trace-dual basis of $\{\rho_m\}$, such that any $x\in \mathbb{F}_q$ can be expressed in the form $x=\sum_{m=1}^tT(x\rho_m)\gamma_m$.  Let $\mathcal{A}_i$ be as defined in Theorem \ref{thm:rs_linear_repair}. For $\underline{a}\in\mathcal{A}_i \subseteq \mathcal{C}^\perp$ and $\underline{c}\in\mathcal{C}$, we have that $c_ia_i=-\sum_{j=1,j\neq i}^{n}c_j a_j$. Hence 
\begin{equation}\label{eq:trace_equation}
T(c_ia_i)=-\sum_{j=1,j\neq i}^{n}T(c_j a_j), \ \ \text{ for } \underline{a}\in\mathcal{A}_i.
\end{equation}

The definition of $\mathcal{A}_i$ implies that $\dim(<a_i, \ \underline{a}\in\mathcal{A}_i>)=t$. Let $b_{ij}$ denote the dimension of the set $\{a_j\}_{\underline{a}\in\mathcal{A}_i}$ and let $\mathcal{B}_{ij}$ denote a basis for the vector space spanned by  $\{a_j\}_{\underline{a}\in\mathcal{A}_i}$. Using the $\mathbb{B}$-linearity of the trace function, it suffices to compute the $b_{ij}$ trace values $\{T(c_j x): x\in \mathcal{B}_{ij}\}$ which can be used to obtain  $\{T(c_j a_j): \underline{a}\in \mathcal{A}_i\}$. Hence by downloading $\sum_{j=1,j\neq i}^n b_{ij}$ symbols over $\mathbb{B}$, one can compute $\{T(c_ia_i):\underline{a}\in\mathcal{A}_i\}$ using \eqref{eq:trace_equation}. Using the trace-dual basis, $c_i$ can be reconstructed from these $t$ traces.

Next, consider the specific case of an $[n,k]$ GRS code $\mathcal{C}$, whose symbols are $n$ (scaled) evaluations of message polynomials $f(x) \in \mathbb{F}[x]$ of degree $\leq k-1$. Let the evaluation points be denoted by the set $\mathcal{A}=\{\alpha_j\}_{j=1}^n$. As the dual of a  GRS code is a GRS code, codewords in the dual are scaled evaluations of message polynomials of degree $\leq (n-k-1)$.  Thus in the context of a GRS code and ignoring w.o.l.o.g. the scaling coefficients, \eqref{eq:trace_equation} takes on the form:

\begin{equation}\label{eq:trace_equation_GRS}
T(f(\alpha_i)g(\alpha_i))=-\sum_{j=1,j\neq i}^{n}T(f(\alpha_j)g(\alpha_j)), \ \text{  for all } \ \ g(x)\in\mathcal{P}_i,
\end{equation}
where $f(x)$ and $g(x)$ are polynomials having degrees at most $k-1$ and $n-k-1$, respectively, $\mathcal{P}_i$ is the set of $t$ message polynomials having degree at most $n-k-1$ corresponding to the $t$ dual codewords in $\mathcal{A}_i$. 

\subsection{Guruswami-Wootters GRS Repair Scheme}
Let $k\leq n-p^{t-1}$ for a GRS code. Then it is possible repair each code-symbol (say, $i$-th) by downloading just one symbol over $\mathbb{B}$ each from the remaining $(n-1)$ nodes. The scheme is as follows. Consider the set of $t$ polynomials $\mathcal{P}_i=\{g_{i,m}(x)\}_{m=1}^t$ and a basis $\{\rho_m\}_{m=1}^t$ , where:
\begin{equation*}
	g_{i,m}(x)=\frac{T\big(\rho_m(x-\alpha_i)\big)}{(x-\alpha_i)}=\sum_{s=0}^{t-1}\rho_m^{p^s}(x-\alpha_i)^{p^s-1}.
\end{equation*}
Each polynomial $g_{i,m}(x)$ has degree $p^{t-1}-1\leq n-k-1$. Hence the evaluations of this polynomial represent a codeword in $\mathcal{C}^\perp$. Note that $\{g_{i,m}(\alpha_i)=\rho_m,m\in[t]\}$ forms a basis for $\mathbb{F}$ over $\mathbb{B}$, i.e., $\text{dim}_\mathbb{F}\langle\{g_{i,m}(\alpha_i),m\in[t]\}\rangle=t$.
Also, $\text{dim}_\mathbb{F}\langle\{g_{i.m}(\alpha_j),m \in[t]\}\rangle=1 ~\forall j\in [n]\setminus\{i\}$.
\begin{thm}\label{thm:mds_bandwidth_lb}
	Let $\mathcal{C}$ be an  $[n,k]$ MDS code over $\mathbb{F}_q$. For any linear repair scheme for $\mathcal{C}$ over $\mathbb{B}$, the repair bandwidth, $b$ (counted according to the number of symbols from $\mathbb{B}$) satisfies the following:
	\begin{equation*}
		b\geq (n-1)\log_{|B|}\Bigg(\frac{n-1}{n-k+\frac{k-1}{|F|}}\Bigg).
	\end{equation*}
\end{thm}
By Theorem \ref{thm:mds_bandwidth_lb}, the repair scheme discussed above is optimal when $\mathcal{A}=\mathbb{F}_q$ and $n-k=p^{t-1}$.
\subsection{Other Related Work}
In \cite{DauMilenRSRepair}, the authors improve the Guruswami-Wootters approach to a larger class of parameters. In \cite{YeBargAsymptOptimalRS16}, the authors provide a family of RS codes that has asymptotically optimal repair bandwidth with respect to the cut-set bound. This result is further developed in \cite{ChowVardy17} to reduce the sub-packetization levels. In \cite{TamYeBargRSRepair17}, the authors present RS codes that meet the MSR point for all parameters: $k<d<n-1$. Bandwidth-efficient recovery from multiple erasures in RS codes is addressed in \cite{DauDuursmaKiahMilenRSMulti16} and is further extended to include general scalar MDS codes in \cite{BartWootMDSMulti17}. In \cite{YeBargRSRepairMultiErasure17}, the authors present codes that universally achieve the optimal bandwidth points for all parameters $h\leq n-k$ and $k\leq d\leq n-h$ simultaneously.

\begin{table}[ht!]
	\caption{A summary of schemes appearing in the literature for the repair of RS codes.}
	\centering
	\begin{tabular}{|M{1.8cm}|M{4cm}|M{2cm}|M{2.3cm}|M{5cm}|}
		\hline \hline
		Reference &  Bandwidth & Sub-packetization & Cut-set bound achievability & Remarks \\ \hline
		\cite{GuruWootRSRepair17} & $n-1$ & $\log_{p}n$ & No & Single node repair; $(n-k)\geq p^{t-1}$ \\ \hline
		\cite{DauMilenRSRepair} & $(n-1)t(1-\log_n(n-k))$ & $\log_{p}n$ & No & Single node repair; $(n-k)\geq p^{\ell};\ell\in[t-1]$\\
		\hline
		\cite{YeBargAsymptOptimalRS16} & $<\frac{t(n+1)}{n-k}$ & $(n-k)^n$ &  asymptotically & Single node repair \\
		\hline
		\cite{ChowVardy17} & $<\frac{t(n-1+3(n-k))}{n-k}; (n-k)=s^m$ & $s^{m+n-1}$ &  asymptotically & Single node repair \\	\hline
		\cite{TamYeBargRSRepair17} & $\frac{td}{d-k+1}$ & $\approxeq n^n$ & Yes & Codes exist for any given  $d\in[k,n-1]$\\
		\hline
		\cite{DauDuursmaKiahMilenRSMulti16} & $2(n-1)$ $(2 \text{-erasures})$; $3(n-1) (3\text{-erasures})$ &$\log_pn$ & No & Distributed repair \\
		\hline
		\cite{DauDuursmaKiahMilenRSMulti16} & $2(n-2)(2\text{-erasures})$; $3(n-3) (3\text{-erasures})$ &$\log_pn$ & No & Centralized repair \\
		\hline
		\cite{BartWootMDSMulti17} & $h (n-h)-(p-1)(h-1)(h\text{-erasures})$ & $\log_pn$ & No & Centralized repair $h\leq \sqrt{\log n}$ \\
		\hline
		\cite{YeBargRSRepairMultiErasure17} & $\frac{tdh}{d-k+h}$ &$\approxeq n^n$ & Yes; bound in \cite{CadJafMalRamSuh} & Code works simultaneously for any given no. of failures, $h:h\in[1,n-k]$ and any $d: d\in[k,n-h]$\\
		\hline
		\hline \hline
		\end{tabular}	
\end{table}





%
%
%
%

\section{An Information Capacity Approach} \label{sec:liquid} 

{\em Capacity Bounds:} In \cite{Luby}, a generic distributed storage system model is introduced and fundamental limits presented. The notion of information capacity of a distributed system is introduced. Let $m$ denote the source data size in bits.  Consider a distributed storage system with $N$ nodes, each storing $s$ bits of data. If $\Delta$ denotes the average time between node failures, the erasure rate $\epsilon$ can be defined as $\epsilon=\frac{s}{\Delta}$.  When a node failure takes place, a repairer carries out node repair in a manner which ensures that the source data can be recovered from the data in the surviving nodes at any point of time.  The Mean Time to Data Loss (MTTDL) is the average amount of time over which the source data can be recovered. Let $\gamma$ denote the repair rate, which is the rate at which the repairer reads and writes data. Let $\sigma=\frac{\gamma}{\epsilon}$ denote the repair rate to erasure rate ratio. The information capacity of a distributed storage system is then defined as the largest amount of source data $m$ for which a large MTTDL is possible. In \cite{Luby} it is shown that the information capacity approaches $(1-\frac{1}{2\sigma})Ns$ bits as $\sigma$ and $N$ grow. 

{\em Liquid Storage:} In \cite{LubyPadRic} the idea of liquid cloud storage was proposed in which codes of large block length (for example, authors use a code of block length 3010 in one of their simulations) are used to spread data stored pertaining to every object over a large number of nodes. Liquid storage employs a lazy repair strategy where the repair runs slowly in the background. The authors present simulation results that shows that liquid storage gives better MTTDL performance in comparison with systems based on small block length codes. The performance of liquid storage systems is shown to approach the fundamental limits proved in \cite{Luby}. 


\section{Codes in Practice\label{sec:applns}} 

Distributed systems such as Hadoop, Google File System and Windows Azure have evolved to include support for erasure codes within their systems, in order to enjoy the benefits of improved storage efficiency in comparison with triple replication.  However, the use of traditional erasure codes results in additional repair traffic resulting in larger repair times. This led to several theoretical code constructions for efficient node repair and these were discussed in the preceding sections of this article.  Among the biggest success stories is undoubtedly the adoption of  LR codes in the Windows Azure production cluster. 

{\em LR codes:} In \cite{HuaSimXu_Azure}, the authors compare performance-evaluation results of an $(n=16,k=12,r=6)$ LR code with that of an $[n=16,k=12]$  RS code in the Azure production cluster and demonstrate the repair savings offered by the LR code. Subsequently, the authors implemented an $(n=18,k=14,r=7)$ LR code in Windows Azure Storage and showed that this code has repair degree comparable to that of an $[9,6]$ RS code, but has storage overhead $1.29$ versus $1.5$ in the case of the RS code. This code has reportedly resulted in the savings of millions of dollars for Microsoft \cite{microsoft}. The authors of \cite{SathiaAstPap_Xorbas} implemented HDFS-Xorbas which uses LR codes in place of RS codes in HDFS-RAID. Xorbas LR code is build on top of an RS code by adding extra local XOR parties. The experimental evaluation of Xorbas was carried out in Amazon EC2 and a cluster in Facebook, in which the repair performance of $(n=16,k=10,r=5)$ LR code  was compared against a $[14,10]$ RS code. A second distributed storage system that has an LR code plug-in \cite{lrc_plugin} is Ceph.

{\em MDS codes with bandwidth savings:} The Hitchhiker erasure coded system presented in \cite{RasShaGu_Hitchhiker} is a practical implementation of the piggybacking framework introduced in \cite{RasShaRam_Piggyback}. 
The authors implemented the Hitchhiker in HDFS and evaluated its performance on a data-warehouse cluster at Facebook. The Hitchhiker has now been incorporated into Apache Hadoop. In \cite{KraGliJen_HashTag}, the HDFS implementation of a class of MDS array codes called HashTag codes is discussed. The theoretical framework of HashTag codes was presented in \cite{KraGliOVe}. These codes allow low sub-packetization levels at the expense of increased repair bandwidth and are designed to efficiently repair systematic nodes.

{\em Regenerating Codes:} The NCCloud \cite{HuChenLee_NCcloud} is one of the earliest works that dealt with the practical performance evaluation of regenerating codes. The NCCloud storage system is build on top of a 2-parity functional MSR code. In \cite{KriPraLal_Pentagon}, the performance of the pentagon code (which is a repair-by-transfer MBR code) and a heptagon-local code (which is a locally regenerating code) in a Hadoop setting are studied. These two codes possess inherent double replication of code symbols, have storage overhead slightly greater than $2$ and their performance is compared against double and triple replication. In \cite{RasNakWan_PMRBT}, the authors present an optimal-access version of the product-matrix MSR code, which they refer to as the PM-RBT code. The results of an experimental evaluation of a rate $\frac{1}{2}$ PM-RBT code on Amazon EC2 instances is reported. In \cite{LiLi_Beehive}, the authors introduced erasure codes termed Beehive that are built on top of MSR codes. These codes repair multiple failures simultaneously and are implemented using the Product-Matrix MSR in C++ using the Intel storage acceleration library (ISAL). In \cite{JuaBlaMat_Butterfly}, the authors present the evaluation of a high-rate MSR code known as the Butterfly code in both Ceph and HDFS. This code is a simplified version of the MSR codes presented in \cite{GadMatBla} corresponding to the presence of two parity nodes. This code possesses the optimal-access property except in the case of the repair of a single parity node, and has sub-packetization level $\alpha=2^{k-1}$.   More recently in \cite{VajRamPur_Clay}, the authors present Clay code that corresponds to the codes in \cite{YeBar_2,LiTangTian,SasVajKum_arxiv}. The Clay code is  implemented over Ceph based on the coupled-layer perspective in \cite{SasVajKum_arxiv} and is evaluated over an Amazon AWS cluster. The Clay code is 
simultaneously optimal in terms of storage overhead, repair bandwidth, optimal access and sub-packetization level. As a part of this work, vector code support has been added to Ceph and the Clay code is under consideration to become a part of Ceph's master code-base.  


\bibliography{master}

\end{document}